\documentclass[sigconf]{acmart}
\usepackage[vlined,boxed]{algorithm2e}
\usepackage{amsmath}
\usepackage{amsfonts}
\usepackage{epsfig}
\usepackage{graphics}
\usepackage{amssymb}
\usepackage{color}
\usepackage{comment}
\usepackage{paralist}
\usepackage{mathtools}
\usepackage{multirow}
\usepackage{subfigure}
\usepackage{eepic}
\usepackage{textcomp}
\usepackage{balance}
\usepackage[inline]{enumitem}
\usepackage{ulem}
\usepackage{todonotes} 
\usepackage{misc}
\usepackage{thmtools}
\usepackage{thm-restate}

\begin{document}

\newcommand{\mi}[1]{\mathit{#1}}
\newcommand{\ins}[1]{\mathbf{#1}}
\newcommand{\adom}[1]{\mathsf{dom}(#1)}
\renewcommand{\paragraph}[1]{\textbf{#1}}
\newcommand{\ra}{\rightarrow}
\newcommand{\fr}[1]{\mathsf{fr}(#1)}
\newcommand{\dep}{\Sigma}
\newcommand{\sch}[1]{\mathsf{sch}(#1)}
\newcommand{\ar}[1]{\mathsf{ar}(#1)}
\newcommand{\body}[1]{\mathsf{body}(#1)}
\newcommand{\head}[1]{\mathsf{head}(#1)}
\newcommand{\guard}[1]{\mathsf{guard}(#1)}
\newcommand{\class}[1]{\mathbb{#1}}
\newcommand{\size}[1]{||#1||}
\newcommand{\tw}[1]{\mathsf{TW}_{#1}}
\newcommand{\var}[1]{\mathsf{var}(#1)}
\newcommand{\omq}[1]{\mathsf{omq}(#1)}
\newcommand{\base}[1]{\mathsf{base}(#1)}
\newcommand{\chase}[2]{\mathsf{chase}(#1,#2)}
\newcommand{\gchase}[2]{\mathsf{chase}_{\downarrow}(#1,#2)}
\newcommand{\lchase}[4]{\mathsf{chase}^{#1}_{#2}\left(#3,#4\right)}
\newcommand{\fmods}[2]{\mathsf{fmods}(#1,#2)}
\newcommand{\level}[2]{\mathsf{level}_{#1}(#2)}
\newcommand{\rew}[2]{\mathsf{rew}(#1,#2)}
\newcommand{\type}{\mathsf{type}}
\newcommand{\atoms}[1]{\mathsf{atoms}(#1)}
\newcommand{\complete}[2]{\mathsf{complete}(#1,#2)}
\newcommand{\N}[1]{\mathsf{N}(#1)}
\newcommand{\rt}[1]{\mathsf{root}(#1)}

\newcommand{\PTime}{\text{\rm \textsc{PTime}}}
\newcommand{\NP}{\text{\rm \textsc{NP}}}
\newcommand{\FPT}{\text{\rm \textsc{FPT}}}
\newcommand{\W}{\text{\rm \textsc{W}[1]}}
\newcommand{\EXP}{\text{\rm \textsc{ExpTime}}}
\newcommand{\TWOEXP}{\text{\rm \textsc{2ExpTime}}}
\newcommand{\THREEEXP}{\text{\rm \textsc{3ExpTime}}}
\newcommand{\FOUREXP}{\text{\rm \textsc{4ExpTime}}}
\newcommand{\logspace}{\text{\rm \textsc{LogSpace}}}


\def\qed{\hfill{\qedboxempty}      
  \ifdim\lastskip<\medskipamount \removelastskip\penalty55\medskip\fi}

\def\qedboxempty{\vbox{\hrule\hbox{\vrule\kern3pt
                 \vbox{\kern3pt\kern3pt}\kern3pt\vrule}\hrule}}

\def\qedfull{\hfill{\qedboxfull}   
  \ifdim\lastskip<\medskipamount \removelastskip\penalty55\medskip\fi}

\def\qedboxfull{\vrule height 4pt width 4pt depth 0pt}

\newcommand{\markfull}{\qedboxfull}
\newcommand{\markempty}{\qed} 

\newcommand{\OMIT}[1]{}

\newtheorem{claim}[theorem]{Claim}
\newtheorem{fact}[theorem]{Fact}
\newtheorem{observation}{Observation}
\newtheorem{remark}{Remark}
\newtheorem{apptheorem}{Theorem}[section]
\newtheorem{appcorollary}[apptheorem]{Corollary}
\newtheorem{appproposition}[apptheorem]{Proposition}
\newtheorem{applemma}[apptheorem]{Lemma}
\newtheorem{appclaim}[apptheorem]{Claim}
\newtheorem{appfact}[apptheorem]{Fact}

\fancyhead{}



\title{The Limits of Efficiency for Open- and Closed-World Query Evaluation Under Guarded TGDs}



\author{Pablo Barcel\'{o}}
\affiliation{%
	\institution{IMC, PUC Chile \& IMFD Chile}
}
\email{pbarcelo@ing.puc.cl}

\author{Victor Dalmau}
\affiliation{%
	\institution{Universitat Pompeu Fabra}
}
\email{victor.dalmau@upf.edu}

\author{Cristina Feier}
\affiliation{%
	\institution{University of Bremen}
}
\email{feier@uni-bremen.de}

\author{Carsten Lutz}
\affiliation{%
	\institution{University of Bremen}
}
\email{clu@uni-bremen.de}

\author{Andeas Pieris}
\affiliation{%
	\city{University of Edinburgh}
}
\email{apieris@inf.ed.ac.uk}

\begin{abstract}
	Ontology-mediated querying and querying in the presence of constraints are two key database problems where tuple-generating dependencies (TGDs) play a central role. In ontology-mediated querying, TGDs can formalize the ontology and thus derive additional facts from the given data, while in querying in the presence of constraints, they restrict the set of admissible databases.
	%
	%
	%
	%
	%
	%
	In this work, we study the limits of efficient query evaluation in the context of the above two problems, focussing on guarded and frontier-guarded TGDs and on UCQs as the actual queries.
	We show that a class of ontology-mediated queries (OMQs) based on guarded TGDs can be evaluated in FPT iff the OMQs in the class are equivalent to OMQs in which the actual query has bounded treewidth, up to some reasonable assumptions. 
	For querying in the presence of constraints, we consider
        classes of constraint-query specifications (CQSs) that bundle
        a set of constraints with an actual query. We show a dichotomy
        result for CQSs based on guarded TGDs that parallels the one
        for OMQs except that, additionally, FPT coincides with PTime
        combined complexity. The proof is based on a novel connection between OMQ and CQS evaluation.
	Using a direct proof, we also show a similar dichotomy result, again up to some reasonable assumptions, for CQSs based on frontier-guarded TGDs with a bounded number of atoms in TGD heads. 
	Our results on CQSs can be viewed as extensions of Grohe's well-known characterization of the tractable classes of CQs (without constraints). Like Grohe's characterization, all the above results assume that the arity of relation symbols is bounded by a constant. 
	We also study the associated meta problems, i.e., whether a given OMQ or CQS is equivalent to one in which the actual query has bounded treewidth.
\end{abstract}

\maketitle

\section{Introduction}\label{sec:introduction}
%


Tuple-generating dependencies (TGDs) are a prominent
rule-based formalism at the core of several areas that are of central importance to databases and artificial intelligence.
They are first-order implications of the form
$\forall \bar x \forall \bar y \, \big(\phi(\bar x,\bar y)
\rightarrow \exists \bar z \, \psi(\bar x,\bar z)\big)$, where $\phi$ and~$\psi$ are conjunctions of relational atoms, and they essentially state that some tuples (facts) in a relational instance imply the presence of some other tuples in that instance (hence the name ``tuple-generating''). 
In query evaluation over relational databases, TGDs have two facets: (1) they can be used as ontology axioms that allow us to derive additional facts from the given data, which supports more complete query answers, and (2) they can be used as integrity constraints that restrict the set of admissible databases, which paves the way to constraint-aware query optimization.
In fact, 
TGDs were originally introduced as a unifying framework for the large range of relational constraints
introduced in the 1970s and the 1980s~\cite{AbHV95}.

\medskip

\noindent \paragraph{The Two Facets of TGDs.}  When a set of TGDs is used as an ontology, they are often bundled with an actual database query, typically a union of conjunctive queries (UCQ),
to form a composite ontology-mediated query (OMQ)~\cite{BCLW14,BiOr15}.
As discussed above, adding the ontology serves the purpose of delivering more complete answers to queries. It also enriches the vocabulary available for querying as it may introduce new relation symbols that are not part of the data signature.
An OMQ language is a pair $(\class{C},\class{Q})$, with $\class{C}$ being a class of TGDs and $\class{Q}$ a query language, which collects all the OMQs in which the ontology is formulated in $\class{C}$ and the actual query comes from $\class{Q}$~\cite{BCLW14}.
The main problem of concern for an OMQ language $(\class{C},\class{Q})$
is OMQ evaluation 
under an open-world semantics, i.e., we are looking for answers to the actual query from $\class{Q}$ that are logically entailed by 
the input database and 
the ontology from $\class{C}$.

On the other hand, when a set of TGDs is used as relational integrity constraints, the TGDs together with the actual query can be bundled together into what we call a constraint-query specification (CQS). A class of CQSs is a pair $(\class{C},\class{Q})$, where, as in OMQs, $\class{C}$ is a class of TGDs and $\class{Q}$ a query language. In this setting, however, the problem of interest is CQS evaluation 
under a closed-world semantics, i.e., we are looking for answers to the query from $\class{Q}$ over the input database, which is promised to comply with the set of constraints from $\class{C}$. In other words, we directly evaluate the query over the database, which we know to satisfy the set of constraints.

%




\medskip

\noindent \paragraph{TGDs and Guardedness.} 
%
It is well-known that OMQ evaluation in $(\class{TGD},\class{UCQ})$, where $\class{TGD}$ is the class of arbitrary TGDs, and $\class{UCQ}$ the class of union of conjunctive queries, is an undecidable problem; see, e.g.,~\cite{CaGK13}. This has led to an intensive research activity for identifying restrictions on TGDs that ensure the decidability of OMQ evaluation; see, e.g.,~\cite{BLMS11,CaGK13,CaGP12,LMTV19} -- the list is by no means exhaustive.
One of the most robust syntactic paradigms that emerged from this extensive effort, which is central to our work, is guardedness.
A TGD is guarded if the left-hand of the implication, called the body, has an atom that contains (or guards) all the universally quantified variables~\cite{CaGK13}; let $\class{G}$ be the class of guarded TGDs.
A natural generalization is to guard only the variables that appear also in the right-hand side, called the head, which leads to the class of frontier-guarded TGDs $\class{FG}$~\cite{BLMS11}.
The complexity of OMQ evaluation in $(\class{C},\class{UCQ})$, where $\class{C} \in \{\class{G},\class{FG}\}$, is by now well-understood. In both cases, it is \TWOEXP-complete~\cite{BMRT11,CaGK13},
%
and it remains hard even in the case of bounded-arity schemas, i.e., \TWOEXP-complete in $(\class{FG},\class{UCQ})$~\cite{BMRT11},
and \EXP-complete in $(\class{G},\class{UCQ})$~\cite{CaGK13}.
%

Although guarded and frontier-guarded TGDs have been proposed as ontology languages, they can also naturally serve as classes of integrity constraints~\cite{BaGP16}. Note, for example, that the 
important class of referential integrity constraints (or inclusion dependencies) is a very special case of guarded TGDs.
%
Now, the complexity of CQS evaluation in $(\class{C},\class{UCQ})$, where $\class{C} \in \{\class{G},\class{FG}\}$, coincides with that of query evaluation for $\class{UCQ}$: it is \NP-complete, even for schemas of bounded arity~\cite{ChMe77}.
%
This holds since CQS evaluation in $(\class{C},\class{UCQ})$ is, in fact, a refinement of the query evaluation problem for $\class{UCQ}$, with the additional promise that the input database satisfies a given (but potentially empty) set of integrity constrains coming from $\class{C}$.
\medskip

\noindent \paragraph{The Limits of Efficiency.} 
Thus both OMQ evaluation and CQS evaluation in $(\class{C},\class{UCQ})$, where $\class{C} \in \{\class{G},\class{FG}\}$, are computationally hard problems. 
However, there are subclasses of $(\class{C},\class{UCQ})$, seen either as an OMQ language or a class of CQSs, for which the evaluation problem is efficient in the sense of being tractable or being fixed-parameter tractable (FPT) where the parameter is the size of the OMQ (resp., CQS).
In particular, with $\class{UCQ}_k$ being the class of UCQs of treewidth at most $k$, we know that: (1) for each $k \geq 1$, OMQ evaluation in $(\class{G},\class{UCQ}_k)$ is in $\FPT$ (actually, this is shown in this work --  Proposition~\ref{pro:fg-omqs-para-complexity}), and (2) for each $k \geq 1$, CQS evaluation in $(\class{FG},\class{UCQ}_k)$ is in \PTime.
Both statements rely on the well-known result that query evaluation for UCQs of bounded treewidth is tractable~\cite{ChRa00}.
%
%
%
In view of this, 
it is natural to ask whether we can precisely characterize
the limits of efficient OMQ and CQS evaluation. The goal of this work is to provide such efficiency characterizations.
%

A seminal result by Grohe precisely characterizes the (recursively enumerable) classes of CQs over schemas of bounded arity that can be evaluated in polynomial time (under the assumption that $\FPT \neq \W$)~\cite{Grohe07}: this is the case if and only if for some $k \geq 1$, every CQ in the class is equivalent to a CQ of treewidth $k$. Grohe's result also establishes that \PTime~and $\FPT$ coincide for evaluating classes of CQs.
The above can be naturally generalized to UCQs.

Efficiency characterizations in the same spirit have been recently obtained for OMQ languages based on description logics (DLs). In particular,~\cite{BFLP19} precisely characterizes the (recursively enumerable) classes of OMQs from $(\mathcal{ELHI}_\bot,\class{UCQ})$, where $\mathcal{ELHI}_\bot$ is an important DL, essentially a fragment of guarded TGDs. Here, OMQ evaluation is in $\FPT$ (assuming $\FPT \neq \W$) if and only if for some $k \geq 1$, every OMQ in the class is equivalent to an OMQ of treewidth $k$, i.e., an OMQ such that the UCQ in it is of treewidth $k$. Note that
the equivalence is now on the level of OMQs rather than on the level of UCQs. 
The same work~\cite{BFLP19} precisely characterizes the (recursively enumerable) classes of OMQs from $(\mathcal{ELH}_\bot,\class{UCQ})$, where $\mathcal{ELH}_\bot$ is a key fragment of $\mathcal{ELHI}_\bot$ that underpins the OWL 2 EL profile of the OWL 2 recommendation~\cite{owl2}.  Here, evaluation is in \PTime~(again assuming $\FPT \neq \W$) if and only if for some $k \geq 1$, every OMQ in the class is equivalent to an OMQ of treewidth $k$.
This also shows that \PTime~and $\FPT$ coincide for evaluating classes of OMQs from $(\mathcal{ELH}_\bot,\class{UCQ})$.
The classification of $(\mathcal{ELH}_\bot,\class{UCQ})$, however, is subject to the condition
that the ontology does not introduce relations beyond those admitted in the database.

\medskip

\noindent \paragraph{Our Results.}  Our main results are as follows, under the widely believed complexity-theoretic assumption that $\FPT \neq \W$:
\begin{enumerate}
\item 
 For a recursively enumerable class of OMQs from $(\class{G},\class{UCQ})$ over a schema of bounded arity, evaluation is in $\FPT$ if and only if there is $k \geq 1$ such that each OMQ in the class is equivalent to one from $(\class{G},\class{UCQ}_k)$ (Theorem~\ref{the:fpt-characterization-guarded-omqs}).
	
\smallskip  
\item For a recursively enumerable class of CQSs from $(\class{G},\class{UCQ})$ over a schema of bounded arity, evaluation is in \PTime~if and only if evaluation is in $\FPT$ if and only if there is $k \geq 1$ such that each CQS in the class is equivalent to one from $(\class{G},\class{UCQ}_k)$ (Theorem~\ref{the:characterization-guarded-cqss}).

\smallskip  
\item Our final result concerns CQSs from $(\class{FG},\class{UCQ})$ in which the TGDs have 
a bounded number of at most $m \geq 1$ head atoms; let $\class{FG}_m$ be the obtained class.
	For a recursively enumerable class of CQSs from $(\class{FG}_m,\class{UCQ})$ over a schema of bounded arity, evaluation is in \PTime~if and only if evaluation is in $\FPT$ if and only if there is $k \geq 1$ such that each CQS in the class is equivalent to one from $(\class{FG}_m,\class{UCQ}_k)$ (Theorem~\ref{the:characterization-fr-guarded-cqss}).
	
	
	
\end{enumerate}


It is not surprising that characterization (1) talks only about $\FPT$ and not \PTime~complexity since we know from~\cite{CaGK13} that the \EXP-hardness of OMQ evaluation in $(\class{G},\class{UCQ})$ in the case of schemas of bounded arity holds even if the actual query is an atomic query of the simplest form, i.e., a propositional atom.
Moreover, it turns out that the notion of being equivalent to an OMQ of bounded treewdith is not enough for characterizing evaluation in $\FPT$ for frontier-guarded TGDs. We can show that OMQ evaluation in $(\class{FG},\class{UCQ}_k)$ is $\W$-hard -- this is actually easily inherited from the fact that Boolean CQ evaluation is $\W$-hard~\cite{PaYa99}.

For all the above characterizations, if the efficiency condition fails, i.e., there is no integer $k \geq 1$ such that each OMQ or CQS in the considered class is equivalent to one in which the actually query falls within $\class{UCQ}_k$, then the evaluation problem is $\W$-hard.
Showing these lower bounds is actually the most challenging task underlying our efficiency characterizations. Together with the assumption that $\FPT \neq \W$, they establish that efficient evaluation implies the efficiency condition.
For (1) and (3), this is done via an fpt-reduction from the parameterized version of the $k$-clique problem, a well-known $\W$-hard problem, building on Grohe's result which also relies on an fpt-reduction from $k$-clique.
For (2), we exploit a novel connection between OMQ evaluation and CQS evaluation, which is of independent interest. In fact, we provide an fpt-reduction from OMQ evaluation to CQS evaluation.

At this point, we would like to stress that the fpt-reductions underlying (1) and (3) are not merely adaptations of the fpt-reduction by Grohe.
For (1), the fact that the ontology can introduce additional relations beyond the data signature causes serious challenges. Moreover, unlike the characterizations of~\cite{BFLP19} for OMQs based on DLs, where only unary and binary relations are used, we need to deal with relations of arity beyond two, which also causes additional non-trivial complications that require novel ideas and techniques.
For (3), unlike Grohe, we are more constrained in defining the right database as it must satisfy the given set of constraints. Moreover, the useful notion of core of a CQ, which was crucial for Grohe's proof, cannot be directly used in the presence of constraints.

We further study the complexity of the associated meta problems of
deciding whether, for some fixed $k$, a given OMQ or CQS is equivalent to one in which the actual query falls within $\class{UCQ}_k$. We show that in all the considered cases the problem is \TWOEXP-complete (under the mild assumption that $k$ is at least the maximum arity of the occurring relation symbols, minus one).
%
%
Note that decidability of the meta problem is needed to prove the lower bounds described above, but we also consider it interesting in its own right.

\section{Preliminaries}\label{sec:preliminaries}

We consider the disjoint countably infinite sets $\ins{C}$ and $\ins{V}$ of {\em constants} and {\em variables}, respectively.
We refer to constants and variables as {\em terms}. For an integer $n \geq 1$, we may write $[n]$ for the set $\{1,\ldots,n\}$.

\medskip

\noindent
\paragraph{Relational Databases.} A {\em (relational) schema} $\ins{S}$ is a finite set of relation symbols (or predicates) with associated arity. We write $\ar{R}$ for the arity of a predicate $R$, and $\ar{\ins{S}}$ for the arity of $\ins{S}$, that is, the number $\max_{R \in \ins{S}} \{\ar{R}\}$. An {\em atom} over $\ins{S}$ is an expression of the form $R(\bar t)$, where $R \in \ins{S}$ and $\bar t$ is an $\ar{R}$-tuple of terms.
An {\em instance} over $\ins{S}$, or simply {\em $\ins{S}$-instance}, is a (possibly infinite) set of atoms over $\ins{S}$ that contain only constants, while a {\em database} over $\ins{S}$, or simply {\em $\ins{S}$-database}, is a finite $\ins{S}$-instance. We write $\adom{I}$ for the set of constants in an instance $I$.
For a set $T \subseteq \adom{I}$, we denote by $I_{|T}$ the restriction of $I$ to atoms that mention only constants of $T$. 
%
%
A {\em homomorphism} from $I$ to an instance $J$ is a function $h : \adom{I} \ra \adom{J}$ such that $R(h(\bar t)) \in J$ for every $R(\bar t) \in I$. We write $I \ra J$ for the fact that there is a homomorphism from $I$ to $J$.

\medskip
\noindent
\paragraph{Conjunctive Queries.}
A {\em conjunctive query} (CQ) over a schema $\ins{S}$ is a first-order formula of the form 
$
q(\bar x) :=
\exists \bar y \, \big(R_1(\bar x_1) \wedge \dots \wedge R_m(\bar x_m)\big),
$
where each $R_i(\bar x_i)$, for $i \in [m]$, is an atom over $\ins{S}$ that contains only variables,
each variable mentioned in the $\bar x_i$s appears either in $\bar x$ or $\bar y$, and $\bar x$ contains all the free variables of $q$ called the {\em answer variables}.
%
%
%
Every CQ $q$ can be naturally seen as a database $D[q]$, known as the {\em canonical database} of $q$, obtained by dropping the existential quantifier prefix and viewing variables as constants.
We may simply write $q$ instead of $D[q]$.
A homomorphism from a CQ $q$ to an instance $I$ is a homomorphism from $D[q]$ to $I$.
A tuple $\bar c \in \adom{I}^{|\bar x|}$ is an {\em answer} to $Q$ over $I$ if there is a homomorphism $h$ from $q$ to $I$ with $h(\bar x) = \bar c$.
The {\em evaluation of $q(\bar x)$ over $I$}, denoted $q(I)$, is the set of all answers to $Q$ over $I$.
%
%
%
We write $\class{CQ}$ for the class of CQs.
A {\em union of conjunctive queries} (UCQ) over a schema $\ins{S}$ is a first-order formula of the form 
$
q(\bar x) := q_1(\bar x) \vee \cdots \vee q_n(\bar x),
$
where $n \geq 1$, and $q_i(\bar x)$, for $i \in [n]$, is a CQ over $\ins{S}$. 
The evaluation of $q$ over an instance $I$, denoted $q(I)$, is defined as the set of tuples $\bigcup_{i \in [n]} q_i(I)$.
We write $\class{UCQ}$ for the class of UCQs.
The {\em arity} of a (U)CQ is defined as the number of its answer variables. A (U)CQ of arity zero is called {\em Boolean}, and it can have as an answer only the empty tuple. For a Boolean (U)CQ $q$, we may write $I \models q$, if $q(I) = \{()\}$, and $I \not\models q$, otherwise.


%

\medskip

\noindent
\paragraph{Treewidth.}
A central notion in our work is that of treewidth, which measures the degree of tree-likeness of a graph.
Let $G = (V,E)$ be an undirected graph. A  \emph{tree decomposition} of $G$ is a pair $\delta = (T_\delta, \chi)$, where $T_\delta = (V_\delta,E_\delta)$ is a tree, and $\chi$ is a labeling function $V_\delta \ra 2^{V}$, i.e., $\chi$  assigns a subset of $V$ to each node of $T_\delta$, such that:
\begin{enumerate}
	\item $\bigcup_{t \in V_\delta} \chi(t) = V$.
	
	\item If $\{u,v\} \in E$, then $u,v \in \chi(t)$ for some $t \in V_\delta$.
	
	\item For each $v \in V$, the set of nodes $\{t \in V_\delta \mid v \in \chi(t)\}$ induces a connected subtree of $T_\delta$.
\end{enumerate}
The {\em width} of $\delta$ is the number $\max_{t \in V_\delta}
\{|\chi(t)|\} - 1$. If the edge-set $E$ of $G$ is non-empty, then the
{\em treewidth} of $G$ is the minimum width over all its tree
decompositions; otherwise, it is defined to be one.
Each instance $I$ is associated with an undirected graph (without self loops) $G^I = (V,E)$, called the {\em Gaifman graph} of $I$, defined as follows: $V = \adom{I}$, and $\{a,b\} \in E$ iff there is an atom $R(\bar t) \in I$ that mentions both $a$ and $b$.
%
The treewidth of $I$ is the treewidth of~$G^I$. For $k \geq 1$,
$\tw{k}$ is the class of instances of treewidth at most~$k$.

The {\em evaluation problem} for (U)CQs takes as input a (U)CQ $q(\bar x)$, a database $D$, and a candidate answer $\bar c$, and asks whether $\bar c \in q(D)$. It is well-known that (U)CQ evaluation  is \NP-complete~\cite{ChMe77}.
On the other hand, it becomes tractable by restricting the syntactic shape of CQs. One of the most widely studied such restrictions is bounded treewidth.
Formally, a CQ $q(\bar x) = \exists \bar y \, \phi(\bar x,\bar y)$ has treewidth $k \geq 1$ if $G^q_{|\bar y}$ has treewidth $k$, where $G^q_{|\bar y}$ is the subgraph of $G^q$ induced by the elements of $\bar y$.
Note that the treewidth of $q$ is defined in a more liberal way than usual. The standard definition considers the treewitdh of $G^q$, 
while here the treewidth of $q$ is only measured with respect to the subgraph of $G^q$ induced by its existentially quantified variables.
A UCQ $q$ has treewidth $k$ if each of its disjuncts has treewidth at most $k$.
We write $\class{CQ}_k$ (resp., $\class{UCQ}_k$) for the class of CQs (resp. UCQs) of treewidth at most $k \geq 1$.
The next well-known result illustrates the usefulness of bounding the treewidth.
%

\begin{proposition}[\cite{ChRa00}]\label{pro:cq-bounded-tw}
	Fix $k \geq 1$. Given a database $D$, an $n$-ary query $q \in \class{CQ}_k$, and a tuple $\bar c \in \adom{D}^{n}$, the problem of deciding whether $\bar c \in q(D)$ can be solved in time $O(\size{D}^{k+1} \cdot \size{q})$.\footnote{As usual, given a syntactic object $O$, we write $\size{O}$ for its size.}
\end{proposition}



\medskip

\noindent
\paragraph{Tuple-generating Dependencies.} A {\em tuple-generating dependency} (TGD) $\sigma$ over $\ins{S}$ is a constant-free first-order sentence of the form
$
\forall \bar x \forall \bar y \, \big(\phi(\bar x,\bar y)
\rightarrow \exists \bar z \, \psi(\bar x,\bar z)\big),
$
where $\phi$ is a possibly empty conjunction of atoms over $\ins{S}$,
while $\psi$ is a non-empty conjunction of atoms over $\ins{S}$. For simplicity, we write $\sigma$ as $\phi(\bar x,\bar y) \rightarrow \exists \bar z \, \psi(\bar x,\bar z)$, and use comma instead of $\wedge$ for joining atoms. We call $\phi$ and $\psi$ the {\em body} and {\em head} of $\sigma$, denoted $\body{\sigma}$ and $\head{\sigma}$, respectively.
The {\em frontier} of $\sigma$, denoted $\fr{\sigma}$, is the set of variables $\bar x$, i.e., the variables that appear both in the body and the head of $\sigma$.
%
The TGD $\sigma$ above is logically equivalent to the expression $\forall \bar x \, (q_\phi(\bar x) \rightarrow q_\psi(\bar x))$, where $q_\phi(\bar x)$ and $q_\psi(\bar x)$ are the CQs $\exists \bar y \, \phi(\bar x,\bar y)$ and $\exists \bar z \, \psi(\bar x,\bar z)$, respectively. Therefore, an instance $I$ over $\ins{S}$ satisfies $\sigma$, denoted $I \models \sigma$, if $q_\phi(I) \subseteq q_\psi(I)$. An instance $I$ {\em satisfies} a set $\dep$ of TGDs, denoted $I \models \Sigma$, if $I \models \sigma$ for each $\sigma 
\in \dep$. 
Henceforth, whenever we refer to a set of TGDs we mean a finite set.
We write $\class{TGD}$ for the class of TGDs, that is, the family of all possible sets of TGDs.

\medskip

\noindent
\paragraph{Frontier-Guardedness.} A TGD $\sigma$ is {\em guarded} if either $\body{\sigma}$ is empty, or there exists an atom $\alpha$ in its body that contains all the variables occurring in $\body{\sigma}$~\cite{CaGK13}. Such an atom $\alpha$ is the {\em guard} of $\sigma$, denoted $\guard{\sigma}$. 
We write $\class{G}$ for the class of guarded TGDs.
A natural generalization of guardedness is frontier-guardedness, where only the frontier variables must be guarded. Formally, a TGD $\sigma$ is {\em frontier-guarded} if either $\body{\sigma}$ is empty or there exists an atom $\alpha$ in its body that contains all the variables of $\fr{\sigma}$~\cite{BLMS11}, which we call again guard and denote as $\guard{\sigma}$. The class of frontier-guarded TGDs is denoted by $\class{FG}$.
Clearly, $\class{G} \subsetneq \class{FG} \subsetneq \class{TGD}$.

\medskip

\noindent
\paragraph{The Chase Procedure.} The {\em chase} is a useful tool when reasoning with TGDs \cite{MaMS79,JoKl84, FKMP05, CaGK13}. We first define a single chase step. Let $I$ be an instance over a schema $\ins{S}$ and $\sigma$ a TGD of the form $\phi(\bar x,\bar y) \rightarrow \exists \bar z \, \psi(\bar x,\bar z)$ over
$\ins{S}$. We say that $\sigma$ is \emph{applicable} with respect to
$I$, if there exists a tuple $(\bar c,\bar c')$ of constants in $I$ such
that $\phi(\bar c,\bar c') \subseteq I$. In this case, {\em the result
  of applying $\sigma$ over $I$ with $(\bar c,\bar c')$} is the
instance $J = I \cup \psi(\bar c,\bar c'')$, where $\bar c''$ is the
tuple obtained from $\bar z$ by simultaneously replacing each variable $z$ with a fresh distinct constant not occurring in $I$. For such a single chase step we write $I \xrightarrow{\sigma, \,(\bar c,\bar c')} J$.

Let $I$ be an instance and $\Sigma$ a set of TGDs. A {\em chase
  sequence for $I$ under $\Sigma$} is a sequence of chase steps $ I_0
\xrightarrow{\sigma_0,\,(\bar c_0,\bar c'_0)} I_1
\xrightarrow{\sigma_1,\,(\bar c_1,\bar c'_1)} I_2 \dots $ such that
(1) $I_0 = I$, (2) $\sigma_i \in \Sigma$ for each $i \geq 0$, and (3)
$J \models \Sigma$, where $J = \bigcup_{i \geq 0} I_i$.  The instance
$J$ is the (potentially infinite) {\em result} of this chase sequence,
which always exists.
Since we consider the oblivious chase, i.e., a TGD is triggered whenever its body is satisfied no matter whether its head is satisfied, every chase sequence for $I$ under $\dep$ leads to the same result (up to isomorphism). Thus, we can refer to {\em the} result of the chase for $I$ under $\dep$, denoted $\chase{I}{\dep}$.
The key property of the chase follows:

\begin{proposition}\label{pro:chase}
Consider an instance $I$ and a set $\dep$ of TGDs. For every instance $J$ such that $J \supseteq I$ and $J \models \dep$, $\chase{I}{\dep} \ra J$ via a homomorphism that is the identity on $\adom{I}$.
\end{proposition}

\noindent
\paragraph{Parameterized Complexity.}
Parameterized complexity has a central role in our work. A \emph{parameterized problem} over an alphabet $\Lambda$ is a pair $(P,\kappa)$, with $P \subseteq \Lambda^*$ a decision problem and $\kappa$ a \emph{parameterization} of $P$, that
is, a \PTime~computable function \mbox{$\kappa: \Lambda^* \rightarrow \mathbb{N}$}. A prime example is {\sf p}-{\sf Clique}, where $P$ is the set of all
pairs $(G,k)$ with $G$ an undirected graph that contains a $k$-clique
and $\kappa(G,k)=k$.

A problem $(P,\kappa)$ is \emph{fixed-parameter tractable} (fpt) if
there is a computable function $f: \mathbb{N} \rightarrow \mathbb{N}$
and an algorithm that decides $P$ in time $|x|^{O(1)} \cdot
f(\kappa(x))$, where $x$ denotes the input. We use \FPT~to denote the
class of all parameterized problems that are fixed-parameter
tractable. Notice that \FPT~corresponds to a relaxation of the usual
notion of tractability: if $P$ is in \PTime, then $(P,\kappa)$~is in
$\FPT$ for any $\kappa$, but the latter might also be the case when 
$P$ is \NP-hard.

An \emph{fpt-reduction} from a problem $(P_1,\kappa_1)$ over $\Lambda_1$ to a problem $(P_2,\kappa_2)$ over $\Lambda_2$ is a function $\rho : \Lambda_1^* \rightarrow \Lambda_2^*$ such that, for some computable functions $f,g: \mathbb{N} \rightarrow \mathbb{N}$,
\begin{enumerate}
	\item $x \in P_1$ iff $\rho(x) \in P_2$, for all $x \in \Sigma_1^*$;
	
	\item $\rho(x)$ is computable in time $\size{x}^{O(1)} \cdot f(\kappa_1(x))$, for $x \in \Lambda_1^*$;

	\item $\kappa_2(\rho(x)) \leq g(\kappa_1(x))$, for all $x \in \Lambda_1^*$.	
\end{enumerate}

An important parameterized complexity class is $\W \supseteq \FPT$. Hardness for $\W$ is defined in terms of fpt-reductions. It is generally believed that $\FPT \neq \W$, the status of this problem being comparable to that of $\PTime \neq \NP$. Hence, if a parameterized
problem $(P,\kappa)$ is $\W$-hard, then $(P,\kappa)$ is not fpt unless $\FPT = \W$. A well-known $\W$-hard problem is {\sf p}-{\sf Clique}~\cite{DoFe95}.


\section{The Two Facets of TGDs  in Querying}\label{sec:facets}

As discussed in Section~\ref{sec:introduction}, TGDs can be used as:
\begin{enumerate}
\item {\em Ontology axioms} that enrich incomplete data with domain knowledge, which leads to more complete answers.
\item {\em Integrity constraints} that specify semantic properties satisfied by all databases, which paves the way to constraint-aware query optimization techniques.
\end{enumerate}
Depending on how TGDs are used, we obtain different evaluation problems, which we now formalize.

\subsection{TGDs as Ontology Axioms}

When a set of TGDs is used as ontology axioms, or an ontology, it is typically seen, together with the actual query, as one composite query, called ontology-mediated query.
Formally, an {\em ontology-mediated query} (OMQ) is a triple $Q = (\ins{S},\dep,q)$, where $\ins{S}$ is a schema, called {\em data schema}, which indicates that $Q$ will be evaluated over $\ins{S}$-databases, $\dep$ is a set of TGDs over an extended schema $\ins{T} \supseteq \ins{S}$, called {\em ontology}, and $q$ is a UCQ over $\ins{T}$. In case $\ins{S} = \ins{T}$, we say that $Q$ has {\em full data schema}.
The {\em arity} of $Q$ is defined as the arity of $q$. We write $Q(\bar x)$ to emphasize that the answer variables of $q$ are $\bar x$, and say that $Q$ is \emph{over $\ins{T}$} to mean that
the extended schema of $Q$ is $\ins{T}$.
%



The evaluation of an OMQ $Q(\bar x) = (\ins{S},\dep,q(\bar x))$ over an $\ins{S}$-database $D$ consists of all the tuples that are answers to $q$ over each model of $D$ and $\dep$. A {\em model} of $D$ and $\dep$ is an instance $I$ such that $I \supseteq D$ and $I \models \dep$. A tuple $\bar c \in \adom{D}^{|\bar x|}$ is an {\em answer} to $Q$ over $D$ if $\bar c \in q(I)$ for each model $I$ of $D$ and $\dep$.
The {\em evaluation of $Q(\bar x)$} over $D$, denoted $Q(D)$, is the set of all answers to $Q$ over $D$.
By Proposition~\ref{pro:chase}, and the monotonicity of UCQs, we get the following useful result:

\begin{proposition}\label{pro:omq-chase}
For every OMQ $Q = (\ins{S},\dep,q) \in (\class{TGD},\class{UCQ})$, $Q(D) = q(\chase{D}{\dep})$. 
\end{proposition}

%

We write $(\class{C},\class{Q})$ for the class of OMQs, called {\em OMQ language}, in which the ontology is formulated in the class of TGDs $\class{C}$, and the actual query is coming from the class of queries $\class{Q}$; for example, we may write $(\class{G},\class{CQ})$, $(\class{FG},\class{UCQ})$, $(\class{G},\class{UCQ}_k)$, for some $k \geq 1$, etc.
This brings us to the evaluation problem for OMQ languages $\class{O}$:

\medskip

\begin{center}
	\fbox{\begin{tabular}{ll}
			{\small PROBLEM} : & {\sf OMQ}-{\sf Evaluation}($\class{O}$)
			\\{\small INPUT} : & An OMQ $Q = (\ins{S},\dep,q(\bar x)) \in \class{O}$,\\
			& an $\ins{S}$-database $D$, and a tuple $\bar c \in \adom{D}^{|\bar x|}$
			\\
			{\small QUESTION} : &  Is it the case that $\bar c \in Q(D)$?
	\end{tabular}}
\end{center}

\medskip

We are also interested in the parameterized version of the above problem, dubbed {\sf p-OMQ}-{\sf Evaluation}($\class{O}$), with the parameter being the size of the OMQ $Q$, as customary in the literature~\cite{PaYa99}.
Thus {\sf p-OMQ}-{\sf Evaluation}($\class{O}$) is in \FPT~if it can be solved in time $\size{D}^{O(1)} \cdot f(\size{Q})$ for some computable function $f: \mathbb{N} \rightarrow \mathbb{N}$.

It is well-known that \text{\rm {\sf OMQ}-{\sf Evaluation}($\class{TGD},\class{CQ}$)} is undecidable; see, e.g.,~\cite{CaGK13}. On the other hand, if we focus on OMQs in which the ontology is formulated as a set of frontier-guarded TGDs, then the problem becomes decidable, in fact, it is \TWOEXP-complete~\cite{BMRT11}.
At this point, one may wonder whether bounding the treewidth of the CQs will have the same positive effect as in the case of CQ evaluation (see Proposition~\ref{pro:cq-bounded-tw}). It is implicit in~\cite{CaGK13} that this is not the case. The next result summarizes some key facts about the evaluation problem for OMQs based on (frontier-)guarded TGDs.


\begin{proposition}\label{pro:g-omqs-complexity}
	It holds that:
	\begin{enumerate}
		\item \text{\rm {\sf OMQ}-{\sf Evaluation}($\class{FG},\class{UCQ}$)} is \TWOEXP-complete even for schemas of bounded arity.
		
		\item \text{\rm {\sf OMQ}-{\sf Evaluation}($\class{G},\class{UCQ}$)} is \TWOEXP-complete, and becomes \EXP-complete for schemas of bounded arity.	
        
		
		
		\item For each $k \geq 1$, \text{\rm {\sf OMQ}-{\sf Evaluation}($\class{G},\class{CQ}_k$)} is \TWOEXP-complete, and still \EXP-hard for schemas of bounded arity.
	\end{enumerate}
\end{proposition}


%
Since the parameterized version of the evaluation problem for CQs is \W-hard~\cite{PaYa99}, even for schemas of bounded arity, we can immediately conclude that \text{\rm {\sf p}-{\sf OMQ}-{\sf Evaluation}($\class{G},\class{CQ}$)} is \W-hard, even for schemas of bounded arity.
Do we gain something if we focus on UCQs of bounded treewidth? The answer is negative for frontier-guarded TGDs, but affirmative for guarded TGDs.
%

\begin{proposition}\label{pro:fg-omqs-para-complexity}
	It holds that:
	\begin{enumerate}
		
		\item \text{\rm {\sf p}-{\sf OMQ}-{\sf Evaluation}($\class{G},\class{CQ}$)} is \W-hard even for schemas of bounded arity.
		
		\item \text{\rm {\sf p}-{\sf OMQ}-{\sf Evaluation}($\class{FG},\class{CQ}_k$)} is \W-hard even for schemas of bounded arity.
		
		\item For each $k \geq 1$, \text{\rm {\sf p}-{\sf OMQ}-{\sf Evaluation}($\class{G},\class{UCQ}_k$)} is in \FPT.
	\end{enumerate}
\end{proposition}

As discussed above, item (1) is a consequence of the fact that the parameterized version of CQ evaluation is \W-hard, even for schemas of bounded arity.
Item (2) is a consequence of the fact that the parameterized version of CQ evaluation is \W-hard, even for Boolean CQs over a schema of bounded arity, while a Boolean CQ can be transformed into a frontier-guarded TGD. Indeed, given a Boolean CQ $\exists \bar x \, \phi(\bar x)$, $\phi(\bar x) \ra \text{\rm Ans}$, where $\text{\rm Ans}$ is a $0$-ary predicate, is trivially a frontier-guarded TGD since its frontier is empty.

We now briefly explain how item (3) is shown, while the details can be found in the appendix.
A TGD is called {\em linear} if it has only one atom in its body, while the class of linear TGDs is denoted $\class{L}$.
%
The key ingredient underlying item (3) is that, given an $\ins{S}$-database $D$ and an OMQ $Q = (\ins{S},\dep,q)$ from $(\class{G},\class{UCQ})$, $Q(D)$ coincides with the evaluation of $q$ over an initial finite portion $C$ of $\chase{D^*}{\dep^*}$, where $D^*$ can be computed from $D$ and $\dep$, and $\dep^* \in \class{L}$ can be computed solely from $\dep$. Roughly, $C$ is the finite instance obtained by keeping only the atoms of $\chase{D^*}{\dep^*}$ up to a finite level that depends only on $\dep$ and $q$, while the notion of level indicates the distance of an atom in the chase from the starting database.
Furthermore, the instance $C$ can be computed in time $\size{D}^{O(1)} \cdot f(\size{Q})$ for some computable triple exponential function $f: \mathbb{N} \rightarrow \mathbb{N}$.
Therefore, to decide whether a tuple $\bar c$ over $\adom{D}$ belongs to $Q(D)$, it suffices to construct the finite instance $C$, and accept if $\bar c$ belongs to $q(C)$; otherwise, reject.
Since $q \in \class{UCQ}_k$, by Proposition~\ref{pro:cq-bounded-tw}, the overall procedure takes time $\size{D}^{O(1)} \cdot g(\size{Q})$ for some computable triple exponential function $g: \mathbb{N} \rightarrow \mathbb{N}$, and the claim follows.

\subsection{TGDs as Integrity Constraints}

When a set of TGDs is used as integrity constraints, the problem that we are interested in is simply a refinement of the standard query evaluation problem, with the additional promise that the input database satisfies the given set of TGDs. To this end, the evaluation problem is parameterized, not only with the query language in question, but also with the class of TGDs from which the constraints are coming. Formally, a {\em constraint-query specification} (CQS) over a schema $\ins{T}$ is a pair $S = (\dep,q)$, where $\dep$ is a set of TGDs over $\ins{T}$, the set of {\em integrity constraints}, and $q$ a UCQ over $\ins{T}$. 
%
We overload the notation and write $(\class{C},\class{Q})$ for the class of CQSs in which the set of integrity constraints is formulated in the class of TGDs $\class{C}$, and the query is coming from the class of queries $\class{Q}$. It will be clear from the context whether $(\class{C},\class{Q})$ is an OMQ language or a class of CQSs.
The evaluation problem for CQSs follows:

\medskip

\begin{center}
	\fbox{\begin{tabular}{ll}
			{\small PROBLEM} : & {\sf CQS}-{\sf Evaluation}($\class{O}$)
			\\{\small INPUT} : & A CQS $S = (\dep,q(\bar x)) \in \class{O}$ over a schema $\ins{T}$,\\
			& an $\ins{T}$-database $D$ such that $D \models \dep$, and\\
			& a tuple $\bar c \in \adom{D}^{|\bar x|}$
			\\
			{\small QUESTION} : &  Is it the case that $\bar c \in q(D)$?
	\end{tabular}}
\end{center}

\medskip

We are also interested in the parameterized version of the above problem, which we call {\sf p-CQS}-{\sf Evaluation}($\class{O}$), with the parameter being the size of the CQS $S$. Therefore, {\sf p-CQS}-{\sf Evaluation}($\class{O}$) is in \FPT\ if it can be solved in time $\size{D}^{O(1)} \cdot f(\size{S})$ for some computable function $f: \mathbb{N} \rightarrow \mathbb{N}$.

Recall that the evaluation problem for CQs is \NP-hard~\cite{ChMe77}, while its parameterized version is \W-hard~\cite{PaYa99}, even for schemas of bounded arity.  Hence, the same holds for \text{\rm {\sf CQS}-{\sf Evaluation}($\class{G},\class{CQ}$)}.
%
Now, bounding the treewidth of the CQs has the same positive effect as in the case of CQ evaluation (see Proposition~\ref{pro:cq-bounded-tw}). In fact, \text{\rm {\sf CQS}-{\sf Evaluation}($\class{FG},\class{UCQ}_k$)} is in \PTime. 
%
Thus, the evaluation problem for UCQs, and the evaluation problem for the classes of CQSs based on $\class{G}$ and $\class{FG}$ have the same complexity.

\section{Semantic Tree-Likeness}\label{sec:semantic-tree-likeness}

A seminal result by Grohe precisely characterizes the (recursively enumerable) classes of CQs over schemas of bounded arity that can be evaluated in polynomial time (under the assumption that $\FPT \neq \W$). In fact, this result shows that the classes of CQs over schemas of bounded arity that can be evaluated in polynomial time are precisely those that are semantically tree-like, or, more formally, are of {\em bounded treewidth modulo equivalence}. Moreover, the result states that fpt does not add anything to standard tractability in the considered setting. 
Before giving the statement, let us formalize the notion of bounded treewidth modulo equivalence.
Recall that two CQs $q,q'$ over a schema $\ins{S}$ are {\em equivalent} if, for every $\ins{S}$-database $D$, $q(D) = q'(D)$. For each $k \geq 1$, let $\class{CQ}_{k}^{\equiv}$ be the class of all CQs that are equivalent to a CQ from $\class{CQ}_k$. Grohe's result follows:

\begin{theorem}[Grohe's Theorem~\cite{Grohe07}]\label{the:grohe}
	Fix $r \geq 1$. Let $\class{Q}$ be a recursively enumerable class of CQs over schemas of arity $r$. The following are equivalent, assuming $\FPT \neq \W$:
	\begin{enumerate}
		\item The evaluation problem for $\class{Q}$ is in \PTime.
		
		\item The evaluation problem for $\class{Q}$ is in \FPT.
		
		
		
		\item There is $k \geq 1$ such that $\class{Q} \subseteq \class{CQ}_{k}^{\equiv}$.
	\end{enumerate}
	If either statements is false, then evaluation for $\class{Q}$ is \W-hard.
\end{theorem}

Interestingly, it is decidable whether a CQ is equivalent to one of treewidth $k$. This is shown by exploiting the notion of core. Recall that the {\em core} of a CQ $q$ is a $\subseteq$-minimal subquery of $q$ that is equivalent to $q$. It is known that, for each $k \geq 1$, a CQ $q$ belongs to $\class{CQ}_{k}^{\equiv}$ iff its core is in $\class{CQ}_k$, and that deciding this property is \NP-complete~\cite{DaKV02}.
There is also a natural generalization of this characterization and of Theorem~\ref{the:grohe} to the class of UCQs. 

At this point, it is natural to ask whether it is possible to obtain a characterization of the classes of OMQs (resp., CQSs) based on (frontier-)guarded TGDs for which {\sf OMQ}-{\sf Evaluation} (resp., {\sf CQS}-{\sf Evaluation}) can be efficiently solved, in the same spirit as Grohe's Theorem. In fact, our main question is whether the natural generalization of the notion of bounded treewidth modulo equivalence for CQs to OMQs (resp., CQSs) is a decidable notion that exhausts tractability or fixed-parameter tractability for {\sf OMQ}-{\sf Evaluation} (resp., {\sf CQS}-{\sf Evaluation}), as in the case of Grohe's Theorem.
%

\subsection{Semantic Tree-likeness for OMQs}

We first concentrate on OMQs and introduce the notion of $\class{UCQ}_k$-equivalence, which essentially tells us that an OMQ can be rewritten into an equivalent one where the UCQ belongs to $\class{UCQ}_k$. Let us clarify that for OMQs the notion of equivalence is not applied at the level of the actual query, but to the whole OMQ. Formally, given two OMQs $Q$ and $Q'$, both with data schema $\ins{S}$, $Q$ is {\em contained} in $Q'$, written $Q \subseteq Q'$, if $Q(D) \subseteq Q'(D)$ for every $\ins{S}$-database $D$. We then say that $Q$ and $Q'$ are {\em equivalent}, denoted $Q \equiv Q'$, if $Q \subseteq Q'$ and $Q' \subseteq Q$.
In what follows, let $\class{C}$ be a class of TGDs, i.e., $\class{C} \subseteq \class{TGD}$.

\begin{definition}[$\class{UCQ}_k$-equivalence for OMQs]
	An OMQ $Q = (\ins{S},\dep,q) \in (\class{C},\class{UCQ})$ is {\em $\class{UCQ}_k$-equivalent}, for $k \geq 1$, if there exists an OMQ $Q' = (\ins{S},\dep',q') \in (\class{C},\class{UCQ}_k)$ such that $Q \equiv Q'$.
	Given an OMQ language $\class{O} = (\class{C},\class{UCQ})$, for each $k \geq 1$, let $\class{O}_{k}^{\equiv}$ be the class of all OMQs from $\class{O}$ that are $\class{UCQ}_k$-equivalent. \hfill\markfull
\end{definition}

According to the above definition, we are allowed to rewrite both the ontology and the UCQ. It is conceptually meaningful though to consider also the setting where the ontology cannot be altered. This leads to the uniform version of $\class{UCQ}_k$-equivalence.

\begin{definition}[Uniform $\class{UCQ}_k$-equivalence for OMQs]
	An OMQ $Q = (\ins{S},\dep,q) \in (\class{C},\class{UCQ})$ is {\em uniformly $\class{UCQ}_k$-equivalent}, for $k \geq 1$, if there is $Q' = (\ins{S},\dep,q') \in (\class{C},\class{UCQ}_k)$ such that $Q \equiv Q'$.
	For an OMQ language $\class{O} = (\class{C},\class{UCQ})$ and $k \geq 1$, let $\class{O}_{k}^{\equiv,u}$ be the class of all OMQs from $\class{O}$ that are uniformly $\class{UCQ}_k$-equivalent. \hfill\markfull
\end{definition}

%
The next example shows that both the ontology and the data schema can have an impact on the treewidth. 


\begin{example}\label{exa:relevant-notion}
	We first illustrate that the ontology can have an impact on the treewidth. Consider the OMQ $Q_1 = (\ins{S},\dep,q)$, where
	\begin{eqnarray*}
		\ins{S} &=& \{R_1,R_2,R_3,R_4,P\}\\
		\dep &=& \{R_2(x) \ra R_4(x)\}\\
		q() &=& P(x_2,x_1) \wedge P(x_4,x_1) \wedge P(x_2,x_3) \wedge P(x_4,x_3) \wedge\\
		&& R_1(x_1) \wedge R_2(x_2) \wedge R_3(x_3) \wedge R_4(x_4).
	\end{eqnarray*}
	Note that $\ins{S}$ contains all the predicates in $\dep$ and $q$, while $q$ is a Boolean CQ (the existential quantifiers are omitted). Observe that $q$ is a core from $\class{CQ}_2$, and thus can easily be seen to not belong to $\class{UCQ}_{1}^{\equiv}$. However, $Q_1$ is equivalent to the OMQ $(\ins{S},\dep,q')$, where
	\[
	q'()\ =\ P(x_2,x_1) \wedge P(x_2,x_3)  \wedge R_1(x_1) \wedge R_2(x_2) \wedge R_3(x_3).
	\]
	Since $q' \in \class{CQ}_1$, we get that $Q_1 \in (\class{G},\class{UCQ})_{1}^{\equiv,u}$.
	
%
	We proceed to show that the data schema can also have an impact. Consider the OMQ $Q_2 = (\ins{S}',\dep',q)$ with full data schema, where
		\begin{eqnarray*} 
		\dep' &=& \{S(x) \ra R_1(x), \quad S(x) \ra R_3(x)\}.
	\end{eqnarray*}
	It is not hard to see that $Q_2$ does not belong to $(\class{G},\class{UCQ})_{1}^{\equiv}$. If, however, the predicate $R_1$ is omitted from the signature, then $Q_2$ is equivalent to the OMQ $(\{S, P, R_2, R_3, R_4\},\dep',q'')$, where
	\[
		q''()\ =\ P(x_2,x_1) \wedge P(x_4,x_1)  \wedge R_1(x_1) \wedge R_2(x_2) \wedge R_3(x_1) \wedge R_4(x_4),
	\]
	and thus, it belongs to $(\class{G},\class{UCQ})_{1}^{\equiv,u}$.
	%
\end{example}

\subsection{Semantic Tree-likeness for CQSs}
\label{sect:semtreeCQS}

We now introduce $\class{UCQ}_k$-equivalence for CQSs. For this setting, only the uniform version is relevant. Indeed, given a CQS $S = (\dep,q)$, by altering the set of integrity constraints $\dep$,
we essentially change the semantics of $S$ which is, of course, not our intention.
Given two CQSs $S = (\dep,q)$ and $S' = (\dep,q')$ over a schema $\ins{T}$, we say that $S$ is {\em contained} in $S'$, denoted $S \subseteq S'$, if $q(D) \subseteq q'(D)$ for every $\ins{T}$-database $D$ that satisfies $\dep$. We then say that $S$ is {\em equivalent} to $S'$ if $S \subseteq S'$ and $S' \subseteq S$. We may also write $q \subseteq_\dep q'$ (resp., $q \equiv_\dep q'$) for the fact that $S \subseteq S'$ (resp., $S \equiv S'$).
%
We recall a known result about containment among CQSs, which will be useful for our analysis:

\begin{proposition}\label{pro:cqs-cont-chase}
	 Let $S_1 = (\dep,q_1(\bar x))$ and $S_2 = (\dep,q_2(\bar y))$ be CQSs from $(\class{TGD},\class{UCQ})$ with $|\bar x| = |\bar y|$. Then $S_1 \subseteq S_2$ iff for each $p_1 \in q_1$, there exists $p_2 \in q_2$ such that $\bar x \in p_2(\chase{p_1}{\dep})$.
\end{proposition}



The notion of uniform $\class{UCQ}_k$-equivalence for CQSs follows:

\begin{definition}[Uniform $\class{UCQ}_k$-equivalence for CQSs]
	A CQS $S = (\dep,q)$ from $(\class{C},\class{UCQ})$ is {\em uniformly $\class{UCQ}_k$-equivalent}, for $k \geq 1$, if there exists a CQS $S' = (\dep,q')$ from $(\class{C},\class{UCQ}_k)$ such that $S \equiv S'$.
	Given a class of CQSs $\class{O}$, for $k \geq 1$, let $\class{O}_{k}^{\equiv}$ be the class of all CQSs from $\class{O}$ that are uniformly $\class{UCQ}_k$-equivalent.\footnote{We overload the notation again. It would be clear from the context when $\class{O}_{k}^{\equiv}$ is an OMQ language or a class of CQSs. We also avoid the superindex $u$ since for CQSs we only consider the uniform version of $\class{UCQ}_k$-equivalence.} \hfill\markfull
\end{definition}

Observe that the first part of Example~\ref{exa:relevant-notion} also works when $(\Sigma,q)$ is viewed as a CQS. This illustrates the fact that integrity constraints can have an impact on the treewidth.


\section{Our Results in a Nutshell}\label{sec:efficinecy-boundaries}

Having the relevant notions in place, we can now provide an answer to our main question, that is, whether (uniform) $\class{UCQ}_k$-equivalence of OMQs (resp., uniform $\class{UCQ}_k$-equivalence of CQSs) 
 exhausts tractability and fixed-parameter tractability of {\sf OMQ}-{\sf Evaluation} (resp., {\sf CQS}-{\sf Evaluation}) in the same spirit as Grohe's Theorem.
%
%
In this section, we give an overview of our results, observe interesting connections between the OMQ and CQS settings, and study the associated meta problems. Further details are deferred to the subsequent sections and the appendix.
%

\subsection{The Guarded Case}\label{sec:guarded-results}

\noindent
\paragraph{Ontology-mediated Queries.}
We start with our results on OMQs based on guarded TGDs. We first ask whether (uniform) $\class{UCQ}_k$-equivalence is decidable for $(\class{G},\class{UCQ})$.

\begin{theorem}\label{the:ucq-k-equiv-complexity-omq}
For each $k$, deciding whether a given OMQ $Q$ from $(\class{G},\class{UCQ})$ 
over $\ins{T}$ with $k \geq \ar{\ins{T}}-1$ is (uniformly) $\class{UCQ}_k$-equivalent is \TWOEXP-complete. If existant, an OMQ $Q' \in (\class{G},\class{UCQ}_{k})$ such that $Q \equiv Q'$ can be computed in double exponential time.  \end{theorem}



The above complexity result exploits a characterization of when an OMQ $Q$ from $(\class{G},\class{UCQ})$ is (uniformly) $\class{UCQ}_k$-equivalent, which in turn relies on what we call a $\class{UCQ}_k$-approximation, that is, approximations of $Q$ from below in terms of an OMQ from $(\class{G},\class{UCQ}_k)$.
In a nutshell, a $\class{UCQ}_k$-approximation of $Q = (\ins{S},\dep,q)$ is an OMQ $Q_{k}^{a} = (\ins{S},\dep,q_{k}^{a})$, where $q_{k}^{a}$ belongs to $\class{UCQ}_k$, that behaves like $Q$ over $\ins{S}$-databases of treewidth at most $k$, i.e., for every $\ins{S}$-database $D$ of treewidth at most $k$, $Q(D) = Q_{k}^{a}(D)$. It follows that $Q_{k}^{a}$ is equivalent to $Q$ if and only if $Q$ is $\class{UCQ}_k$-equivalent. The formal definition of $Q_{k}^{a}$ can be found in the appendix. It is signficantly more involved
than in the case of description logics \cite{BFLP19} because there the chase only generates structures
of treewidth one. 

\begin{proposition}\label{pro:guarded-ucq-k-approximation}
	Let $Q$ be an OMQ from $(\class{G},\class{UCQ})$ over $\ins{T}$, and let $k \geq \ar{\ins{T}}-1$. The following are equivalent:
	\begin{enumerate}
		\item $Q$ is $\class{UCQ}_k$-equivalent.
		\item $Q$ is uniformly $\class{UCQ}_k$-equivalent.
		\item $Q \equiv Q_{k}^{a}$.
	\end{enumerate}
\end{proposition}
Since $\class{UCQ}_k$-equivalence and uniform $\class{UCQ}_k$-equivalence 
turn out to be
equivalent, we henceforth only use $\class{UCQ}_k$-equivalence.
Note that Proposition~\ref{pro:guarded-ucq-k-approximation} also provides an approach to deciding $\class{UCQ}_k$-equivalence, and thus to establishing Theorem~\ref{the:ucq-k-equiv-complexity-omq}: compute the $\class{UCQ}_k$-approximation $Q_{k}^{a}$ of $Q$ and accept if $Q \subseteq Q_{k}^{a}$ (note that $Q_{k}^{a} \subseteq Q$ holds always); otherwise, reject.
We can show that $Q_{k}^{a}$ can be computed in double exponential time and it is known that OMQ containment for $(\class{G},\class{UCQ})$ can be decided in double exponential time~\cite{BaBP18}. A naive use of these observations yields only a \FOUREXP~upper bound, but we show in the appendix how to improve this to \TWOEXP. The lower bound is inherited from~\cite{BFLP19}, where the same problem for OMQs based on DLs has been studied.

We remark that the case $k < \ar{\ins{T}}-1$, excluded in Theorem~\ref{the:ucq-k-equiv-complexity-omq} and Proposition~\ref{pro:guarded-ucq-k-approximation}, is somewhat esoteric as there are relation symbols whose arity is so high that they cannot be part of the UCQ in a $\class{UCQ}_k$-approximation (unless variables are reused). In the appendix, we provide a concrete example which illustrates that this case is significantly different from the case $k \geq \ar{\ins{T}}-1$; in particular, we can see that Proposition~\ref{pro:guarded-ucq-k-approximation} is provably wrong.


We now state our main result concerning guarded OMQs, which shows that $\class{UCQ}_k$-equivalence characterizes fpt for classes of OMQs from $(\class{G},\class{UCQ})$ over schemas of bounded arity.

\begin{theorem}[\textbf{Main Result I}]\label{the:fpt-characterization-guarded-omqs}
	Fix $r \geq 1$. Let $\class{O}$ be a recursively enumerable class of OMQs from $(\class{G},\class{UCQ})$ over a schema of arity $r$. The following are equivalent, assuming $\FPT \neq \W$:
	\begin{enumerate}
		\item \text{\rm {\sf p}-{\sf OMQ}-{\sf Evaluation}($\class{O}$)} is in \FPT.
		
		\item There is $k \geq 1$ such that $\class{O} \subseteq (\class{G},\class{UCQ})_{k}^{\equiv}$.
	\end{enumerate}
	If either statement is false, then \text{\rm {\sf p}-{\sf OMQ}-{\sf Evaluation}($\class{O}$)} is \W-hard.
\end{theorem}

The easy direction is (2) implies (1), which exploits Theorem~\ref{the:ucq-k-equiv-complexity-omq}, and the third item of Proposition~\ref{pro:fg-omqs-para-complexity}.
The hard task is to show that (1) implies (2). To this end, since we assume that $\FPT \neq \W$, it suffices to show the following lower bound, which is  our main technical result on OMQs based on guarded TGDs:

\begin{theorem}\label{the:omq-main-technical-result}
	Fix $r \geq 1$. Let $\class{O}$ be a recursively enumerable class of OMQs from $(\class{G},\class{UCQ})$ over a schema of  arity $r$, and, for each $k \geq 1$, $\class{O} \not\subseteq (\class{G},\class{UCQ})_{k}^{\equiv}$. Then, \text{\rm {\sf p}-{\sf OMQ}-{\sf Evaluation}($\class{O}$)} is \W-hard. 
\end{theorem}

In view of Proposition~\ref{pro:g-omqs-complexity}, which states that, for each $k \geq 1$, \text{\rm {\sf OMQ}-{\sf Evaluation}($\class{G},\class{CQ}_k$)} is \EXP-hard, it is not surprising that Theorem~\ref{the:fpt-characterization-guarded-omqs} does not state a pure tractability result. 
One may think that the above result can be easily obtained by using Grohe's construction underlying the fpt-reduction from {\sf p}-{\sf Clique} that establishes the lower bound of Theorem~\ref{the:grohe}.
However, the fact that the ontology can introduce additional relations that are not part of the data schema causes serious challenges. 
%
A detailed sketch of our proof is presented in Section~\ref{sec:guarded}, while full details are in the appendix.

\medskip

\noindent
\paragraph{Constraint Query Specifications.} We now turn our attention to CQSs based on guarded TGDs. Interestingly, there is a strong connection between CQSs and OMQs with full data schema, which allows us to transfer results from the OMQ to the CQS setting. 
Recall that a full data schema consists of all the predicates occurring in the ontology and the UCQ. 
Note that we can naturally convert a CQS $S = (\dep,q)$ over a schema $\ins{S}$ into the OMQ $(\ins{S},\dep,q)$ that has full data schema, denoted $\mathsf{omq}(S)$. 
The following result relates the UCQ$_k$-equivalence 
of CQSs to the UCQ$_k$-equivalence of OMQs.

\begin{proposition}\label{pro:from-omq-to-cqs}
	Consider a CQS $S \in (\class{G},\class{UCQ})$ over $\ins{T}$. For each $k \geq \ar{\ins{T}}-1$, the following are equivalent:
	\begin{enumerate}
		\item $S$ is uniformly $\class{UCQ}_k$-equivalent.
		\item $\mathsf{omq}(S)$ is $\class{UCQ}_k$-equivalent.
	\end{enumerate}
\end{proposition}


From Theorem~\ref{the:ucq-k-equiv-complexity-omq} and Proposition~\ref{pro:from-omq-to-cqs}, we get a \TWOEXP~upper bound for deciding uniform $\class{UCQ}_k$-equivalence for guarded CQSs, while the lower bound is inherited from~\cite{BFGP19}. As in Theorem~\ref{the:ucq-k-equiv-complexity-omq}, $k$ should be greater than the arity of the schema, minus one.


\begin{theorem}\label{the:ucq-k-equiv-complexity-guarded-cqs}
	Let $S$ be a CQS from $(\class{G},\class{UCQ})$ over a schema $\ins{T}$. For each $k \geq \ar{\ins{T}}-1$, deciding whether $S$ is uniformly $\class{UCQ}_k$-equivalent is \TWOEXP-complete. 
\end{theorem}

Our main result concerning guarded CQSs shows that uniform $\class{UCQ}_k$-equivalence characterizes tractability and fpt for classes of CQSs from $(\class{G},\class{UCQ})$ over schemas of bounded arity.

\begin{theorem}[\textbf{Main Result II}]\label{the:characterization-guarded-cqss}
	Fix $r \geq 1$. Let $\class{O}$ be a recursively enumerable class of CQSs from $(\class{G},\class{UCQ})$ over a schema of arity $r$. The following are equivalent, assuming $\FPT \neq \W$:
	\begin{enumerate}
		\item \text{\rm {\sf CQS}-{\sf Evaluation}($\class{O}$)} is in \PTime.
		
		\item \text{\rm {\sf p}-{\sf CQS}-{\sf Evaluation}($\class{O}$)} is in \FPT.
		
		
		\item There is $k \geq 1$ such that $\class{O} \subseteq (\class{G},\class{UCQ})_{k}^{\equiv}$.
	\end{enumerate}
	If either statement is false, then \text{\rm {\sf p}-{\sf CQS}-{\sf Evaluation}($\class{O}$)} is \W-hard.
\end{theorem}

The fact that (3) implies (1) is shown in~\cite{BFGP19}, while (1) implies (2) holds by definition. The interesting task is to show that (2) implies (3). To this end, we exploit the subsequent result which relates the evaluation of CQSs to the evaluation of OMQs. Given a class $\class{O}$ of CQSs, we write $\mathsf{omq}(\class{O})$ for the class of OMQs $\{\mathsf{omq}(S) \mid S \in \class{O}\}$.

\begin{proposition}\label{pro:fpt-reduction-omq-to-cqs}
	Consider a class $\class{O}$ of CQSs from $(\class{G},\class{UCQ})$. 
	There exists an fpt-reduction from \text{\rm {\sf p}-{\sf OMQ}-{\sf Evaluation}($\mathsf{omq}(\class{O})$)} to \text{\rm {\sf p}-{\sf CQS}-{\sf Evaluation}($\class{O}$)}.
\end{proposition}

The reduction exploits the fact that guarded TGDs are {\em finitely controllable}, which means that {\sf OMQ}-{\sf Evaluation}($\class{G},\class{UCQ}$) coincides with {\sf OMQ}-{\sf Evaluation}$_{{\sf fin}}$($\class{G},\class{UCQ}$) that considers only {\em finite} models~\cite{BaGO14}.
%
%
Note that the reduction does not extend to frontier-guarded TGDs since it requires the evaluation of TGD bodies to be fpt. Actually, no such reduction is possible for frontier-guarded TGDs (unless $\FPT = \W$) since there exists a class $\class{O}$ of CQSs from $(\class{FG},\class{UCQ}_1)$ such that \text{\rm {\sf p}-{\sf OMQ}-{\sf Evaluation}($\mathsf{omq}(\class{O})$)} is \W-hard, while \text{\rm {\sf p}-{\sf CQS}-{\sf Evaluation}($\class{O}$)} is in FPT: $\class{O}$ consists of CQSs $(\dep,\exists x R(x))$, where $\dep$ contains TGDs of the form $\varphi(x,\bar y) \ra R(x)$, i.e., frontier-guarded TGDs with only one frontier variable.
Now, by Theorem~\ref{the:omq-main-technical-result}, Proposition~\ref{pro:from-omq-to-cqs}, and Proposition~\ref{pro:fpt-reduction-omq-to-cqs}, we get the following result, which establishes the direction (2) implies (3) of Theorem~\ref{the:characterization-guarded-cqss} (since we assume that  $\FPT \neq \W$).

\begin{theorem}\label{the:cqs-guadred-technical-result}
	Fix $r \geq 1$. Let $\class{O}$ be a recursively enumerable class of CQSs from $(\class{G},\class{UCQ})$ over a schema of arity $r$, and, for each $k \geq 1$, $\class{O} \not\subseteq (\class{G},\class{UCQ})_{k}^{\equiv}$. Then, \text{\rm {\sf p}-{\sf CQS}-{\sf Evaluation}($\class{O}$)} is \W-hard. 
\end{theorem}

\subsection{The Frontier-guarded Case}\label{sec:frontier-guarded-results}

By item (2) of Proposition~\ref{pro:fg-omqs-para-complexity}, the notion of (uniform) $\class{UCQ}_k$-equivalence does not provide a characterization of fpt for classes of OMQs from $(\class{FG},\class{UCQ})$. Such a characterization result for OMQs based on frontier-guarded TGDs via an alternative notion of "being equivalent to an OMQ of low treewidth" remains an interesting open problem.
The situation is different for querying in the presence of constraints. In fact, uniform $\class{UCQ}_k$-equivalence allows us to characterize tractability and fpt for classes of CQSs from $(\class{FG},\class{UCQ})$ over schemas of bounded arity~$r$ provided that the number of atoms in the head of TGDs is bounded by an integer $m$. The latter is required in our proof for reasons that are explained in Section~\ref{sec:frontier-guarded}.

The first question is whether uniform $\class{UCQ}_k$-equivalence for CQSs based on frontier-guarded TGDs, with at most $m \geq 1$ head atoms, over schemas of arity $r \geq 1$, is decidable.
We write $\class{FG}_m$ for the class of frontier-guarded TGDs with at most $m$ head atoms. 
%

\begin{theorem}\label{the:ucq-k-equiv-complexity-cqs}
	Fix $r,m \geq 1$. Let $S$ be a CQS from $(\class{FG}_m,\class{UCQ})$ over a schema of arity $r$. For each $k \geq (r \cdot m - 1)$, deciding whether $S$ is uniformly $\class{UCQ}_k$-equivalent is \TWOEXP-complete. 
\end{theorem}


The above result exploits a characterization, analogous to Proposition~\ref{pro:guarded-ucq-k-approximation}, of when a CQS from $(\class{FG}_m,\class{UCQ})$ over a schema of arity $r$ is uniformly $\class{UCQ}_k$-equivalent via a suitable notion of $\class{UCQ}_k$-approximation.
%
The fact that we consider schemas of arity $r$, and frontier-guarded TGDs with at most $m$ atoms in the head,
%
allows us to show that chasing a database of treewidth $k$, with $k \geq r \cdot m - 1$, will lead to an instance that has treewidth $k$.
This in turn allows us to adopt a notion of $\class{UCQ}_k$-approximation that is significantly simpler than the one for guarded OMQs in Proposition~\ref{pro:guarded-ucq-k-approximation}.
%
A {\em contraction} of a CQ $q$ is a CQ obtained from $q$ by identifying variables. When an answer variable $x$ is identified with a non-answer variable $y$, the resulting variable is $x$, while the identification of two answer variables is not allowed.
The $\class{UCQ}_k$-approximation of a CQS $S = (\dep,q) \in (\class{FG},\class{UCQ})$ is the CQS $S_{k}^{a} = (\dep,q_{k}^{a})$, where $q_{k}^{a}$ is the UCQ that consists of all contractions of a CQ from $q$ that belong to $\class{CQ}_k$.
We can then show the following:

\begin{proposition}\label{pro:fr-guarded-ucq-k-approximation-cqs}
	Fix $r,m \geq 1$, $k \geq r \cdot m - 1$. Let $S$ be a CQS from $(\class{FG}_m,\class{UCQ})$ over a shema of arity $r$. The following are equivalent:
	\begin{enumerate}
		\item $S$ is uniformly $\class{UCQ}_k$-equivalent.
		\item $S \equiv S_{k}^{a}$.
	\end{enumerate}
\end{proposition}

From the above result, we get a procedure for deciding uniform $\class{UCQ}_k$-equivalence: compute the $\class{UCQ}_k$-approximation $S_{k}^{a} = (\dep,q_{k}^{a})$ of $S = (\dep,q)$, and accept if $S \subseteq S_{k}^{a}$ (the other direction holds by construction); otherwise, reject.
Clearly, $S_{k}^{a}$ can be computed in exponential time, while $q_{k}^{a}$ consists, in general, of exponentially many CQs, each of polynomial size.
By Proposition~\ref{pro:cqs-cont-chase}, $S \subseteq S_{k}^{a}$ iff for each CQ $p(\bar x) \in q$, there exists $p' \in q_{k}^{a}$ such that $\bar x  \in p'(\chase{p}{\dep})$.
Moreover, by Proposition~\ref{pro:omq-chase} and item (1) of Proposition~\ref{pro:g-omqs-complexity}, we get that checking whether $\bar x  \in p'(\chase{p}{\dep})$ is feasible in double exponential time.
Summing up, to check whether $S \subseteq S_{k}^{a}$ we need to perform exponentially many checks, while each check is feasible in double exponential time. This leads to the desired \TWOEXP~upper bound, while a matching lower bound is inherited from~\cite{BFGP19}.

Our main result concerning frontier-guarded CQSs follows:




\begin{theorem}[\textbf{Main Result III}]\label{the:characterization-fr-guarded-cqss}
	Fix $r,m \geq 1$. Let $\class{O}$ be a recursively enumerable class of CQSs from $(\class{FG}_m,\class{UCQ})$ over a schema of arity $r$. The following are equivalent, assuming $\FPT \neq \W$:
	\begin{enumerate}
		\item \text{\rm {\sf CQS}-{\sf Evaluation}($\class{O}$)} is in \PTime.
		
		\item \text{\rm {\sf p}-{\sf CQS}-{\sf Evaluation}($\class{O}$)} is in \FPT.
		
		
		\item There is $k \geq 1$ such that $\class{O} \subseteq (\class{FG}_m,\class{UCQ})_{k}^{\equiv}$.
	\end{enumerate}
	If either statement is false, then \text{\rm {\sf p}-{\sf CQS}-{\sf Evaluation}($\class{O}$)} is \W-hard.
\end{theorem}

As for Theorem~\ref{the:characterization-guarded-cqss}, the fact that (3) implies (1) is shown in~\cite{BFGP19}, while (1) implies (2) holds by definition. Thus, the non-trivial task is to show that (2) implies (3). To this end, since we assume that $\FPT \neq \W$, it suffices to show the following, which is our main technical result on CQSs based on frontier-guarded TGDs:

\begin{theorem}\label{the:cqs-frontier-guadred-technical-result}
	Fix $r,m \geq 1$. Let $\class{O}$ be a recursively enumerable class of CQSs from $(\class{FG}_m,\class{UCQ})$ over a schema of arity $r$, and, for each $k \geq 1$, $\class{O} \not\subseteq (\class{FG},\class{UCQ})_{k}^{\equiv}$. Then, \text{\rm {\sf p}-{\sf CQS}-{\sf Evaluation}($\class{O}$)} is \W-hard. 
\end{theorem}

%
Unlike Theorem~\ref{the:cqs-guadred-technical-result}, the above result cannot benefit from results on OMQs since (uniform) $\class{UCQ}_k$-equivalence does not characterize fpt for OMQs based on frontier-guarded TGDs.
It is shown via a reduction from {\sf p}-{\sf Clique} by following the same approach as for the lower bound of Grohe's Theorem, and exploiting the fact that frontier-guarded TGDs are finitely controllable.
%
Let us stress, however, that our fpt-reduction is not a merely adaptation of Grohe's reduction.
Unlike Grohe, we are more constrained in defining the right database as it must satisfy the given set of constraints. Moreover, the useful notion of core of a CQ, which was crucial for Grohe's proof, cannot be directly used in the presence of constraints.

%
%


\medskip

\noindent
\paragraph{The Rest of the Paper.} We proceed to give some details on how our main lower bounds are shown; the full proofs are deferred to the appendix. Moreover, we defer to the appendix more details about the proofs of the results on the complexity of deciding $\class{UCQ}_k$-equivalence, and of the upper bounds stated in our main results.
%
%
%
The lower bounds for guarded OMQs and CQS, i.e., Theorems~\ref{the:omq-main-technical-result} and~\ref{the:cqs-guadred-technical-result}, respectively, 
are discussed in Section~\ref{sec:guarded}, while the lower bound for frontier-guarded CQSs, i.e., Theorem~\ref{the:cqs-frontier-guadred-technical-result}, in Section~\ref{sec:frontier-guarded}.
\section{Guardedness}
\label{sec:guarded}
%



\subsection{Ontology-mediated Queries}\label{sec:guarded-omq-details}

We provide a proof sketch of Theorem~\ref{the:omq-main-technical-result},
illustrating only the main ideas and abstracting away many technical
details. We rely on a result by Grohe, stated here in a form tailored towards our proof.  
For $k,\ell \geq 1$,
the \emph{$k \times \ell$-grid} is the graph with vertex set $\{ (i,j)
\mid 1\leq i \leq k \text{ and } 1 \leq j \leq \ell \}$, and an edge
between $(i,j)$ and $(i',j')$ iff $|i-i'|+|j-j'|=1$. A \emph{minor} of
an undirected graph is defined in the usual way, see,
e.g.,~\cite{Grohe07}. When $k$ is understood from the context, we use
$K$ to denote~${k \choose 2}$. We say that a database $D$ is
\emph{connected} if the graph $G^D$ is connected. A constant $a \in \adom{D}$ is \emph{isolated in} $D$ if there is only a single atom $R(\abf) \in D$ with $a \in \abf$.
\begin{restatable}[Grohe]{theorem}{grohetechnew}
	\label{thm:grohetechnew}
	Given an undirected graph $G$, a $k \geq 1$, a
	connected 
	$\Sbf$-database~$D$, and a set $A \subseteq \adom{D}$ such that
	the restriction $G^{D}_{|_A}$ of the Gaifman graph of $D$
	to vertices $A$
	contains the $k \times K$-grid as a minor, 
	one can construct in time $f(k)\cdot\mn{poly}(\size{G},\size{D})$
	an $\Sbf$-database $D_G$ with
	 $\adom{D} \setminus A \subseteq \adom{D_G}$
	such that:
	  \begin{enumerate}
	
	  \item there is a surjective 
	    homomorphism $h_0$ from $D_G$ to $D$ that is the identity
	    on $\adom{D} \setminus A$, 

	
	
	
	
	  \item $G$ contains a $k$-clique iff there is a homomorphism $h$
	    from $D$ to $D_G$ such that $h$ is the identity
	    on $\adom{D} \setminus A$ and $h_0(h(\cdot))$ is the identity, and
	  
	  \item if $R(a_1,\dots,a_n) \in D$, $h_0(b_i)=a_i$ for each $i \in [n]$, and $\{ b_i \mid a_i \text{ non-isolated in } D\}$
	    is a clique in the Gaifman graph of $D_G$, then $D_G$
	    contains an atom $R(c_1,\dots,c_n)$ where $c_i=b_i$ if
	     $a_i$ is non-isolated in $D$ and $h_0(c_i)=a_i$ for
	     $i \in [n]$.
	

	  \end{enumerate}
\end{restatable}

The condition in item~(3) is not considered in~\cite{Grohe07}, but can easily be verified to hold.  It can be viewed as an ontoness requirement for the homomorphism $h_0$ from item~(1) as it says that for certain atoms $R(\bar a) \in D$, we must find certain atoms $R(\bar c) \in D_G$
with $h_0(\bar c) = \bar a$. The set $A$ is not present in Grohe's
construction (that is, he considers the case $A = \adom{D}$),
but we show in the appendix that it can be taken into account by a
minor change in the construction.

For the proof of Theorem~\ref{the:omq-main-technical-result}, we are given a recursively enumerable class of OMQs $\class{O} \subseteq (\class{G},\class{UCQ})$ such that $\class{O}\not\subseteq (\class{G},\class{UCQ})_{k}^{\equiv}$ for all $k \geq 1$. As a preliminary, we show that for every OMQ $Q \in \class{O}$, we can effectively find an equivalent OMQ $Q' \in (\class{G} \cap \class{FULL},\class{UCQ})$ over the same extended schema, where $\class{FULL}$ is the class of {\em full} TGDs, i.e., TGDs without existentially quantified variables. In other words, we can rewrite away the existential quantifiers in TGD heads. Although $Q'$ need not be in $\class{O}$, we can switch between $Q$ and $Q'$ in the proof whenever we only care about $Q$'s semantics. To simplify this proof sketch, assume that no rewriting is needed as all ontologies in $\class{O}$ are already from $\class{G} \cap \class{FULL}$, and assume that $\class{O}$ contains only Boolean OMQs where the actual queries are CQs.

The proof is by fpt-reduction from {\sf p}-{\sf Clique}, a \W-hard problem.  Assume that $G$ is an undirected graph and $k \geq 1$ a clique size.  We aim to construct a database $D_G$ and find an OMQ $Q \in \class{O}$ where 
%
%
\begin{itemize}

\item[($*$)] $G$ has a $k$-clique iff $D_G \models Q$.  

\end{itemize}
We first describe how to find $Q$. By Robertson and Seymour's Excluded Grid Theorem, there exists an $\ell$ such that every graph of treewidth exceeding $\ell$ contains the $k \times K$-grid as a minor~\cite{EGT}. We may assume, w.l.o.g., that $\ell$ exceeds the fixed arity $r$ of relation symbols that occur in $\class{O}$. Since $\class{O}\not\subseteq (\class{G},\class{UCQ})_{\ell}^{\equiv}$, there is an OMQ $Q=(\Sbf,\Sigma,q)$ from $\class{O}$ that is not in $(\class{G},\class{UCQ})^\equiv_\ell$, and we can find $Q$ by enumerating $\class{O}$ and relying on Theorem~\ref{the:ucq-k-equiv-complexity-omq}. It thus remains to construct $D_G$.  In the special case where $\Sigma$ is empty and $\Sbf$ is full, this can be done by replacing $q$ with its core and applying Theorem~\ref{thm:grohetechnew}, using $D[q]$ for $D$. Then, item~(2) of that theorem corresponds to~($*$).

In the general case, however, we cannot use $D[q]$ for $D$ since $D[q]$ might use relation symbols that are not in $\Sbf$ and thus so would the resulting $D_G$.  This, in fact, is what makes our proof rather delicate and difficult. The general idea is to replace $D[q]$ with an \Sbf-database $D_0$ such that $D_0 \models Q$ and $(D_0)^u_\ell \not\models Q$, where $(D_0)^u_\ell$ denotes the unraveling of $D_0$ into a structure of treewidth~$\ell$. Ideally, we would like to apply the Grohe construction to $D_0$ instead of $D[q]$. However, we of course still have to attain ($*$), which is formulated in terms of $Q$ (and thus, in terms of $\dep$ and $q$) rather than~$D_0$. This means that we need more from $D_0$ than just the stated properties and, in fact, we want $\chase{D_0}{\Sigma}$ and $q$ to be structurally as similar as possible. The following example illustrates a structural dissimilarity.
\begin{example}
\label{ex:sketchfirst}
  Assume that $\Sigma$ is empty, that \Sbf is full, and that $q$ is
  the $3 \times 4$-grid, that is, 
  $$
     \bigwedge_{i \in \{1,2,3\}, j \in \{1,2,3\}} X(x_{i,j},x_{i+1,j}) 
     \wedge
     \bigwedge_{i \in \{1,2,3,4\}, j \in \{1,2\}} Y(x_{i,j},x_{i,j+1}) .
  $$
  Moreover, let $D_0$ be the following database, the $3 \times 3$-grid 
  augmented with reflexive $X$-loops in the rightmost column:
  $$
  \begin{array}{l}
  \{ X(a_{i,j},a_{i+1,j}) \mid i \in \{1,2\}, j \in \{1,2,3\} \}
  \, \cup \\[1mm]
\{ Y(a_{i,j},a_{i,j+1}) \mid
i \in \{1,2,3\}, j \in \{1,2\}\}
  \, \cup \\[1mm]
  \{ X(a_{3,j},a_{3,j}) \mid j \in \{1,2,3\} \}.
  \end{array}
  $$
  Clearly, $D_0 \models Q$ and $(D_0)^u_2 \not\models Q$.
\end{example}
Example~\ref{ex:sketchfirst} is a particularly simple case: the ontology is empty, the data schema is full, and the query is a core. We could in fact choose $D_0=D[q]$. In general, however, it might not be possible to make $D_0$ and $D[q]$ isomorphic, and there might not even be some $D_0$ such that $q$ maps injectively to~$D_0$.  What we do is choosing $D_0$ such that the following holds, where $D \models^{\text{io}} p$ means that $D \models p$ and all homomorphisms from $p$ to $D$ are injective (`io' stands for `injectively only'):
\begin{itemize}

\item[($**$)]
  if $\chase{D_0}{\Sigma} \models^{\text{io}} q_c$, where $q_c$ is a contraction of $q$, then there is no $\Sbf$-database $D$ and contraction $q'_c$ of $q$ such that $D$ homorphically maps to $D_0$, $\chase{D}{\Sigma} \models^{\text{io}} q'_c$, and $q_c\neq q'_c$ is a contraction of~$q'_c$.

\end{itemize}
 In Example~\ref{ex:sketchfirst}, $D_0$ does not satisfy this condition as witnessed by choosing $D=D[q]$, where $q_c$ is the contraction of $q$ that identifies $x_{3,j}$ with $x_{4,j}$ for $j \in \{1,2,3\}$, and $q'_c=q$. 
%
At this point, $\chase{D_0}{\Sigma}$ is still not as closely related
to $q$ as we need it to be. In fact, $D_0$ might contain parts that
are superfluous for $D_0 \models Q$, and parts in which the atoms are unnecessarily entangled with each other. The following example illustrates the latter.
%
%
%
%
%
%
\begin{example}
  Assume that $q$ takes the form of an $n
  \times m$-grid that uses the binary relations $X$ and $Y$. Let $\Sbf =
  \{ X',Y' \}$ and:
  $$\Sigma = \{
  X'(x,y,z) \rightarrow X(x,y), \ Y'(x,y,z) \rightarrow Y(x,y)
  \}.
  $$
  Further, let $D_0$ be the database
  $$
  \begin{array}{l}
     \{ X'(a_{i,j},a_{i,j+1},b) \mid 1 \leq i \leq m \text{ and } 1
     \leq j < n \} \; \cup \\[1mm]
     \{ Y'(a_{i,j},a_{i+1,j},b) \mid 1 \leq i < m \text{ and } 1
     \leq j \leq n \}.
  \end{array}
  $$
  Clearly, $D_0 \models Q$. However, the following `untangled' homomorphic preimage $D_1$ of $D_0$ also satisfies  $D_1 \models Q$, and for our proof it is much preferable to work with $D_1$
  rather than with $D_0$:
  $$
  \begin{array}{l}
     \{ X'(a_{i,j},a_{i,j+1},b_{ij}) \mid 1 \leq i \leq m \text{ and } 1
     \leq j < n \} \; \cup \\[1mm]
     \{ Y'(a_{i,j},a_{i+1,j},b'_{ij}) \mid 1 \leq i < m \text{ and } 1
     \leq j \leq n \}.
  \end{array}
  $$
\end{example}
Again, the example illustrates only a very simple case. To relate
$D_0$ and $q$ even closer, we apply to $D_0$ a `diversification'
construction that replaces each atom $R(\bar a) \in D_0$ with a
(potentially empty) set of atoms $R(\bar a_1),\dots,R(\bar a_n)$,
where each $R(\bar a_i)$ can be obtained from $R(\bar a)$ by replacing
some constants with fresh isolated constants, and then attaching to
the new atoms a finite initial piece of the guarded unraveling of
$D_0$ that starts at $R(\bar a)$. The latter part is necessary to not
lose entailments of atoms through the ontology~$\Sigma$; such entailments are preserved by guarded unraveling, which like the introduction of fresh constants is part of the `untangling'.  This diversification operation is made precise in the appendix. We use as $D_G$ the database obtained by applying the Grohe construction to the result $D_1$ of diversifying $D_0$ in a maximal way such that $D_1 \models Q$. It then remains to show that ($*$) holds.

The ``only if'' direction is rather straightforward.  For the ``if'' direction, assume that $D_G \models Q$, that is, $\chase{D_G}{\Sigma} \models q$.  Thus, there is a contraction $q'_c$ of $q$ with $\chase{D_G}{\Sigma} \models^{\text{io}} q'_c$.  The composition $g$ of a witnessing homomorphism with $h_0$ shows $\chase{D}{\Sigma} \models q$, and thus there is a contraction $q_c$ of $q$ with $\chase{D}{\Sigma} \models^{\text{io}} q_c$.  Property ($**$) is preserved by diversification, and thus by putting in $D_1$ for $D_0$ and $D_G$ for $D$, we obtain $q_c = q'_c$. It follows that $g$ must be injective and thus, so is~$h_0$. We can also show that $h_0$ is `sufficiently surjective and onto' so that from the inverse of $h_0$, we can construct a homomorphism $h$ from $D_1$ to $D_G$ that satisfies the conditions from Theorem~\ref{thm:grohetechnew}, and then apply that theorem to conclude that $G$ contains a $k$-clique, as needed.

We remark that our proof is considerably more involved than the one in~\cite{BFLP19}, where ontologies are formulated in the description logic $\mathcal{ELHI}$, essentially a fragment of $\class{G}$. In particular, the presence of relations of arity larger than two makes it necessary to use diversifications, isolated constants, and condition~(iii) from Theorem~\ref{thm:grohetechnew}, while none of this is necessary for the $\mathcal{ELHI}$ case.


\subsection{Constraint Query Specifications}\label{sec:guarded-cqs-details}

We now discuss the fpt-reduction of Proposition~\ref{pro:fpt-reduction-omq-to-cqs}, which is underlying Theorem~\ref{the:cqs-guadred-technical-result}.
Recall that this reduction relies on the fact that guarded TGDs are finitely controllable. 
Recall also that the notion of finite controllability is crucial for establishing our main technical result on CQSs based on frontier-guarded TGDs (Theorem~\ref{the:cqs-frontier-guadred-technical-result}), for which more details are given in the next section. Let us then give some useful details around finite controllability.

\medskip

\noindent
\paragraph{Finite Controllability.}
Given an $\ins{S}$-database $D$, and a set $\dep$ of TGDs over $\ins{S}$, we write $\fmods{D}{\dep}$ for the set of all {\em finite} models of $D$ and $\dep$. 
Finite controllability guarantees that, for computing the set of tuples that are answers to a UCQ $q$ over every finite model of $D$ and $\dep$, i.e., the set $\bigcap_{M \in \fmods{D}{\dep}} q(M)$, it suffices to evaluate $q$ over the (possibly infinite) instance $\chase{D}{\dep}$.

%


\begin{definition}[Finite Controllability]\label{def:fc}
	 A class $\class{C}$ of TGDs is {\em finitely controllable} if, for every $\ins{S}$-database $D$, set $\dep \in \class{C}$ of TGDs over $\ins{S}$, and UCQ $q$ over $\ins{S}$, $q(\chase{D}{\dep}) = \bigcap_{M \in \fmods{D}{\dep}} q(M)$. \hfill\markfull
\end{definition}


Actually, as we shall see,  in our proofs we rely on a seemingly
stronger property than finite controllability. Roughly speaking, this
property states that, once we fix the number of variables that can
appear in UCQs, let say $n \geq 0$, given a database $D$, and a set $\dep$ of TGDs, there exists a finite model $M$ of $D$ and $\dep$, which
depends on~$n$, that allows us to compute the set of tuples
$q(\chase{D}{\dep})$, i.e., $q(\chase{D}{\dep}) =  q(M)$, no matter
what the UCQ looks like.

\begin{definition}[Strong Finite Controllability]\label{def:sfc}
	Consider a class $\class{C}$ of TGDs. We say that $\class{C}$ is {\em strongly finitely controllable} if, for every $\ins{S}$-database $D$, set $\dep \in \class{C}$ of TGDs over $\ins{S}$, and integer $n \geq 0$, there exists an instance $M(D,\dep,n) \in \fmods{D}{\dep}$ such that, for each UCQ $q$ that has at most $n$ variables, $q(\chase{D}{\dep}) = q(M(D,\dep,n))$. If there is a computable function that takes as input $D$, $\dep$, and $n$, and outputs $M(D,\dep,n)$, then we say that {\em a finite witness is realizable}. \hfill\markfull
\end{definition}

The next technical result, which is of independent interest, shows that finite and strong finite controllability are equivalent properties. 

\begin{lemma}\label{lem:fc-vs-sfc}
	Consider a class $\class{C}$ of TGDs. It holds that $\class{C}$ is finitely controllable iff $\class{C}$ is strongly finitely controllable.
\end{lemma}

By using the above lemma and results on the {\em guarded negation} fragment of first-order logic (GNFO) from~\cite{BaCS15}, we can show that:

\begin{theorem}\label{the:fg-sfc}
	The class $\class{FG}$ is strongly finitely
        controllable. Moreover, finite witnesses are realizable.
\end{theorem}

In particular, every satisfiable GNFO sentence  has a finite model. This allows us to show that the class of frontier-guarded TGDs is finitely controllable, and thus, by Lemma~\ref{lem:fc-vs-sfc}, it is also strongly finitely controllable. Moreover, we know that if a GNFO sentence $\varphi$ has a finite model, then it has a finite model of size $2^{2^{\size{\varphi}^{O(1)}}}$. This allows us to show, in addition, that finite witnesses are realizable.

%



\medskip

\noindent
\paragraph{The FPT-Reduction.}
Let us now come back to 
Proposition~\ref{pro:fpt-reduction-omq-to-cqs}.
We want to show the following: given a class of CQSs $\class{O}$ from $(\class{G},\class{UCQ})$,
there exists an fpt-reduction from \text{\rm {\sf p}-{\sf OMQ}-{\sf Evaluation}($\mathsf{omq}(\class{O})$)} to \text{\rm {\sf p}-{\sf CQS}-{\sf Evaluation}($\class{O}$)}.
Consider $Q = (\ins{S},\dep,q(\bar x))$ from $\mathsf{omq}(\class{O})$, an $\ins{S}$-database $D$, and a tuple $\bar c \in \adom{D}^{|\bar x|}$. We are going to construct an $\ins{S}$-database $D^*$ such that $D^* \models \dep$, and $\bar c \in Q(D)$ iff $\bar c \in q(D^*)$, or, equivalently (due to Proposition~\ref{pro:chase}), $\bar c \in \chase{D}{\dep}$ iff $\bar c \in q(D^*)$.
We first define $D^+$ as the database 
\[
D\ \cup\ \{R(\bar a) \in \chase{D}{\dep} \mid \bar a \subseteq \adom{D}\}.
\]
%
Let $A$ be the family of all maximal tuples $\bar a$ over $\adom{D}$
that are guarded in $D^+$, i.e., there is an atom $R(\bar b) \in D^+$ such that $\bar a \subseteq \bar b$.
Fix an arbitrary tuple $\bar a \in A$. Since, by Theorem~\ref{the:fg-sfc}, the class $\class{G}$ is strongly finitely controllable, and also finite witnesses are realizable, we can compute an instance $M(D^{+}_{|\bar a},\dep,n) \in \fmods{D^{+}_{|\bar a}}{\dep}$,
where $n$ is the number of variables in $q$, such that for each CQ $q'$ of arity $|\bar a|$ with at most $n$ variables, it holds that
\begin{itemize}
	
	\item[($*$)] $\bar a \in q'(M(D^{+}_{|\bar a},\dep,n))\ \Longrightarrow\ \bar a \in q'(\chase{D^{+}_{|\bar a}}{\dep})$.
	
\end{itemize}
%
%
W.l.o.g., we assume that $\adom{M(D^{+}_{|\bar a},\dep,n)} \cap \adom{M(D^{+}_{|\bar b},\dep,n)} \subseteq \adom{D}$, for every two distinct tuples $\bar a,\bar b \in A$. The database $D^*$ is 
\[
D^{+}\ \cup\ \bigcup_{\bar a \in A} M(D^{+}_{|\bar a},\dep,n).
\]
The next lemma shows that $D^*$ is the desired database.

\begin{lemma}\label{lem:reduction-correctness}
	It holds that:
	\begin{enumerate}
		\item $D^* \models \dep$.
		\item $\bar c \in \chase{D}{\dep}$ iff $\bar c \in q(D^*)$.
		\item There exists a computable function $f : \mathbb{N} \to \mathbb{N}$ such that $D^{*}$ can be constructed in time $\size{D}^{O(1)} \cdot f(\size{Q})$.
	\end{enumerate}
\end{lemma}

Item (1) can be shown by exploiting the fact that $\dep$ is guarded.
Concerning item (2), the $(\Rightarrow)$ direction is shown by exploiting the universality of the chase (Proposition~\ref{pro:chase}), while the $(\Leftarrow)$ direction by using the implication~($*$). 
Finally, item (3) relies on the fact that the database $D^{+}$ can be constructed in time $\size{D}^{O(1)} \cdot g(\size{Q})$ for some computable function $g : \mathbb{N} \to \mathbb{N}$, and also on the observation that the cardinality of $A$, as well as the size of $D_{|\bar a}$ for some $\bar a \in A$, do not depend on $D$.
The full proof can be found in the appendix.

\section{Frontier-guardedness}
\label{sec:frontier-guarded}
%

%
%
%

In this final section, we focus on the proof of Theorem
\ref{the:cqs-frontier-guadred-technical-result}.  Consider a
recursively enumerable class $\class{O}$ of CQSs from
$(\class{FG}_m,\class{UCQ})$ over a schema of arity $r$ such that, for
each $k \geq 1$, $\class{O} \not\subseteq
(\class{FG}_m,\class{UCQ})_{k}^{\equiv}$. Our goal is to show that
{\sf CQS}-{\sf Evaluation}($\class{O}$) is $\W$-hard.  Note that,
unlike in the proof of Theorem~\ref{the:omq-main-technical-result}, we
cannot eliminate the existential quantifiers. The reason is that in
the OMQ case we only need an equivalent OMQ while in the CQS case we
would have to replace the set of constraints $\dep$ from $\class{FG}$
with an equivalent set $\dep'$ from $\class{FG} \cap \class{FULL}$,
but clearly such a $\dep'$ need not exist.
%
%
As in the proof of Grohe's Theorem, we provide an fpt-reduction from {\sf p}-{\sf Clique} (in fact, a restricted version of it as explained later).
For the sake of clarity, here we discuss the case where $\class{O}$ consists only of CQSs $Q = (\Sigma,q)$ where $q$ is a Boolean CQ whose Gaifman graph is connected. The extension to non-Boolean UCQs consisting of non-connected CQs is treated in the appendix.

\medskip

\noindent \paragraph{A Variation of Grohe's Database.} 
For technical reasons that we clarify later, we do not use Grohe's database used in~\cite{Grohe07} to establish 
Theorem \ref{the:grohe}, nor the one that we used for showing Theorem~\ref{the:omq-main-technical-result}. Instead, we use a subset of Grohe's database that satisfies several good properties summarized in the next result.
\begin{theorem}[Grohe] 
	\label{lemma:grohe-correct-fg-body} 
	Given an undirected graph $G=(V, E)$, $k \geq 1$, databases $D,D'$ with $D \subseteq D'$, and a set $A \subseteq \adom{D}$ such that the restriction $G^{D}_{|_A}$ of the Gaifman graph of $D$
	to vertices $A$
	contains the $k \times K$-grid as a minor, 
	one can construct in time $f(k)\cdot\mn{poly}(\size{G},\size{D'})$
	a database $D^{*} = D^{*}(G,D,D',A)$ such that:
	   \begin{enumerate} 
		\item there is a surjective homomorphism $h_0$ from $D^*$ to $D'$, 
		
		\item $G$ contains a $k$-clique iff there is a homomorphism $h$ from $D$ to $D^*$ such that $h_0 (h(\cdot))$ is the identity on $A$, and
		
		\item if $D' \models \Sigma$, and every clique of size at most 
		$3 \cdot r$ in $G$ is contained in a clique of size $3 \cdot r \cdot m$, then $D^* \models \Sigma$.
	\end{enumerate} 
\end{theorem}

The condition in item (3) is not considered by Grohe in~\cite{Grohe07}, but can easily be verified to hold 
for our modified version of his database. 
It is crucial for our proof as we need to ensure that the database that the fpt-reduction constructs satisfies the 
given set $\Sigma$ of frontier-guarded TGDs. 

\medskip

\noindent \paragraph{A Crucial Lemma.}
Our proof borrows several ideas and techniques from the proof of Grohe's Theorem. However, the adaptation of such techniques is non-trivial for a couple of reasons.
First, we cannot construct an arbitrary database, but need to
construct one that satisfies the given set $\dep$ of frontier-guarded TGDs.
Second, the notion of core of a CQ, which is crucial for Grohe's proof, cannot be directly used in our context. Recall that, for $k \geq 1$, a CQ $q$ belongs to $\class{CQ}_{k}^{\equiv}$ iff its core is in $\class{CQ}_k$~\cite{DaKV02}. This allows Grohe to work with the core of $q$ instead of $q$ itself, and thus, exploit the property that a homomorphism from the core of $q$ to itself is injective. This is far from being true in the presence of constraints.
%
%
%
Instead, we need a technical result, which, intuitively speaking, states that some subsets of our CQs 
behave like cores for the sake of our proof.

\begin{restatable}{lemma}{lemmamain-body}	\label{le:main-body} 
	Fix $\ell \geq r \cdot m$. There exists a computable function that takes as input a CQS $S = (\dep,q) \in (\class{FG}_m,\class{CQ})$ that is not uniformly $\class{CQ}_\ell$-equivalent, 
	and outputs two CQs $p$ and $p'$, and a subset $X$ of the 
	variables of $p$, 
	such that the following hold:
	\begin{enumerate}
		\item $q \equiv_\Sigma p$.
		\item $D[p'] \models \Sigma$.
		\item $D[p] \subseteq D[p']$.
		\item $h(X) = X$, for every homomorphism $h$ from $p$ to $p'$. 
		\item The treewidth of $G^p_{|X}$
		is larger than $\ell$.
	\end{enumerate}
\end{restatable}


The proof of the above lemma can be found in the appendix. Let us now  briefly comment on the reason why we have 
to bound the number of head atoms in frontier-guarded TGDs for our proof to work. In very rough terms, 
this is because the way in which $p'$ is constructed is by 
generating from $p$ the finite model $M = M(D[p],\dep,n)$, for some $n \geq 0$, that is obtained from Definition~\ref{def:sfc} 
given the strong finite 
controllability of $\class{FG}$, as established in Theorem~\ref{the:fg-sfc}. 
A property of this construction that is crucial for the proof of Lemma~\ref{le:main-body} is that the treewidth of $M$ is not larger than that of $D[p]$. The only way in which we can ensure this property to hold is by fixing the number of atoms in TGD heads. 

Remarkably, this additional bound on the number of head atoms is 
not needed for obtaining Theorem \ref{the:characterization-guarded-cqss}. This 
is because there exists an fpt-reduction from the evaluation of CQSs from $(\class{G},\class{UCQ})$ to 
CQS from the same class in which there are are no existential quantifiers in TGD heads -- this is implicit in our proofs. 
When such a reduction is possible, the CQ $p'$ 
can be obtained in a simpler way: it suffices to chase $p$ using $\dep$, with $p$ being the CQ with a minimal number of atoms among all the CQs that are equivalent to $q$ w.r.t.~$\dep$. 
Notice that no bound on the number of head atoms in the TGDs of $\dep$ is required in this case. Unfortunately, when $\Sigma \in \class{FG}$, a reduction like the one described above is not possible. 
This prevents us from applying the previous idea since $\chase{p}{\dep}$ may be infinite.

\medskip

\noindent \paragraph{The FPT-Reduction.}
%
We now proceed to explain how Lemma~\ref{le:main-body} is applied in order to prove Theorem~\ref{the:cqs-frontier-guadred-technical-result}.
Let $(G,k)$ be an instance of {\sf p}-{\sf Clique}. It is easy to see that 
we can assume, w.l.o.g., that every clique of size at most $3 \cdot r$ in $G$ is contained in a clique of size $3 \cdot r \cdot m$. 
From the Excluded Grid Theorem \cite{EGT}, 
%
there is a computable integer $\ell$ such that every simple graph $G$ of treewidth at least $\ell$ 
contains a $(k \times K)$-grid as a minor.
By hypothesis on the class $\class{O}$, 
there exists a CQS $S = (\Sigma,q)$ from $\class{O}$ such that $S \not\in (\class{FG}_m,\class{CQ})_{\ell}^{\equiv}$. 
We can assume, w.l.o.g., that $\ell \geq r \cdot m$. 
We build from $(G,k)$ the tuple $(D^*,\Sigma,q)$, where $D^*$ is a database defined as follows.
%
%
%
%
Since, by assumption, $\ell \geq r \cdot m$, it is possible to compute from $q$ a 
CQ $p$,  
a subset $X$ of the 
variables of $p$, and a
CQ $p'$, that satisfy the properties stated in Lemma~\ref{le:main-body}. 
In particular, the treewidth of $G^p_{|X}$ is at least $\ell$, 
and hence, 
$G^p_{|X}$ contains the $(k \times K)$-grid as a minor. 
%
We then define $D^*$ as $D^*(G,D[p],D[p'],X)$, where $D(G,D[p],D[p'],X)$ is the database 
that is defined in Theorem \ref{lemma:grohe-correct-fg-body} for $D = D[p]$, $D' = D[p']$, and $A = X$.   


It remains to show that the above is indeed an fpt-reduction from {\sf p}-{\sf Clique} to {\sf CQS}-{\sf Evaluation}($\class{O}$). To this end, we need to show that:

\begin{restatable}{lemma}{lemmacorrectness-body}
\label{lem:fg-cqs-correcteness-body}
	The following statements hold:
\begin{enumerate} 
\item $D^* \models \Sigma$. 

\item $G$ has a $k$-clique iff $D^* \models q$. 

\item There are computable functions $f,g : \mathbb{N}\to\mathbb{N}$ such that $(D,\Sigma,q)$ 
can be constructed in time $\size{G}^{O(1)} \cdot f(k)$ and  $(\size{q} + \size{\Sigma}) \leq g(k)$.
\end{enumerate} 
\end{restatable}

%
%
%

The proof of (1) follows from the last item in Theorem~\ref{lemma:grohe-correct-fg-body}. 
We now proceed to show item (2).
%
Assume that $G$ has a $k$-clique. Then, by Theorem~\ref{lemma:grohe-correct-fg-body}, there is a homomorphism 
$h$ from $p$ to $D^*$ such that $h_0 (h(\cdot))$ is the identity. Hence 
$D^* \models p$.
%
%
%
But $q \equiv_\Sigma p$, and thus, 
$D^* \models q$
since, by item (1), we know that $D^* \models \Sigma$. 
%
Conversely, assume that 
$D^* \models q$,
and thus 
$D^* \models p$.
Then, there is a homomorphism $h$ from $p$ to $D^*$. 
It follows that $h_0 (h(\cdot))$ is a homomorphism from $p$ to $p'$.  
Consequently, by Lemma~\ref{le:main-body}, we obtain that $h(X) = X$. 
Hence, there must exist some $n \geq 0$ such that $g=h \circ (h_0 \circ h)^n$ is a homomorphism from $p$ to $D^*$ with
$h_0 (g(\cdot))$ being the identity on $X$. 
Then, $G$ has a $k$-clique from item (2) of Theorem~\ref{lemma:grohe-correct-fg-body}.
%
%

As for items (3) and (4), first notice that the CQS $S = (\Sigma,q)$ can be computed by simply enumerating the CQSs from $\class{O}$ until we find $S$ since, by Theorem~\ref{the:ucq-k-equiv-complexity-cqs}, we can check whether $S \not\in (\class{FG},\class{CQ})_{\ell}^{\equiv}$.
%
From $q$ we can construct the CQs $p$ and $p'$, as well as the set of variables $X$, by applying Lemma~\ref{le:main-body}. All these constructions 
depend only on $k$. 
Theorem~\ref{lemma:grohe-correct-fg-body}
states, on the other hand, that 
it is possible to construct $D$ in time $\size{G}^{O(1)} \cdot f'(k)$ for some computable function $f' : \mathbb{N} \to \mathbb{N}$. Putting all these together, we obtain the existence of computable functions $f,g : \mathbb{N} \to \mathbb{N}$ as needed. 
\section{Conclusions and Future Work}

We have studied the limits of efficient query evaluation in the context of ontology-mediated querying and querying in the presence of constraints, focussing on guarded and frontier-guarded TGDs, and on UCQs as the actual queries. We have obtained novel efficiency characterizations in the spirit of Grohe's well-known characterization of the tractable classes of CQs (without constraints).

An interesting open problem that emerges from this work is whether our main result on guarded OMQs (Theorem~\ref{the:fpt-characterization-guarded-omqs}) can be generalized to frontier-guarded TGDs. Recall that the notion of being equivalent to an OMQ of bounded treewdith is not enough for frontier-guarded TGDs.
We would also like to investigate whether in our main result on CQSs based on frontier-guarded TGDs (Theorem~\ref{the:characterization-fr-guarded-cqss}) the bound on the number of head atoms can be dropped.
Recall that our results, similarly to Grohe's characterization, assume that the underlying schema is of bounded arity. Establishing similar efficiency characterizations for schemas of unbounded arity is a challenging open problem, where the more involved notion of submodular width, introduced by Marx~\cite{Marx10}, should be considered.

\bibliographystyle{ACM-Reference-Format}




%
\newpage

\appendix

\section{Proof of Proposition~\ref{pro:fg-omqs-para-complexity}}

The first two items are easy, and have been already discussed in the main body of the paper. Here we focus on item (3).
The proof relies on two technical lemmas.

The first such lemma, which is implicit in~\cite{CaGL12}, states that, given an $\ins{S}$-database $D$ and an OMQ $Q = (\ins{S},\dep,q)$ from $(\class{L},\class{UCQ})$, $Q(D)$ coincides with the evaluation of $q$ over an initial finite portion $C$ of $\chase{D}{\dep}$. To formalize this, we need then notion of chase level.
Let $s = I_0 \xrightarrow{\sigma_0,\,(\bar t_0,\bar u_0)} I_1 \xrightarrow{\sigma_1,\,(\bar t_1,\bar u_1)} I_2 \dots$ be a chase sequence for a database $D$ under a set $\dep$ of TGDs, where $\sigma_i = \phi_i(\bar x,\bar y) \rightarrow \exists \bar z \, \psi_i(\bar x,\bar z)$. The {\em $s$-level} of an atom $\alpha \in \chase{D}{\dep}$ is inductively defined as follows: $\level{s}{\alpha} = 0$ if $\alpha \in D$, and $\level{s}{\alpha} = \ell$ if $\max \{\level{s}{\beta} \mid \beta \in \phi_i(\bar t_{i-1},\bar u_{i-1})\} = \ell -1$, where $i > 0$ is such that $\alpha \in I_i$ and $\alpha \not\in \bigcup_{0 \leq j < i} I_j$.
Let $\lchase{\ell}{s}{D}{\dep}$ be the instance $\{\alpha \in \bigcup_{i \geq 0} I_i \mid \level{s}{\alpha} \leq \ell\}$.
A chase sequence for $D$ under $\dep$ is {\em level-wise} if, for $i > 0$, it produces all the atoms of $s$-level $i$ before generating an atom of $s$-level $i+1$.\footnote{We keep the definition of level-wise chase sequences somewhat informal since it is clear what the underlying intention is.}

\begin{lemma}\label{lem:linear-bddp}
Consider a database $D$, a set $\dep \in \class{L}$, and a UCQ $q$. There exists a computable function $g: \mathbb{N} \rightarrow \mathbb{N}$ such that
\[\
q(\chase{D}{\dep})\ =\ q\left(\lchase{g(\size{\dep} + \size{q})}{s}{D}{\dep}\right),
\] 
where $s$ is a level-wise chase sequence for $D$ under $\dep$. Furthermore, the instance $\lchase{g(\size{\dep} + \size{q})}{s}{D}{\dep}$ can be computed in time $\size{D} \cdot f(\size{Q})$ for some computable function $f: \mathbb{N} \rightarrow \mathbb{N}$. 
\end{lemma}

\begin{proof}
	The existence of the computable function $g$, which is actually an exponential function, has been shown in~\cite{CaGL12}. It remains to show that $\lchase{g(\size{\dep} + \size{q})}{s}{D}{\dep}$ can be computed in time $\size{D} \cdot f(\size{Q})$ for some computable function $f: \mathbb{N} \rightarrow \mathbb{N}$; recall that $s$ is a level-wise chase sequence for $D$ under $\dep$.
	We first provide an upper bound on the size of $\lchase{g(\size{\dep} + \size{q})}{s}{D}{\dep}$. Let $H_\dep$ be the maximum number of atoms in the head of a TGD of $\dep$.
	
	\begin{lemma}\label{lem:chase-bound}
		It holds that
		\[
		\left|\lchase{g(\size{\dep} + \size{q})}{s}{D}{\dep}\right|\ \leq\ |D| \cdot (|\dep| \cdot H_\dep + 1)^{g(\size{\dep} + \size{q})}.
		\]
	\end{lemma}

	\begin{proof}
		We proceed by induction on the $s$-level $i \geq 0$.
		
		\medskip
		
		{\em Base case.} It is clear that $\lchase{0}{s}{D}{\dep} = D$, and the claim follows.
		
		\medskip
		
		{\em Inductive step.} Observe that, due to linearity, each TGD of $\dep$ can be triggered by an atom of $\lchase{i-1}{s}{D}{\dep}$ at most once, and generate $H_\dep$ new atoms. Hence,
		\[
		\left|\lchase{i}{s}{D}{\dep}\right|\ \leq\ \left|\lchase{i-1}{s}{D}{\dep}\right| + |\dep| \cdot \left|\lchase{i-1}{s}{D}{\dep}\right| \cdot H_\dep.
		\]
		By induction hypothesis,
		\[
		\left|\lchase{i-1}{s}{D}{\dep}\right|\ \leq\ |D| \cdot (|\dep| \cdot H_\dep +1)^{i-1}.
		\]
		Therefore,
		\begin{eqnarray*}
			&& \left|\lchase{i}{s}{D}{\dep}\right|\\
			&\leq& |D| \cdot (|\dep| \cdot H_\dep +1)^{i-1} + |\dep| \cdot |D| \cdot (|\dep| \cdot H_\dep +1)^{i-1} \cdot H_\dep\\
			&=& |D| \cdot (|\dep| \cdot H_\dep +1)^{i-1} \cdot (|\dep| \cdot H_\dep + 1)\\
			&=& |D| \cdot (|\dep| \cdot H_\dep +1)^{i},
		\end{eqnarray*}
		and the claim follows.
	\end{proof}
	
	We can now show that $\lchase{g(\size{\dep} + \size{q})}{s}{D}{\dep}$ can be computed in time $\size{D} \cdot f(\size{Q})$ for some computable function $f: \mathbb{N} \rightarrow \mathbb{N}$.
	It is easy to see that, for each $i > 0$, $\lchase{i}{s}{D}{\dep}$ can be computed from $\lchase{i-1}{s}{D}{\dep}$ in time (assume that $\dep$ is over the schema $\ins{T}$)
	\[
	|\dep| \cdot \left|\lchase{i-1}{s}{D}{\dep}\right| \cdot (\ar{\ins{T}} +1) \cdot H_\dep \cdot (\ar{\ins{T}} +1).
	\]
	Indeed, for each linear TGD $\sigma \in \dep$, we need to scan the instance $\lchase{i-1}{s}{D}{\dep}$, and check whether the atom $\body{\sigma}$ matches with an atom of $\lchase{i-1}{s}{D}{\dep}$, and, if this is the case, add to the instance under construction $H_\dep$ new atoms.
	By Lemma~\ref{lem:chase-bound}, for each $i \geq 0$,
	\begin{eqnarray*}
		\lchase{i}{s}{D}{\dep} &\leq& \lchase{g(\size{\dep} + \size{q})}{s}{D}{\dep}\\
		&\leq& |D| \cdot (|\dep| \cdot H_\dep + 1)^{g(\size{\dep} + \size{q})}.
	\end{eqnarray*}
	Thus, $\lchase{g(\size{\dep} + \size{q})}{s}{D}{\dep}$ can be computed in time
	\[
	|D| \cdot g(\size{\dep} + \size{q}) \cdot (|\dep| \cdot H_\dep)^{g(\size{\dep} + \size{q})} \cdot |\dep| \cdot H_\dep \cdot (\ar{\ins{T}}+1)^2.
	\]
	This completes the proof of Lemma~\ref{lem:linear-bddp}.
\end{proof}

The second technical lemma the we need to complete the proof of item (3) of Proposition~\ref{pro:fg-omqs-para-complexity} follows:

\begin{lemma}\label{lem:linearization}
Let $D$ be an $\ins{S}$-database, and $Q = (\ins{S},\dep,q) \in (\class{G},\class{UCQ})$. There exists a database $D^*$, and a set $\dep^* \in \class{L}$ such that
\[
Q(D)\ =\ q(\chase{D^*}{\dep^*}).
\]
Furthermore, $D^*$ can be computed in time $\size{D}^{O(1)} \cdot f(\size{Q})$, and $\dep^*$ in time $g(\size{Q})$, for some computable functions $f,g: \mathbb{N} \rightarrow \mathbb{N}$. 
\end{lemma}

Before giving the proof of the above result, let us explain how Lemma~\ref{lem:linear-bddp} and Lemma~\ref{lem:linearization} allow us to obtain item (3) of Proposition~\ref{pro:fg-omqs-para-complexity}.
Given an OMQ $Q = (\ins{S},\dep,q(\bar x)) \in (\class{G},\class{UCQ}_k)$, for $k \geq 1$, an $\ins{S}$-database $D$, and $\bar c \in \adom{D}^{|\bar x|}$, by Lemmas~\ref{lem:linear-bddp} and~\ref{lem:linearization}, we simply need to construct a finite instance $I$ in time $\size{D}^{O(1)} \cdot f(\size{Q})$ for some computable function $f: \mathbb{N} \rightarrow \mathbb{N}$, and then check if $\bar c \in q(I)$. Since $q \in \class{UCQ}_k$, by Proposition~\ref{pro:cq-bounded-tw}, the overall procedure takes time $\size{D}^{O(1)} \cdot g(\size{Q})$ for some computable function $g: \mathbb{N} \rightarrow \mathbb{N}$.

%
%




\subsection{Proof of Lemma~\ref{lem:linearization}}

\noindent The rest of this section is devoted to the proof of Lemma~\ref{lem:linearization}.
Consider an $\ins{S}$-database $D$, and an OMQ $Q = (\ins{S},\dep,q) \in (\class{G},\class{UCQ})$. We proceed to construct a database $D^*$, and a set $\dep^* \in \class{L}$, by adapting a construction from~\cite{GoMP14}, such that $Q(D) = q(\chase{D^*}{\dep^*})$. Let us first explain the high-level idea underlying $D^*$ and $\dep^*$.

\medskip

\noindent
\paragraph{The High-level Idea.} A key notion when reasoning with guarded TGDs is the {\em type} of an atom $\alpha$ (w.r.t.~$D$ and $\dep$), denoted $\type_{D,\dep}(\alpha)$, defined as the set $\{\beta \in \chase{D}{\dep} \mid \adom{\beta} \subseteq \adom{\alpha}\}$. Roughly speaking, the key property of the type is that the atoms that can be derived during the chase from an atom $\alpha$ used as a guard is determined by its type; see~\cite{CaGL12} for further details.
The main idea underlying the construction of $D^*$ is to encode an atom $\alpha \in D$ and its type as a single atom of the form $[\tau](\cdot)$, where $\tau$ is essentially a representation of the ``shape'' of $\alpha$ and its type.
For example, given an atom $R(a,b) \in D$, and assuming that its type is $\{R(a,b),S(b,a),T(a),T(b)\}$, we can encode $R(a,b)$ and its type as the atom $[R(1,2),\{S(2,1),T(1),T(2)\}](a,b)$.
When it comes to $\dep^*$, the intention is, for a TGD $\sigma \in \dep$, to encode the shape of the type $\tau$ of the guard of $\sigma$ in a predicate $[\tau]$, and then replace $\sigma$ with a linear TGD that uses in its body an atom of the form $[\tau](\cdot)$. However, we need an effective way to compute the type of an atom $\alpha$ by completing its known part, which is inherited from the type of the guard atom that generates $\alpha$, with atoms that mention the new nulls invented in $\alpha$. This relies on the main property of the type mentioned above.
In particular, for each $\alpha \in \chase{D}{\dep}$ obtained due to the application of the TGD $\sigma = \varphi(\bar x,\bar y) \ra \exists \bar z \, \psi(\bar x,\bar z) \in \dep$ with the witness $(\bar t,\bar u)$, we can construct $\type_{D,\dep}(\alpha)$ from $\psi(\bar t,\bar v)$, where $\bar v$ is a tuple of new null values, together with the restriction of $\type_{D,\dep}(\alpha)$ to the terms $\bar t$. This is precisely how we are going to generate new types from existing ones.
Before giving the formal construction, we need to introduce some auxiliary terminology.

\medskip

\noindent 
\paragraph{Auxiliary Terminology.} Let $\sch{\dep}$ be the set of predicates occurring in $\dep$. 
For an atom $\alpha$, let $\base{\alpha,\dep}$ be the set of all atoms that can be formed using terms from $\adom{\alpha}$ and predicates from $\sch{\dep}$.
A {\em $\dep$-type} $\tau$ is a pair $(\alpha,T)$, where $\alpha = R(t_1,\ldots,t_n)$ with $R \in \sch{\dep}$, $t_1 = 1$ and $t_i \in \{t_1,\ldots,t_{i-1},t_{i-1}+1\}$ for each $i \in \{2,\ldots,n\}$, and $T \subseteq \base{\{R(t_1,\ldots,t_n)\},\dep} \setminus \{\alpha\}$.
We write $\guard{\tau}$ for the atom $\alpha$, and $\atoms{\tau}$ for $(\{\alpha\} \cup T)$. The {\em arity} of $\tau$, denoted $\ar{\tau}$, is the maximum integer occurring in $\guard{\tau}$. Intuitively, $\tau$ encodes the shape of a guard atom $\alpha$ and a set of side atoms that are ``covered'' by $\alpha$.
Consider a tuple $\bar u = (u_1,\ldots,u_n)$ that is isomorphic to $\bar t = (t_1,\ldots,t_n)$, written $\bar u \simeq \bar t$, which means that $u_i = u_j$ iff $t_i = t_j$. The {\em instantiation} of $\tau$ with $\bar u$, denoted $\tau(\bar u)$, is the set of atoms obtained from $\atoms{\tau}$ after replacing each $t_i$ with $u_i$.
The {\em projection} of $\tau$ over $P \subseteq \{1,\ldots,\ar{\tau}\}$ is the set of atoms $\Pi_{P}(\tau) = \{\beta \in \atoms{\tau} \mid \adom{\beta} \subseteq P\}$.
The {\em completion} of an instance $I$ w.r.t.~$\dep$, denoted $\complete{I}{\dep}$, is defined as the instance $\{R(\bar t) \mid \bar t \in \adom{I}^{\ar{R}} \text{ and }  \{R(\bar t)\} \ra \chase{I}{\dep}\}$.

\medskip

\noindent 
\paragraph{The Formal Construction.} The new database is defined as 
\[
\begin{array}{rcl}
D^* &=& \left\{[\tau](\bar c)\ ~\left|~
\begin{array}{l}
R(\bar c) \in D\\
\tau = (R(\bar t),\cdot) \text{~~with~~} \bar c \simeq \bar t\\
\tau(\bar c) \subseteq \complete{D}{\dep}
\end{array} \right.\right\}.
\end{array}
\]

The new set $\dep^*$ of linear TGDs is the union of $\dep^{*}_{\mathit{tg}}$ and $\dep^{*}_{\mathit{ex}}$, where
\begin{enumerate}
	\item $\dep^{*}_{\mathit{tg}}$ is the {\em type generator}, i.e., is responsible for generating new $\dep$-types from existing ones.
	\item $\dep^{*}_{\mathit{ex}}$ is the {\em expander}, it expands a derived $\dep$-type $\tau$, i.e., it explicitly constructs the atoms over $\sch{\dep}$ encoded by $\tau$.
\end{enumerate}

The type generator is defined as follows. For each $\sigma \in \dep$
\[
\varphi(\bar x,\bar y)\ \ra\ \exists z_1 \cdots \exists z_m \, R_1(\bar u_1),\ldots,R_n(\bar u_n)
\]
with $\guard{\sigma} = G(\bar u)$, and $(\bar x \cup \{z_i\}_{i \in [m]})$ the variables occurring in $\head{\sigma}$, and for every $\dep$-type $\tau$ such that there is a homomorphism $h$ from $\varphi(\bar x,\bar y)$ to $\atoms{\tau}$ and $h(G(\bar u)) = \guard{\tau}$, we add to $\dep^{*}_{\mathit{tg}}$
\[
[\tau](\bar u)\ \ra\ \exists z_1 \cdots \exists z_m \, [\tau_1](\bar u_1),\ldots,[\tau_n](\bar u_n),
\]
where, having the function $f$ with 
\begin{eqnarray*}
	f(t)\
	=\ \left\{
	\begin{array}{ll}
		h(t) & \quad \text{if } t \in \bar x\\
		&\\
		\ar{\dep}+i & \quad \text{if } t = z_i,
	\end{array}\right.
\end{eqnarray*}
as well as the function $\rho$ that renames the integers in atoms in order to appear in increasing order starting from $1$ (e.g., $\rho(R(2,2,4,1)) = R(1,1,2,3)$),
for each $i \in [n]$, with $\alpha_i = R_i(f(\bar u_i))$, the $\dep$-type $\tau_i$ is
\[
\big(\rho(\alpha_i), \{\beta \in \complete{I}{\dep} \mid \adom{\beta} \subseteq \adom{\rho(\alpha_i)}\} \setminus \{\rho(\alpha_i)\}\big)
\]
with
\[
I\ =\ \rho(\{\alpha_i\}_{i \in \{1,\ldots,n\}} \cup \Pi_{\{h(x) \mid x \in \bar x\}}(\tau)).
\]

The expander it actually constructs, for each $\dep$-type $\tau$, the guard atom of $\tau$. More precisely, for each $\dep$-type $\tau$, we add to $\dep^{*}_{\mathit{ex}}$
\[
[\tau](x_1,\ldots,x_k)\ \ra\ R(x_1,\ldots,x_k),
\]
where $R$ is the $k$-ary predicate of $\guard{\tau}$.

This completes the construction of $\dep^*$. It is clear that $\dep^* \in \class{L}$. It also not difficult to show, by exploiting the properties of the type, that $Q(D) = q(\chase{D^*}{\dep^*})$, as needed.
It remains to show that $D^*$ and $\dep^*$ can be constructed in the claimed times.

\medskip

\noindent 
\paragraph{An FPT Construction.} By construction, $\dep^*$ depends only on $\dep$, and thus, it is clear that it can be computed in time $g(\size{Q})$ for some computable function $g : \mathbb{N} \rightarrow \mathbb{N}$.
On the other hand, $D^*$ depends both on $D$ and $\dep$. It remains to show that $D^*$ can be computed in time $\size{D}^{O(1)} \cdot f(\size{Q})$ for some computable function $f : \mathbb{N} \rightarrow \mathbb{N}$.
To this end, we first show the following auxiliary result:

\begin{lemma}\label{lem:GDlog-fpt}
	For a database $D'$ and set $\dep' \in \class{G} \cap \class{FULL}$, the (finite) instance $\chase{D'}{\dep'}$ can be constructed in time $\size{D'}^{O(1)} \cdot g(\size{\dep'})$ for some computable function $g : \mathbb{N} \rightarrow \mathbb{N}$.
\end{lemma}

\begin{proof}
	We provide similar analysis as in the proof of Lemma~\ref{lem:linear-bddp}.
	Let us first establish an upper bound on the size of $\chase{D'}{\dep'}$. We assume that $\dep'$ is over the schema $\ins{T}$. We can show that:
	\[
	|\chase{D'}{\dep'}|\ \leq\ |D'| \cdot |\ins{T}| \cdot \ar{\ins{T}}^{\ar{\ins{T}}}.
	\]
	Observe that, due to guardedness, two constants $c,d \in \adom{D'}$ can occur together in an atom of $\chase{D'}{\dep'}$ only if they already occur together in an atom of $D'$. Thus, for each $n$-tuple $\bar c$ occurring in $D'$ (i.e., there is an atom $R(\bar c)$ in $D'$), the chase can produce at most $n^{\ar{\ins{T}}}$ new tuples of constants  from $\bar c$, while each new such tuple can be stored in a predicate of $\ins{T}$. Since $n \leq \ar{\ins{T}}$, we get that $\chase{D'}{\dep'}$ can have at most $|D'| \cdot |\ins{T}| \cdot \ar{\ins{T}}^{\ar{\ins{T}}}$ atoms.
	
	We can now show that indeed $\chase{D'}{\dep'}$ can be constructed in time $\size{D}^{O(1)} \cdot g(\size{Q})$. 
	Let $B_{\dep'}$ (resp., $H_{\dep'}$) be the maximum number of atoms in the body (resp., head) of a TGD of $\dep'$.
	Observe that, for $i \geq 0$, $\lchase{i}{s}{D'}{\dep'}$, where $s$ is a level-wise chase sequence for $D'$ under $\dep'$, can be computed from $\lchase{i-1}{s}{D'}{\dep'}$ in time
	\begin{multline*}
	|\dep'| \cdot \left(\left|\lchase{i-1}{s}{D'}{\dep'}\right| \cdot (\ar{\ins{T}}+1) + B_{\dep'} \cdot (\ar{\ins{T}}+1) +\right.\\
	\left.\left|\lchase{i-1}{s}{D'}{\dep'}\right| \cdot B_{\dep'} \cdot (\ar{\ins{T}}+1)\right) \cdot H_{\dep'} \cdot (\ar{\ins{T}}+1).
	\end{multline*}
	Indeed, for each guarded TGD $\sigma \in \dep'$, we need to scan the instance $\lchase{i-1}{s}{D'}{\dep'}$, and check whether the atom $\guard{\sigma}$ matches with an atom $\alpha \in \lchase{i-1}{s}{D'}{\dep'}$. Note that the matching of $\guard{\sigma}$ with $\alpha$ uniquely determines the atoms $\alpha_1,\ldots,\alpha_m$ with which the atoms of $\body{\sigma}$ apart from $\guard{\sigma}$ should match. Then, for each such atom $\alpha_i$, we need to scan $\lchase{i-1}{s}{D'}{\dep'}$, and check that indeed it is present. If all the atoms $\alpha_1,\ldots,\alpha_m$ occur in $\lchase{i-1}{s}{D'}{\dep'}$, then we add $H_{\dep'}$ new atoms.
	
	As shown above, for each $i \geq 0$,
	\[
	\left|\lchase{i}{s}{D'}{\dep'}\right|\ \leq\ |\chase{D'}{\dep'}|\ \leq\ |D'| \cdot |\ins{T}| \cdot \ar{\ins{T}}^{\ar{\ins{T}}}.
	\]
	Since $\chase{D'}{\dep'}$ can have at most $n = |D'| \cdot |\ins{T}| \cdot \ar{\ins{T}}^{\ar{\ins{T}}}$ $s$-levels, it can be computed in time
	\begin{multline*}
	n \cdot |\dep'| \cdot \left(n \cdot (\ar{\ins{T}}+1) + B_{\dep'} \cdot (\ar{\ins{T}}+1) +\right.\\
	\left. n \cdot B_{\dep'} \cdot (\ar{\ins{T}}+1)\right) \cdot H_{\dep'} \cdot (\ar{\ins{T}}+1).
	\end{multline*}
	Summing up, $\chase{D'}{\dep'}$ can be computed in time
	\[
	\size{D'}^2 \cdot g(\size{\dep'})
	\]
	for some computable function $g : \mathbb{N} \rightarrow \mathbb{N}$, and the claim follows.
\end{proof}

In view of Lemma~\ref{lem:GDlog-fpt}, to show that $D^*$ can be computed in time $\size{D}^{O(1)} \cdot f(\size{Q})$, it suffices to show that we can construct in time $g(\size{\dep})$ a set $\hat{\dep} \in \class{G} \cap \class{FULL}$ such that $D^* = \chase{D}{\hat{\dep}}_{|\ins{T}}$, where $\ins{T}$ consists of all the predicates of the form $[\tau]$, where $\tau$ is a $\dep$-type, occurring in $\chase{D}{\hat{\dep}}$.
The set $\hat{\dep}$ will simply:
\begin{enumerate} 
	\item complete the database $D$ with all the atoms of $\chase{D}{\dep}$ that mention only constants from $\adom{D}$, i.e., add to $D$ the set of atoms $\chase{D}{\dep} \cap \base{D,\dep}$, and
	\item generate all the atoms of $D^*$.
\end{enumerate}

For achieving task (1) above, we exploit a result from~\cite{GoRS14}, which states the following: given a set $\dep' \in \class{G}$, we can construct a set $\xi(\dep') \in \class{G} \cap \class{FULL}$ over $\sch{\dep'}$ such that, for every database $D'$, $\chase{D'}{\xi(\dep')} = \chase{D'}{\dep'} \cap \base{D',\dep'}$.
Thus, task (1) can be done via the set $\xi(\dep) \in \class{G}$ over $\sch{\dep}$.

Task (2) can be achieved via the set of TGDs $\dep_{\mathit{types}} \in \class{G} \cap \class{FULL}$ defined as follows. For each $\dep$-type $\tau$, we have a TGD
\[
\tau(x_1,\ldots,x_k)\ \ra\ [\tau]\left(x_{f_\tau(1)},\ldots,x_{f_\tau(\ell)}\right),
\] 
where $k = \ar{\tau}$, $\ell$ is the arity of the predicate of $\guard{\tau}$, and, assuming that $(i_1,\ldots,i_\ell)$ is the tuple of $\guard{\tau}$, $f_{\tau}(i) = i_j$, for each $j \in [\ell]$. By abuse of notation, we use $\tau(x_1,\ldots,x_k)$ to denote the conjunction of atoms in the instantiation of $\tau$ with $(x_1,\ldots,x_k)$.

The set $\hat{\dep}$ is defined as $\xi(\dep) \cup \dep_{\mathit{types}}$, which can be clearly constructed in time $g(\size{\dep})$ since it depends only on $\dep$.

\section{Proof of Theorem~\ref{the:ucq-k-equiv-complexity-omq}}

Given a CQ $q$, we write $\var{q}$ for the set of variables in $q$. We may also treat, as usual, a CQ $q$ as the set of atoms in $q$.
The proofs in this section rely on notions and results
introduced in the proof of Proposition~\ref{pro:guarded-ucq-k-approximation}, which we advise the reader to read first.


The \TWOEXP-hardness is inherited from~\cite{BFLP19}, where the same problem for OMQs based on DLs has been studied.
It remains to establish the \TWOEXP~upper bound. Let $Q = (\ins{S},\dep,q)$. 
Proposition~\ref{pro:guarded-ucq-k-approximation} provides a procedure
for deciding (uniform) $\class{UCQ}_k$-equivalence: 
\begin{enumerate}
	\item compute the $\class{UCQ}_k$-approximation $Q_{k}^{a} = (\ins{S},\dep,q_{k}^{a})$ of $Q$;
	\item accept if $Q \equiv Q_{k}^{a}$; otherwise, reject.
\end{enumerate}
This yields decidability. It also shows that if $Q$ is
$\class{UCQ}_k$-equivalent, then an OMQ $Q'$ from
$(\class{G},\class{UCQ}_k)$ such that $Q \equiv Q'$ can be constructed
in double exponential time. Actually, to find $Q'$ we can simply construct
the OMQ $Q_{k}^{a}$.
\begin{lemma}\label{lem:computin-ucq-k-apx}
$Q_{k}^{a}$ can be constructed in double exponential time.
\end{lemma}
\begin{proof}
  For each disjunct $p$ of $q$, we first need to compute its
  specializations $(p',V)$. It is easy to see that there are single
  exponentially many specialization that can be found in single
  exponential time. For each specialization $s=(p',V)$, we then have
  to find all $\Sigma$-groundings $g_s(\bar x)$ of $s$, and keep the ones of treewidth at most $k$.
  Note that we can form double exponentially many guarded CQs, each of exponential size. Thus, there are double exponentially many candidates for being a $\dep$-grounding of $s$, each for exponential size, that can be found in double exponential time.
  Now, in the check whether such a candidate $g$ is indeed a $\dep$-grounding of $s$, the only non-trivial part is that we have to check for each maximally $[V]$-connected component $p_i$ of $p'[V]$ whether $p_i \rightarrow \chase{g_i}{\Sigma}$, where $g_i$ is a guarded subquery of $g$ of exponential size, via a homomorphism that is the identity on $\var{p_i} \cap V = \{y_1,\ldots,y_m\}$. This boils down to the
  problem whether $(y_1,\ldots,y_m) \in p'_i(\chase{g_i}{\dep})$, where $p'_i(y_1,\ldots,y_m)$ consists of the atoms of $p_i$. By item (2) of
  Proposition~\ref{pro:g-omqs-complexity}, this can be checked in
 \TWOEXP. Actually, this is not immediate from Proposition~\ref{pro:g-omqs-complexity} since $g_i$ might be of exponential size. We additionally need to say that the double exponential bound obtained from Proposition~\ref{pro:g-omqs-complexity} is only polynomial in the size of the database; implicit in~\cite{CaGK13}. Finally, checking whether a $\Sigma$-grounding $g_s$ of $s$ has treewidth at most $k$ is feasible in double exponential time since $g_s$ is, in general, of exponential size. 
\end{proof}

Containment between OMQs from $(\class{G},\class{UCQ})$ can be decided in \TWOEXP\xspace \cite{BaBP18}. Since, however, $q_{k}^{a}$ may consists of double exponentially many CQs, it is not clear how to
implement step~(2) above in \TWOEXP\xspace and how to obtain the
\TWOEXP\xspace upper bound in
Theorem~\ref{the:ucq-k-equiv-complexity-omq} by a direct
implementation of the above procedure. We solve this problem in two
steps. First, we replace $Q_k^{a}$ with an OMQ $Q'_k=(\ins{S},\dep',q_{k}')$ that is equivalent to $Q_k^a$, and such
that $q_k'$ contains only single exponentially many CQs, each of polynomial size. And second,
we modify the containment check from \cite{BaBP18} in a mild way.

%

\subsection{Replacing $Q^a_k$ with $Q'_k$}

The OMQ $Q'_k=(\ins{S},\dep',q_{k}')$ is defined as follows. We
introduce a fresh unary relation symbol $A$ (that is not in $\ins{S}$)
and obtain $\dep'$ from $\dep$ by extending every TGD head with the
atom $A(x)$ whenever $x$ is an existentially quantified variable in
that TGD head. Then $q'_k$ is the UCQ that includes $p_c$ for all
specializations $s=(p_c,V)$ of a CQ $p$ in $q$ such that there exists a
$\Sigma$-grounding of $s$ that is of treewidth at most $k$, extended
by adding $A(x)$ for every variable $x \in \mn{var}(p_c) \setminus V$.

It is clear that $q_k'$ has only single exponentially many disjuncts,
each of polynomial size. Reusing the arguments from the proof of
Lemma~\ref{lem:computin-ucq-k-apx}, it is also clear that $Q'_k$ can
be constructed in double exponential time.  We further make the
following important observation, which rests on the assumption that $k \geq
\mn{ar}(\Tbf)-1$.
\begin{lemma}
  \label{lem:allgroundsametw}
  Let $s=(p_c,V)$ be a specialization of a CQ in $q$. If there is a
  $\Sigma$-grounding of $s$ that has treewidth at most $k$, then
  all $\Sigma$-groundings of $s$ have treewidth at most $k$.
\end{lemma}
\begin{proof}
  Let $p_1,...,p_n$ be the maximally $[V]$-connected components of
  $p_c[V]$, let $g_s(\bar x)=\exists \bar z \, (g_0 \wedge g_1 \wedge
  \cdots \wedge g_n)$ be a $\Sigma$-grounding of $s$, and let
  $\delta=(T_\delta,\chi)$, where $T_\delta = (V_\delta,E_\delta)$, be
  a tree decomposition of $G^{g_s}_{|\bar{y}}$ of width~$k$, $\bar{y}$
  the existentially quantified variables of $g_s(\bar x)$.  For $i \in [n]$, there must be a $v_i \in V_\delta$ with $\mn{var}(p_i) \cap V \subseteq \chi(v_i)$.
  Now, let $g'_s(\bar x)=\exists \bar z \, (g_0 \wedge g'_1 \wedge
  \cdots \wedge g'_n)$ be another $\Sigma$-grounding of $s$.  Let
  $\delta'=(T_{\delta'},\chi')$ be obtained from $\delta$ by dropping
  all variables that occur in $g_s$ but not in $g'_s$ from every bag
  and adding, for $i \in [n]$, a fresh successor $u_i$ to $v_i$
  and setting $\chi'(u_i)=\mn{var}(g'_i)$. It can be verified that
  $\delta'$ is a tree decomposition of $G^{g'_s}_{|\bar{y}'}$,
  $\bar{y}'$ the existentially quantified variables of $g'_s$. In
  particular, Condition~(2) of tree decompositions is satisfied for
  every edge induced by an atom in some $g'_i$ due to the presence of
  $u_i$. Since $g'_i$ is guarded, the number of variables in
  $\mn{var}(g'_i)$ is bounded by $\mn{ar}(\Tbf)$. Since $k \geq
  \mn{ar}(\Tbf)-1$, the width of $\delta'$ is bounded by $k$, and the claim follows.
\end{proof}

We now proceed to show that $Q^a_k$ and $Q'_k$ are indeed equivalent.

\begin{lemma}
\label{lem:primeapproxequiv}
  $Q^a_k \equiv Q'_k$.
\end{lemma}
\begin{proof}
  First assume that $\bar c \in Q'_k(D)$ with $D$ an
  $\ins{S}$-database.  Then, there is a homomorphism $h$ from a CQ
  $p(\bar x)$ in $q'_k$ to $\chase{D}{\Sigma'}$ with $h(\bar x)=\bar
  c$. Assume that $p$ was constructed for the specialization
  $s=(p_c,V)$, that is, $p$~is the extension of $p_c$ with $A(x)$ for
  all $x \in \mn{var}(p_c) \setminus V$.  By construction of
  $\Sigma'$, $h$ thus maps all variables in $\mn{var}(p) \setminus V$
  to constants that the chase has introduced to satisfy existential
  quantifiers in TGD heads. Let $p_1,\dots,p_n$ be the maximally
  $[V]$-connected components of $p_c[V]$. It follows from the facts
  that $p_i$ is $[V]$-connected, and that $h$ maps all variables
  in $\mn{var}(p) \setminus V$ to existential constants that we find a
  fact $\alpha_i \in D$ such that the restriction of $h$ to the
  variables in $p_i$ is a homomorphism to
  $\chase{\type_{D,\dep'}(\alpha_i)}{\dep'}$, for $i 
  \in [n]$. Let the $\Sigma$-grounding $g_s$ of $s$ by obtained by using as $g_i$ the guarded full CQ that is obtained by viewing
  $\type_{D,\dep'}(\alpha_i)$ as a full CQ.  Clearly, the restriction
  of $h$ to the variables in $V$ can be extended to a homomorphism
  $h'$ from $g_s$ to $\chase{D}{\Sigma'}$ by mapping the variables
  from the CQs $g_i$ that are not in $V$ to constants from $\alpha_i$
  in the expected way. Since $g_s$ does not contain the relation $A$, $h'$ maps $g_s$ to $\chase{D}{\Sigma}$. It remains to argue that $g_s$ is a CQ in $q^a_k$.  Since $p$ is in $q'_k$, there is a $\Sigma$-grounding of $s$ that is of treewidth at most $k$. By
  Lemma~\ref{lem:allgroundsametw}, this implies that $g_s$ is of
  treewidth at most $k$ and thus $g_s$ is in $q^a_k$.

  Conversely, assume that $\bar c \in Q^a_k(D)$ with $D$ an
  $\ins{S}$-database. Then, there is a homomorphism $h$ from a CQ
  $g_s(\bar x)$ in $q^a_k$ to $\chase{D}{\Sigma}$ with $h(\bar x)=\bar
  c$. Let $g_s$ be a $\Sigma$-grounding of the specialization
  $s=(p_c,V)$ of a CQ $p$ in $q$. Then $g_s$ is of treewidth at most
  $k$. Consequently, $s$ gives rise to a corresponding CQ $p'$ in
  $q'_k$, that is, $p'$~is the extension of $p_c$ with $A(x)$ for all
  $x \in \mn{var}(p_c) \setminus V$. Clearly, the restriction of
  $h_{|_V}$ to the variables in $V$ is a homomorphism from $p'_{|_V}$ to
  $\chase{D}{\Sigma}$. Let $p_1,\dots,p_n$ be the maximally
  $[V]$-connected components of $p_c[V]$. By definition of
  $\Sigma$-groundings, $p_i \rightarrow \chase{g_i}{\Sigma}$ via a
  homomorphism that is the identity on $\mn{var}(p_i) \cap V$, for $i \in [n]$. We can combine all these homomorphisms with $h_{|_V}$
  to obtain a homomorphism $h'$ from $p_c$ to $\chase{D}{\Sigma}$ that
  maps all variables outside of $V$ to constants that have been
  introduced by the chase to satisfy existential quantifiers in TGD   heads. By construction of $p'$ and $\Sigma'$, $h'$ is also a
  homomorphism from $p'$ to $\chase{D}{\Sigma'}$.
\end{proof}

%

\subsection{Checking Equivalence}
We now proceed to analyze the complexity of checking whether $Q \equiv
Q_{k}'$. By item (1) of Lemma~\ref{lem:ucq-k-apx} and Lemma~\ref{lem:primeapproxequiv}, $Q_{k}' \subseteq Q$. The non-trivial task is to check whether $Q \subseteq Q_{k}'$.
We know from~\cite{BaBP18} that checking containment among OMQs from $(\class{G},\class{UCQ})$ is in~\TWOEXP. 
In fact,~\cite{BaBP18} first provides a polynomial time reduction from containment among OMQs from $(\class{G},\class{UCQ})$, denoted ${\sf Cont}(\class{G},\class{UCQ})$, to containment among OMQs from $(\class{G},\class{CQ})$, denoted ${\sf Cont}(\class{G},\class{CQ})$, and then devises an automata-based procedure for ${\sf Cont}(\class{G},\class{CQ})$ that runs in double exponential time in the combined size of the OMQs.
However, $Q_{k}'$ might be of exponential size, which implies that if we first apply the reduction from ${\sf Cont}(\class{G},\class{UCQ})$ to ${\sf Cont}(\class{G},\class{CQ})$, we get an OMQ that has a CQ of exponential size, and then, by using the decision procedure for ${\sf Cont}(\class{G},\class{CQ})$ from~\cite{BaBP18} as a black box, we only get a \THREEEXP~upper bound.
Nevertheless, it turns out that we can reuse the automata
constructions from~\cite{BaBP18} for ${\sf
  Cont}(\class{G},\class{CQ})$ in order to devise a procedure that
directly operates on OMQs from $(\class{G},\class{UCQ})$ and decides
containment in time double exponential in the size of the ontology,
the schema, and the maximum size of the CQs in the UCQs in the OMQs,
but only single exponential in the number of disjuncts in those UCQs.
Having such a procedure in place, it is then clear that checking
whether $Q \subseteq Q_{k}'$ is in \TWOEXP~since, due to
Lemma~\ref{lem:computin-ucq-k-apx}, and since although the UCQ in $Q_{k}'$ has exponentially many disjuncts, each one is of polynomial size.
For the sake of completeness, we recall the key ingredients underlying the automata-based procedure from~\cite{BaBP18}, and then explain how we get the automata-based procedure that directly operates on OMQs from $(\class{G},\class{UCQ})$ that leads to the desired upper bound.

\bigskip

\noindent \underline{\textbf{The Automata-based Procedure for ${\sf Cont}(\class{G},\class{UCQ})$}}

\medskip

\noindent Let us first recall that we can focus, w.l.o.g., on Boolean
OMQs since there is an easy reduction, originally proposed
in~\cite{BHLW16}, of containment for non-Boolean OMQs to containment
for Boolean OMQs. In what follows, we write $\class{UCQ}$ and
$\class{CQ}$ meaning the classes of Boolean UCQs and Boolean CQs,
respectively.

We know from~\cite{BaBP18} that, given two OMQs $Q_1$ and $Q_2$ from
$(\class{G},\class{CQ})$, if $Q_1 \not\subseteq Q_2$, then this is witnessed via a ``nearly acyclic'' database. To formalize this we need the notion of guarded tree decomposition.
A tree decomposition $\delta = (T_\delta,\chi)$, where $T_\delta =
(V_\delta,E_\delta)$, of $G^D$ for some database $D$ is called {\em
  $[V]$-guarded}, where $V \subseteq V_\delta$, if for every node $v
\in V_\delta \setminus V$, there exists a fact $R(c_1,\ldots,c_n) \in D$ such that $\chi(v) \subseteq \{c_1,\ldots,c_n\}$, i.e., $\chi(v)$ is guarded in $D$. We write $\rt{\delta}$ for the root node of $\delta$, and $D_{\delta}(v)$, for $v \in V_\delta$, for the subset of $D$ induced by $\chi(v)$.
An $\ins{S}$-database is called a  {\em $C$-tree}, where $C \subseteq D$, if there exists a tree decomposition $\delta$ of $G^D$ such that: (i) $D_{\delta}(\rt{\delta}) = C$, and (ii) $\delta$ is $[\{\rt{\delta}\}]$-guarded.
Roughly, if  a database $D$ is a $C$-tree, then $C$ can be viewed as
the cyclic part of $D$, while the rest of $D$ is
acyclic. In~\cite{BaBP18}, it has been shown that for containment
purposes we can focus on databased that are $C$-trees with the size of $C$ bounded and depending only on the left-hand side OMQ. Note, however, that the result of~\cite{BaBP18} talks only about OMQs from $(\class{G},\class{CQ})$. A careful inspection of the proof given in~\cite{BaBP18} reveals that the following statement for OMQs from $(\class{G},\class{UCQ})$ holds:

\begin{lemma}\label{lem:cont-tree-witness}
	Let $Q_1 = (\ins{S},\dep_1,q_1)$ and $Q_2 = (\ins{S},\dep_2,q_2)$ be OMQs from $(\class{G},\class{UCQ})$. The following are equivalent:	
	\begin{enumerate}
		\item $Q_1 \not\subseteq Q_2$.
		\item There exists a $C$-tree $\ins{S}$-database $D$ with $\size{\adom{C}} \leq \ar{\ins{S} \cup \sch{\dep_1}} \cdot \max_{p \in q_1} \{\size{p}\}$ such that $Q_1(D) \not\subseteq Q_2(D)$.
	\end{enumerate}
\end{lemma}

The goal is to devise a two-way alternating parity tree automaton (2ATA) that checks for the second item of Lemma~\ref{lem:cont-tree-witness}. To this end, we need a convenient encoding of a $C$-tree database as a tree over a finite alphabet that can be accepted by a 2ATA. As shown in~\cite{BaBP18}, this can be done by exploiting standard encodings of bounded treewidth databases. In a nutshell, a $C$-tree $\ins{S}$-database $D$ can be encoded as a finite  $\Gamma_{\ins{S},\ell}$-labeled tree, with $\size{\adom{C}} \leq \ell$ and $\Gamma_{\ins{S},\ell}$ being an alphabet of double exponential size in $\ar{\ins{S}}$, and exponential in $\size{\ins{S}}$ and $\ell$; the actual encoding is not  important for our discussion.
Of course, it is possible that a $\Gamma_{\ins{S},\ell}$-labeled tree does not encode a $C$-tree database. Thus, some additional consistency properties are needed, which are also not important for our discussion. A $\Gamma_{\ins{S},\ell}$-labeled tree is {\em consistent} if it satisfies those syntactic properties, and it can be shown that a consistent $\Gamma_{\ins{S},\ell}$-labeled tree $L$ can be decoded into an $\ins{S}$-database, denoted $[L]$, that is a $C$-tree with $\size{\adom{C}} \leq \ell$. Therefore, by Lemma~\ref{lem:cont-tree-witness}, we get the following:

\begin{lemma}\label{lem:cont-tree-witness-enc}
	Let $Q_1 = (\ins{S},\dep_1,q_1)$ and $Q_2 = (\ins{S},\dep_2,q_2)$ be OMQs from $(\class{G},\class{UCQ})$. The following are equivalent:
	\begin{enumerate}
		\item $Q_1 \not\subseteq Q_2$.
		\item There exists a consistent $\Gamma_{\ins{S},\ell}$-labeled tree $L$, where $\ell = \ar{\ins{S} \cup \sch{\dep_1}} \cdot \max_{p \in q_1} \{\size{p}\}$, such that $Q_1([L]) \not\subseteq Q_2([L])$.
	\end{enumerate}
\end{lemma}

It should be clear now that the goal is to devise a 2ATA $A$ over finite trees such that its language, denoted $L(A)$, is the set of $\Gamma_{\ins{S},\ell}$-labeled tree $L$, where $\ell \leq \ar{\ins{S} \cup \sch{\dep_1}} \cdot \max_{p \in q_1} \{\size{p}\}$, such that $Q_1([L]) \not\subseteq Q_2([L])$.
Having such an automaton in place, checking whether $Q_1 \subseteq Q_2$ boils down to checking whether $L(A)$ is empty, which can be done in exponential time in the number of states of $A$, and in polynomial time in the size of the input alphabet. Thus, to obtain the desired \TWOEXP~upper bound, we need to construct $A$ in double exponential time, while the number of its states should be at most exponential.
To construct the desired 2ATA we are going to exploit the automata constructions from~\cite{BaBP18}.

First, we need a way to check consistency of labeled trees, which, as shown in~\cite{BaBP18}, can be done via an automaton.

\begin{lemma}\label{lem:consisteny-2ata}
Consider a schema $\ins{S}$, and an integer $\ell > 0$. There is a 2ATA $C_{\ins{S},\ell}$ that accepts a $\Gamma_{\ins{S},\ell}$-labeled tree $L$ iff $L$ is consistent. The number of states of $C_{\ins{S},\ell}$ is logarithmic in the size of $\Gamma_{\ins{S},\ell}$. Furthermore, $C_{\ins{S},\ell}$ can be constructed in polynomial time in the size of $\Gamma_{\ins{S},\ell}$.
\end{lemma}

We also know from~\cite{BaBP18} that, given an OMQ $Q$ from $(\class{G}, \class{CQ})$, we can devise an automaton that accepts labeled trees which correspond to databases that make $Q$ true.

\begin{lemma}\label{lem:evaluation-2ata}
Let $Q = (\ins{S},\dep,q) \in (\class{G}, \class{CQ})$. There is a 2ATA $A_{Q,\ell}$, where $\ell > 0$, that accepts a consistent $\Gamma_{\ins{S},\ell}$-labeled tree $L$ iff $[L] \models Q$. $A_{Q,\ell}$ has exponentially many states in $\size{Q}$ and $\ell$, and it can be constructed in double exponential time in $\size{Q}$ and $\ell$.
\end{lemma}

We are now ready to devise the desired 2ATA. To this end, we are going
to exploit Lemmas~\ref{lem:cont-tree-witness-enc}
to~\ref{lem:evaluation-2ata}, and some well-known facts about
2ATAs. In particular, languages accepted by 2ATAs are closed under
union and complementation. Given two 2ATAs $A_1$ and $A_2$, we write
$A_1 \cup A_2$ 
or a 2ATA, which can be constructed
in polynomial time, that accepts the language $L(A_1) \cup L(A_2)$.
Moreover, we write $\overline{A}$ for the 2ATA, which can also be
constructed in polynomial time, that accepts the complement of $L(A)$.
Importantly, the number of states of $A_1 \cup A_2$ is the sum of the
numbers of states of $A_1$ and $A_2$ plus one.
Given $Q_1 = (\ins{S},\dep_1,q_1)$ and $Q_2 = (\ins{S},\dep_2,q_2)$ from $(\class{G},\class{UCQ})$, we write $Q_{1,2}^{\max}$ for the largest OMQ among $\bigcup_{p \in q_1} (\ins{S},\dep_1,p) \cup \bigcup_{p \in q_2} (\ins{S},\dep_2,p)$.

\begin{lemma}\label{lem:final-construction-2ata}
	Let $Q_1 = (\ins{S},\dep_1,q_1)$ and $Q_2 = (\ins{S},\dep_2,q_2)$ be OMQs from $(\class{G},\class{UCQ})$. There is a 2ATA $A$ such that
	\[
	Q_1 \subseteq Q_2\ \Longleftrightarrow\ L(A) = \emptyset. 
	\]
	$A$ has exponentially many states in $\size{Q_{1,2}^{\max}}$
        and $\ell$, and polynomially many states in $\size{q_1} +
        \size{q_2}$. It can be constructed in double exponential time
        in $\size{Q_{1,2}^{\max}}$ and $\ell$, and in polynomial time
        in $\size{q_1} + \size{q_2}$.
\end{lemma}


\begin{proof}
Let $\ell = \ar{\ins{S} \cup \sch{\dep_1}} \cdot \max_{p \in q_1} \{\size{p}\}$. The 2ATA $A$ is
\[
\left(C_{\ins{S},\ell} \cap \left(\bigcup_{p \in q_1} A_{(\ins{S},\dep_1,p),\ell}\right)\right)\ \cap\ \left(\overline{\bigcup_{p \in q_2} A_{(\ins{S},\dep_2,p),\ell}}\right).
\]
Using
Lemmas~\ref{lem:cont-tree-witness-enc},~\ref{lem:consisteny-2ata},
and~\ref{lem:evaluation-2ata}, the fact that union and complementation
of 2ATA are feasible in polynomial time, and the stated bound on the
number of states in automata for union, it is easy to verify that
indeed $A$ is the desired 2ATA.
\end{proof}

It is now easy to see, having Lemma~\ref{lem:final-construction-2ata}
in place, that indeed checking whether $Q \subseteq Q_{k}'$ is
feasible in double exponential time. This actually follows from the
fact that the UCQ of $Q_{k}'$ has exponentially many disjuncts, each of polynomial size, and the fact that emptiness for 2ATA can be decided in exponential time in the number of states of the automaton.

\section{Proof of Proposition~\ref{pro:guarded-ucq-k-approximation}}


$\class{UCQ}_k$-approximations for OMQs from
$(\class{G},\class{UCQ})$ rely on the notion of specialization of a
CQ. We first introduce CQ specializations, and give a lemma
that illustrates their usefulness. We then define the notion of
grounding of a CQ specialization, and establish a technical lemma that
will allow us to define $\class{UCQ}_k$-approximations. We then
establish our main technical result about
$\class{UCQ}_k$-approximations, which in turn allows us to complete the proof of Proposition~\ref{pro:guarded-ucq-k-approximation}.
In what follows, given a CQ $q$, we write $\var{q}$ for the set of variables in $q$. We may also treat, as usual, a CQ $q$ as the set of atoms in $q$.


%

\subsection{CQ Specializations}
Recall that a contraction of a CQ $q(\bar x)$ is a CQ $p(\bar x)$ that can
be obtained from $q$ by identifying variables. When an answer variable $x$ is identified with a non-answer variable $y$, the resulting variable is $x$; the identification of two answer variables is not allowed.
The formal definition of CQ specialization follows:

\begin{definition}\label{def:specialization}
	Consider a CQ $q(\bar x)$. A {\em specialization} of $q$ is a pair $s = (p,V)$, where $p$ is a contraction of $q$, and $\bar x \subseteq V \subseteq \var{p}$. \hfill\markfull
\end{definition}

Intuitively speaking, a specialization of a CQ $q$ describes a way how
$q$ can be homomorphically mapped to the chase of a database $D$ under a set $\dep$ of
TGDs. First, instead of mapping $q$ to $\chase{D}{\dep}$, we
injectively map $p$ to $\chase{D}{\dep}$, while the set of variables
$V$ collects all the variables of $p$ that are mapped to constants of
$\adom{D}$.
The next result, which relies on the properties of the type of an atom
(see the proof of Lemma~\ref{lem:linear-bddp} in Appendix~\ref{app:facets})
illustrates the usefulness of CQ specializations. 
A CQ $q$ is {\em $[V]$-connected}, for $V \subseteq \var{q}$, if $G^{q}_{|\var {q} \setminus V}$, that is, the subgraph of $G^q$ induced by $\var {q} \setminus V$, is connected.
We also write $q[V]$ for the subquery of $q$ obtained after dropping the atoms of $q_{|V}$, i.e., the atoms that contain only variables of $V$.
%
%
Given a database $D$, and a set $\dep$ of TGDs, we write $\gchase{D}{\dep}$ for the ground part of $\chase{D}{\dep}$, that is, the set consisting of all atoms in $\chase{D}{\dep}$ with only constants from $\adom{D}$.

\begin{lemma}\label{lem:specialization-lemma}
	Consider a database $D$, a set $\dep \in \class{G}$, a CQ $q(\bar x)$, and a tuple $\bar c \in \adom{D}^{|\bar x|}$. The following are equivalent:
	\begin{enumerate}
		\item $\bar c \in q(\chase{D}{\dep})$.
		\item There exists a specialization $s = (p,V)$ of $q$ such that $p \ra \chase{D}{\dep}$ via an injective homomorphism $\mu$ with (i) $\mu(\bar x) = \bar c$, (ii) $\mu(p_{|V}) \subseteq \gchase{D}{\dep}$, and (iii) for every maximally $[V]$-connected component $p'$ of $p[V]$, there exists an atom $\alpha \in D$ such that $\adom{\alpha} \supseteq \{\mu(y) \mid y \in \var{p'} \cap V\}$ and $\mu(p') \subseteq \chase{\type_{D,\dep}(\alpha)}{\dep}$.
	\end{enumerate}
\end{lemma}


%

\subsection{Grounding Specializations}
We now define the notion of grounding of a specialization $(p,V)$ of a CQ $q$. The intention is to replace each maximally $[V]$-connected component $p'$ of $p[V]$ with a guarded set of atoms, which is ought to be mapped to the ground part of the chase, that entails $p'$ under the given set of TGDs.
%
%
A CQ $q$ is {\em guarded} if it has an atom, denoted $\guard{q}$, that contains all the variables of $\var{q}$, and {\em full} if all the variables of $\var{q}$ are answer variables.

\begin{definition}\label{def:grouding-squid-decomposition}
	Consider a set $\dep$ of TGDs over $\ins{T}$, and a CQ $q(\bar x)$ over $\ins{T}$. Let $s = (p,V)$ be a specialization of $q$ with $p_1,\ldots,p_n$ being the maximally $[V]$-connected components of $p[V]$.
	A {\em $\dep$-grounding} of $s$ is a CQ $g_s(\bar x)$ over $\ins{T}$ of the form
	\[
	\exists \bar z \, \left(g_0 \wedge g_1 \wedge \cdots \wedge g_n\right)
	\]
	where 
	\begin{itemize}
		\item $g_0$ is the full CQ consisting of the atoms of $p_{|V}$,
		\item for each $i \in [n]$, $g_i$ is a guarded full CQ such that:
		\begin{itemize}
			\item $\var{g_i} \subseteq (\var{p_i} \cap V) \cup \left\{y^{i}_{1},\ldots,y^{i}_{\ar{\ins{T}}-m}\right\} \subset \ins{V}$,
			\item $\var{p_i} \cap V \subseteq \var{\guard{g_i}}$, and
			\item $p_i \ra \chase{g_i}{\dep}$ via a homomorphism that is the identity on $\var{p_i} \cap V$, and
		\end{itemize}
		\item $\bar z = \left(\bigcup_{0 \leq i \leq n} \var{g_i}\right) \setminus \bar x$. \hfill\markfull
	\end{itemize}
\end{definition}

Before giving the main lemma concerning groundings of CQ specializations, let us establish a useful auxiliary claim.

\begin{claim}\label{cla:groundings}
	Let $\dep$ be a set of TGDs over $\ins{T}$, and $q(\bar x)$ a CQ over $\ins{T}$. Let $g_s(\bar x)$ be a $\dep$-grounding of a specialization $s = (p,V)$ of $q$, and assume that $g_s \ra I$, for an instance $I$, via a homomorphism $\mu$. Then, $p \ra \chase{I}{\dep}$ via a homomorphism $\xi$ with $\mu(\bar x) = \xi(\bar x)$.
\end{claim}

\begin{proof}
	Let $p_1,\ldots,p_n$ be the maximally $[V]$-connected components of $p[V]$. Assume that $g_s(\bar x)$ is of the form
	\[
	\exists \bar z \, (g_0 \wedge g_1 \wedge \cdots \wedge g_n), 
	\]
	where $g_0$ is the full CQ consisting of the atoms of $p_{|V}$, and for each $i \in [n]$, $g_i$ is a guarded full CQ with the properties given in Definition~\ref{def:grouding-squid-decomposition}.
	Clearly, for each $i \in [n]$, there exists a mapping $\lambda_i$ that is the identity on $\var{p_i} \cap V$ that maps $p_i$ to $\chase{g_i}{\dep}$. Since the sets of atoms $p_{|V},p_1,\ldots,p_n$ share only variables of $V$, we conclude that $\lambda = \bigcup_{i \in [m]} \lambda_i$ is a well-defined mapping that maps $p$ to $\chase{g_s}{\dep}$ with $\lambda(\bar x) = \bar x$. 
	Since, by hypothesis, $g_s \ra I$ via some $\mu$, it is not difficult to show that $\chase{g_s}{\dep}$ maps to $\chase{I}{\dep}$ via $\mu'$ that extends $\mu$, and thus, $\mu(\bar x) = \mu'(\bar x)$. Therefore, $\xi = \mu' \circ \lambda$ maps $p$ to $\chase{I}{\dep}$ with $\mu(\bar x) = \xi(\bar x)$, as needed.
\end{proof}

We are now ready to show the main lemma about groundings of CQ specializations. 

\begin{lemma}\label{lem:groudnings-mail-lemma}
	Consider a database $D$, a set $\dep \in \class{G}$ over a schema $\ins{T}$, a CQ $q(\bar x)$ over $\ins{T}$, and $\bar c \in \adom{D}^{|\bar x|}$. The following are equivalent:
	\begin{enumerate}
		\item $\bar c \in q(\chase{D}{\dep})$.
		\item There exists a specialization $s$ of $q$ such that a $\dep$-grounding $g_s(\bar x)$ of $s$ maps to $\gchase{D}{\dep}$ via an injective mapping $\mu$ with $\mu(\bar x) = \bar c$.
	\end{enumerate}
\end{lemma}

\begin{proof}
	\underline{$(1) \Rightarrow (2).$} By Lemma~\ref{lem:specialization-lemma}, there is a specialization $s = (p,V)$ of $q(\bar x)$, with $p_1,\ldots,p_n$ being the maximally $[V]$-connected components of $p[V]$, such that $p \ra \chase{D}{\dep}$ via an injective mapping $\lambda$
	with (i) $\lambda(\bar x) = \bar c$, (ii) $\lambda(p_{|V}) \subseteq \gchase{D}{\dep}$, and (iii) for every maximally $[V]$-connected component $p'$ of $p[V]$, there exists an atom $\alpha \in D$ such that $\adom{\alpha} \supseteq \{\lambda(y) \mid y \in \var{p'} \cap V\}$ and $\lambda(p') \subseteq \chase{\type_{D,\dep}(\alpha)}{\dep}$.
	%
	%
	Consider the CQ $g_s(\bar x)$ of the form 
	\[
	\exists \bar z \, \left(g_0 \wedge g_1 \wedge \cdots \wedge g_n\right),
	\] 
	where $g_0$ is the full CQ consisting of the atoms of $p_{|V}$, and, for each $i \in [n]$, $g_i$ is the guarded full CQ obtained from $\type_{D,\dep}(\alpha_i)$ by converting $\lambda(y)$, where $y \in \var{p_i} \cap V$, into the variable $y$, and each constant $c^{i}_{j} \in \{c^{i}_{1},\ldots,c^{i}_{\ell_i}\} = \adom{\type_{D,\dep}(\alpha_i)} \setminus \{\lambda(y) \mid y \in \var{p_i} \cap V\}$ into the variable $y^{i}_{j}$. Finally, let $\bar z = \left(\bigcup_{0 \leq i \leq n} \var{g_i}\right) \setminus \bar x$.
	Clearly, $g_s(\bar x)$ is a $\dep$-grounding of $s$. Moreover, the mapping $\mu = \lambda \cup \{y^{i}_{j} \mapsto c^{i}_{j}\}_{i \in [n],j \in [\ell_i]}$ maps injectively $g_s$ to $\gchase{D}{\dep}$, since $\type_{D,\dep}(\alpha_i) \subseteq \gchase{D}{\dep}$, with $\mu(\bar x) = \bar c$.
	
	\medskip
	
	\underline{$(2) \Rightarrow (1).$} Consider a specialization $s = (p,V)$ of $q$.
	Assume that there exists a $\dep$-grounding $g_s(\bar x)$ of $s$ that maps to $\gchase{D}{\dep}$ via an injective mapping $\mu$ with $\mu(\bar x) = \bar c$; actually, the fact that $\mu$ is injective is not crucial here. 
	By Claim~\ref{cla:groundings}, we get that $p$ maps to
	\[
	\chase{\gchase{D}{\dep}}{\dep}\ =\ \chase{D}{\dep}
	\] 
	via a homomorphism $\xi$ with $\xi(\bar x) = \bar c$. Clearly, by the definition of contractions, there exists a homomorphism $h$, which is the identity on $\bar x$, that maps $q$ to $p$. Therefore, $\lambda = \xi \circ h$ maps $q$ to $\chase{D}{\dep}$ with $\lambda(\bar x) = \bar c$, which implies that $\bar c \in q(\chase{D}{\dep})$.
\end{proof}


%

\subsection{$\class{UCQ}_k$-approximations}
We proceed to introduce the notion of $\class{UCQ}_k$-approximation for OMQs from $(\class{G},\class{UCQ})$, which relies on the notion of grounding of a CQ specialization.
Roughly, the $\class{UCQ}_k$-approximation of an OMQ $Q = (\ins{S},\dep,q)$, for $k \geq 1$, is the OMQ $Q_{k}^{a}$ obtained from $Q$ by replacing each disjunct $p$ of $q$ with the UCQ consisting of all the $\dep$-groundings of all specializations of $p$ with treewidth at most $k$.

\begin{definition}\label{def:ucq-k-apx-via-squids}
	Consider an OMQ $Q = (\ins{S},\dep,q)$ from $(\class{G},\class{UCQ})$ over a schema $\ins{T}$. The {\em $\class{UCQ}_k$-approximation} of $Q$, for $k \geq 1$, is the OMQ $Q_{k}^{a} = (\ins{S},\dep,q_{k}^{a})$, where $q_{k}^{a}$ is 
	obtained from $q$ by replacing each disjunct $p$ of $q$ with the UCQ consisting of all the $\dep$-groundings over $\ins{T}$ of treewidth at most $k$ of all specializations of $p$. \hfill\markfull
\end{definition}

The main technical lemma concerning $\class{UCQ}_k$-approximations
follows. This relies on the notion of $k$-unraveling of a database $D$
up to a certain tuple $\bar c$ over $\adom{D}$, which can be defined
in the same way as in~\cite{BFLP19}. In particular, the {\em $k$-unraveling of $D$ up to $\bar c$}, for some $k \geq
\ar{\ins{S}}-1$, denoted $D_{\bar c}^{k}$, is a possibly infinite
$\ins{S}$-instance for which the following properties hold:
\begin{enumerate}
	\item The subgraph of $G^{D_{\bar c}^{k}}$ induced by the elements of $\adom{D_{c}^{k}} \setminus \bar c$ has treewidth $k$.
	
	\item $D_{\bar c}^{k} \ra D$ via a homomorphism that is the identity on $\bar c$.
	
	\item For every OMQ $Q$ from $(\class{G},\class{UCQ}_k)$ with data schema $\ins{S}$, $\bar c \in Q(D)$ implies $\bar c \in Q(D_{\bar c}^{k})$.
\end{enumerate}
Let us stress that item~(3) relies on the assumption that $k \geq
\ar{\ins{S}}-1$ because, otherwise, there can be single atoms in $D$ that
cannot be part of any database of treewidth $k$.
In what follows, we say that $D$ has {\em treewidth $k$ up to $\bar c$} if the subgraph of $G^D$ induced by the elements of $\adom{D} \setminus \bar c$ has treewidth $k$.

\begin{lemma}\label{lem:ucq-k-apx}
	Let $Q$ be an OMQ from $(\class{G},\class{UCQ})$ over $\ins{T}$ with data schema $\ins{S} \subseteq \ins{T}$, and  $Q_{k}^{a}$ its $\class{UCQ}_k$-approximation for $k \geq \ar{\ins{T}}-1$.
	\begin{enumerate}
		\item $Q_{k}^{a} \subseteq Q$.
		
		
		\item For every $\ins{S}$-database $D$ and tuple $\bar c$ over $\adom{D}$ such that $D$ has treewidth at most $k$ up to $\bar c$, $\bar c \in Q(D)$ implies $\bar c \in Q_{k}^{a}(D)$.
				
		\item For every OMQ $Q' \in (\class{G},\class{UCQ}_k)$ with data schema $\ins{S}$ such that $Q' \subseteq Q$, it holds that $Q' \subseteq Q_{k}^{a}$.
	\end{enumerate}
\end{lemma}

\begin{proof}
	Let $Q = (\ins{S},\dep,q(\bar x))$, and 
	$Q_{k}^{a} = (\ins{S},\dep,q_{k}^{a}(\bar x))$. 
	
	\medskip
	
	\underline{Item (1).} Consider an $\ins{S}$-database $D$, and a tuple $\bar c \in \adom{D}^{|\bar x|}$. Assume that $\bar c \in Q_{k}^{a}(D)$, or, equivalently, there is a disjunct $p_{k}^{a}$ of $q_{k}^{a}$ that maps to $\chase{D}{\dep}$ via a homomorphism $\mu$ with $\mu(\bar x) = \bar c$. 
	By definition, $p_{k}^{a}$ is a $\dep$-grounding of a specialization $s = (p_c,V)$ of some disjunct $p$ of $q$. By Claim~\ref{cla:groundings}, $p_c$ maps to
	\[
	\chase{\chase{D}{\dep}}{\dep}\ =\ \chase{D}{\dep}
	\] 
	via a homomorphism $\mu'$ with $\mu'(\bar x) = \bar c$.
	Since, by definition of contractions, there exists a homomorphism $h$, which is the identity on $\bar x$, that maps $p$ to $p_c$, we get that $\lambda = \mu' \circ h$ maps $p$ to $\chase{D}{\dep}$ with $\lambda(\bar x) = \bar c$. Since $p$ is a disjunct of $q$, we get that $\bar c \in Q(D)$.
	
	\medskip
	
	\underline{Item (2).} Consider an $\ins{S}$-database $D$, and a tuple $\bar c \in \adom{D}^{|\bar x|}$ such that $D$ has treewidth at most $k$ up to $\bar c$. Assume that $\bar c \in Q(D)$. Thus, there exists a disjunct $p$ of $q$ such that $\bar c \in p(\chase{D}{\dep})$. 
	By Lemma~\ref{lem:groudnings-mail-lemma}, there is a specialization $s$ of $p$ such that a $\dep$-grounding $g_s(\bar x)$ of $s$ maps to $\gchase{D}{\dep}$ via an injective mapping $\mu$ with $\mu(\bar x) = \bar c$. Since $D$ has treewidth at most $k$ up to $\bar c$, also $\gchase{D}{\dep}$ has treewdith at most $k$ up to $\bar c$. Consequently, $g_s$ has treewidth at most $k$, and thus is a disjunct of $q_{k}^{a}$. This implies that $\bar c \in {Q_{k}^{a}}(D)$.
	
	\medskip
	
	\underline{Item (3).} Let $Q' = (\ins{S},\dep',q')$. Consider an $\ins{S}$-database $D$ and tuple of constants $\bar c \in \adom{D}^{|\bar x|}$, and assume that $\bar c \in Q'(D)$. Since $Q'$ belongs to $(\class{G},\class{UCQ}_k)$, we get that $\bar c \in Q'(D_{\bar c}^{k})$, and thus, $\bar c \in Q(D_{\bar c}^{k})$; recall that $D_{\bar c}^{k}$ is the $k$-unraveling of $D$ up to $\bar c$. From item (2), we get that $\bar c \in Q_{k}^{a}(D_{\bar c}^{k})$, or, equivalently, there is $p \in q_{k}^{a}$ that maps to $\chase{D_{\bar c}^{k}}{\dep}$ via a homomorphism $\mu$ with $\mu(\bar x) = \bar c$.
	We also know that $D_{\bar c}^{k}$ maps to $D$ via a homomorphism $\lambda$ with $\lambda(\bar c) = \bar c$, which allows us to show that $\chase{D_{\bar c}^{k}}{\dep}$ maps to $\chase{D}{\dep}$ via an extension $\lambda'$ of $\lambda$. Therefore, $\xi = \lambda' \circ \mu$ maps $p$ to $\chase{D}{\dep}$ with $\xi(\bar x) = \bar c$, and thus, $\bar c \in Q_{k}^{a}(D)$, as needed.
\end{proof}

%

\subsection{Finalizing the Proof of Proposition~\ref{pro:guarded-ucq-k-approximation}}
Having Lemma~\ref{lem:ucq-k-apx} in place, it is now easy to establish Proposition~\ref{pro:guarded-ucq-k-approximation}. Observe that $(3) \Rightarrow (2)$ and $(2) \Rightarrow (1)$ hold trivially. It remains to show that $(1) \Rightarrow (3)$. 
By hypothesis, there exists an OMQ $Q'$ from $(\class{G},\class{UCQ}_k)$ such that $Q \equiv Q'$. By item (3) of Lemma~\ref{lem:ucq-k-apx}, $Q' \subseteq Q_{k}^{a}$, which implies that $Q \subseteq Q_{k}^{a}$. By item (1) of Lemma~\ref{lem:ucq-k-apx}, we also get that $Q_{k}^{a} \subseteq Q$, and the claim follows.


\subsection{The Case $k < \ar{\ins{T}}-1$}
Notice that Theorem~\ref{the:ucq-k-equiv-complexity-omq} and
Proposition~\ref{pro:guarded-ucq-k-approximation} only cover the case
where $k \geq \ar{\ins{T}}-1$, and also Lemma~\ref{lem:ucq-k-apx},
which gives the main technical properties of $\class{UCQ}_k$-approximations,
requires this condition. We view the case $k < \ar{\ins{T}}-1$ as a
somewhat esoteric corner case. In the context of
Theorem~\ref{the:ucq-k-equiv-complexity-omq}, for example, we ask for
an equivalent OMQ from $(\class{G},\class{UCQ}_k)$ when the arity of
some relations $R$ is so high that the actual UCQ in such an OMQ
cannot contain a fact $R(x_1,\dots,x_n)$ unless some of the $x_i,x_j$
are identical.
We show that this case is also technically very different. In particular, our $\class{UCQ}_k$-approximations do not serve their purpose.

\medskip
We consider the case where $k=1$, while $\ar{\ins{S}}=3$ and $\ar{\ins{T}}=6$ with $\ins{S}$ being the data schema and $\ins{T}$ being the extended schema.
Let $\ins{S}=\{T_1,T_2\}$ with both $T_1$ and $T_2$ ternary. We use
the Boolean~CQ (for brevity, the existential quantifiers are omitted)
$$
  q() = 
R(x_1,x_2) \wedge R(x_1,x_3) \wedge R(x_2,x_4) \wedge R(x_3,x_4)
\wedge A_1(x_2) \wedge A_2(x_3)
$$
which is a core of treewidth~$2 > k$. Our ontology takes the form
$$
\Sigma = \{ T_1 \rightarrow q()\} \cup \Sigma_1 \cup \Sigma_2
$$
where $\Sigma_1$ makes sure that there is an $S$-path of length $2^n$
whenever there is a $T_1$-atom and $\Sigma_2$ make sure that there is
an $S$-path of length $2^n-1$ whenever there is a $T_2$-atom, where
$S$ is a binary relation symbol.  We need auxiliary relation symbols
to establish these paths, but we will make sure that they are of arity
at least three so that they cannot be `seen' by
UCQ$_1$-approximations.  Note that while the presence of a $T_1$-atom
implies that $q()$ is true, the presence of a $T_2$-atom does not. In
detail, $\Sigma_1$ contains the following TGDs:
\begin{itemize}

\item $T_1(x_1,x_2,x_3)  \rightarrow B^0_i(x_1,x_2,x_3)$ for $1 \leq i
  \leq n$;

\item for $1 \leq i \leq n$:
$$
\begin{array}{r@{\,}c@{\,}l}
  B^0_i(x_1,x_2,x_3) & \rightarrow& \exists y_1 \exists y_2 \exists
  y_3 \, G(x_1,x_2,x_3,y_1,y_2,y_3) \\[1mm]
  && \qquad\qquad \ \ \, \wedge \, S(x_1,y_1) 
\end{array}
$$


\item for all $i <n$, where \Gbf abbreviates $G(x_1,x_2,x_3,y_1,y_2,y_3)$:
$$
\begin{array}{r@{\,}c@{\,}l}
\Gbf \wedge \bigwedge_{j=0}^{i-1} 
B^1_j(x_1,x_2,x_3) \wedge B^1_i(x_1,x_2,x_3)&\rightarrow&
B^0_j(y_1,y_2,y_3)\\[1mm]
\Gbf \wedge \bigwedge_{j=0}^{i-1} 
B^1_j(x_1,x_2,x_3) \wedge B^0_i(x_1,x_2,x_3)&\rightarrow&
B^1_j(y_1,y_2,y_3)
\end{array}
$$

\item for all $i,j$ with $j< i <n$, using the same abbreviation \Gbf:
$$
\begin{array}{r@{\,}c@{\,}l}
\Gbf \wedge 
B^0_j(x_1,x_2,x_3) \wedge B^0_i(x_1,x_2,x_3)&\rightarrow&
B^0_j(y_1,y_2,y_3)\\[1mm]
\Gbf \wedge 
B^0_j(x_1,x_2,x_3) \wedge B^1_i(x_1,x_2,x_3)&\rightarrow&
B^1_j(y_1,y_2,y_3).
\end{array}
$$

\end{itemize}
%
%
We leave the definition of $\Sigma_2$ to the reader.

This defines an OMQ $Q=(\ins{S},\Sigma,q)$. It is not hard to see that
$Q$ is equivalent to an OMQ $Q'$ from $(\class{G},\class{UCQ}_1)$,
namely to $Q'=(\ins{S},\Sigma,q')$ where $q'()$ is the Boolean CQ that
asks for the existence of an $S$-path of length $2^n$. Note that the
size of $Q'$ is exponential in the size of $Q$. This cannot be avoided
unless we use an ontology different from $\Sigma$, that is, resort
to non-uniform UCQ$_1$-equivalence.
\begin{lemma}
\label{lem:explargeapprox}
  Let $Q''=(\ins{S},\Sigma,q'') \in (\class{G},\class{UCQ}_1)$ be
    equivalent to~$Q$. Then $q''$ must contain a CQ with at least
    $2^n$ atoms.
\end{lemma}
\begin{proof}
  Take the $\ins{S}$-database $D_1= \{ T_1(c_1,c_2,c_3) \}$. Then $Q$
  evaluates to true on $D_1$ and thus so does $Q''$. Consequently,
  there is a CQ $p$ in $q''$ and a homomorphism $h$ from $p$ to
  $\chase{D_1}{\Sigma}$. Since $p$ is of treewidth 1, it cannot
  contain any atoms $P(\bar x)$ with $P$ of arity at least three and
  all variables in $\bar x$ different, and thus $h$ is a homomorphism
  from $p$ to $\chase{D_1}{\Sigma}^-$, which is obtained from
  $\chase{D_1}{\Sigma}$ by removing all facts that use predicates of
  arity at least three (by choice of $D_1$ and definition of $\Sigma$,
  none of these facts contains repeated constants) and which takes the
  form of an $S$-path of length $2^n$. We aim to show that $p$ has at
  least $2^n$ atoms. To this end, it suffices to show that for every
  edge $S(a_1,a_2)$ in the path $\chase{D_1}{\Sigma} \setminus \{
  T_1(c_1,c_2,c_3)\}$, there is an atom $S(x_1,x_2) \in p$ such that
  $h(x_i)=a_i$ for $i \in \{1,2\}$. Assume that this is not the case.
  Then cleary we also find a homomorphism from $p$ into an $S$-path of
  length $2^n-1$, thus into $\chase{D_2}{\Sigma}$ where $D_2= \{
  T_2(c_1,c_2,c_3) \}$. But $Q$ evaluates to false on $D_2$, which is a contradiction.
\end{proof}
Lemma~\ref{lem:explargeapprox} allows us to show that the $\class{UCQ}_k$-approximations $Q^a_k$
do not behave as intended when $k < \ar{\ins{T}}-1$. In particular, although the OMQ $Q = (\ins{S},\dep,q)$ defined above is $\class{UCQ}_1$-equivalent, we can show that is not equivalent to $Q^{a}_{1}$.
By Lemma~\ref{lem:explargeapprox}, we get that for any $Q'' = (\ins{S},\dep,q'')$ from $(\class{G},\class{UCQ}_1)$ that is equivalent to $Q$, $q''$ necessarily contains a CQ of exponential size.
On the other hand, the $\class{UCQ}_1$-approximation $Q^{a}_{1}$ of $Q$ contains only CQs of polynomial size since $\ar{\ins{T}}$ is a constant. Therefore, $Q$ is not equivalent to $Q^{a}_{1}$.

For $k < \ar{\ins{T}}-1$, we conjecture that every $\class{UCQ}_k$-approximation that is based on the same ontology as the original OMQ $Q$, and has the fundamental property of being equivalent to $Q$ when the later is uniformly $\class{UCQ}_k$-equivalent, must contain a CQ of double exponential size. This is not true our notion of $\class{UCQ}_k$-approximation in Definition~\ref{def:ucq-k-apx-via-squids}, where each CQ is, in general, of exponential size.

\section{Proof of Theorem~\ref{the:omq-main-technical-result}}

\subsection{Preliminaries}

We start by showing that existential quantifiers can be eliminated:

\begin{theorem}\label{thm:elimEx}
	Consider an OMQ $Q = (\ins{S},\dep,q) \in (\class{G},\class{UCQ})$, where $\dep$ and $q$ are over the schema $\ins{T} \supseteq \ins{S}$. We can construct an OMQ $Q' = (\ins{S},\dep\,q') \in (\class{G} \cap \class{FULL}, \class{UCQ})$, where $\dep'$ and $q'$ are over $\ins{T}$, such that, for every $\ins{S}$-database $D$, $Q(D) = Q'(D)$.
\end{theorem}

Before giving the proof of the above result, let us recall a key property of linear TGDs, that is, UCQ-rewritability, shown in~\cite{CaGL12}.

\begin{proposition}\label{pro:linear-ucq-rewritable}
	Given a set $\dep \in \class{L}$ and a UCQ $q$, both over a schema $\ins{T}$, we can construct a UCQ $q'$ over $\ins{T}$ such that, for every $\ins{T}$-database $D$, $q(\chase{D}{\dep}) = q'(D)$.
\end{proposition}

We are now ready to give the proof.

\begin{proof}(Theorem~\ref{thm:elimEx})
	The desired OMQ $Q'$ can be actually extracted from the proof of Lemma~\ref{lem:linearization}. From that proof, we know that we can construct, for some schema $\ins{T'}$ disjoint from $\ins{T}$,
	\begin{itemize}
		\item a set $\dep_1 \in \class{G} \cap \class{FULL}$, over $\ins{T}$,
		\item a set $\dep_2 \in \class{G} \cap \class{FULL}$ over $\ins{T} \cup \ins{T'}$, where the predicates of $\ins{T}$ (resp., $\ins{T'}$) occur only in the body (resp., head) of a TGD, and
		\item a set $\dep_3 \in \class{L}$ over $\ins{T} \cup \ins{T'}$,
	\end{itemize}
	such that, for every $\ins{S}$-database $D$,
	\[
	Q(D)\ =\ q(\chase{\chase{D}{\dep_1 \cup \dep_2}}{\dep_3}).
	\]
	Notice that we can further assume that $\dep_1 \cup \dep_2$ consists of TGDs with only one atom in the head; simply split a multi-head full TGD into several single-head TGDs, one for each head atom, that have the same body.
	Since $\dep_3$ is set of linear TGDs, by Proposition~\ref{pro:linear-ucq-rewritable}, we can construct a UCQ $\hat{q}$ over $\ins{T} \cup \ins{T'}$ such that
	\[
	q(\chase{\chase{D}{\dep_1 \cup \dep_2}}{\dep_3})\ =\ \hat{q}(\chase{D}{\dep_1 \cup \dep_2}).
	\]
	To obtain the desired set $\dep'$ and UCQ $q'$, it remains to eliminate the predicates of $\ins{T'}$. This can be achieved by simply unfolding the atoms in $\hat{q}$ with a predicate from $\ins{T'}$ using the single-head TGDs of $\dep_2$, and get a new UCQ $\tilde{q}$ over $\ins{T}$. It is clear that
	\[
	\hat{q}(\chase{D}{\dep_1 \cup \dep_2})\ =\ \tilde{q}(\chase{D}{\dep_1}).
	\]
	Therefore, the claim follows with $Q' = (\ins{S},\dep_1,\tilde{q})$.
\end{proof}




Let \Sbf be a schema, $D$ an \Sbf-database, and
$\abf \subseteq \adom{D}$ a
guarded set in $D$.  We aim to define a potentially infinite database $D^{\abf}$, the \emph{guarded unraveling of $D$ at} \abf. In parallel with $D^{\abf}$, we define a tree decomposition $(T,\chi)$, $T=(V,E)$, of the Gaifman graph of $D^{\abf}$ whose width is the maximum arity of relation names in \Sbf minus one.

Let $V$ be the set of sequences $v=\abf_{0}\cdots \abf_{n}$ of
guarded sets in $D$ such that $\abf_{0}=\abf$ and $\abf_{i}\cap
  \abf_{i+1}\not=\emptyset$ for $0 \leq i < n$.
%
Further, let 
$$
E=\{ (v,v') \in V \times V \mid v'=v\bbf \text{ for some
} \bbf \}.
$$
We associate with each $v \in V$ a database $D(v)$ and then define
$\chi(v)$ to be $\adom{D(v)}$. This completes the definition of
$(T,\chi)$ and allows us to define $D^{\abf}$ as $\bigcup_{v\in
  V}D(v)$.

Take an infinite supply of \emph{copies} of every $a\in
\adom{D}$. Set $b^{\uparrow}=a$ if $b$ is a copy of~$a$, and
$a^\uparrow =a$.  For $v \in V$, define $D(v)$ and its domain
$\adom{D(v)}$ by induction on the length of the sequence~$v$. For $v=\abf$,
$D(v)$ is
$D_{|\abf}$, the restriction
of $D$ to those atoms with only constants in~$\abf$.
%
To define $D(v)$ when $v=\abf_{0}\cdots \abf_{n}$
with $n>0$, 
take for any $a\in \abf_{n}\setminus \abf_{n-1}$ a fresh
copy $a'$ of $a$ and define $D(v)$ with domain 
$$
\begin{array}{rcl}
\adom{D(v)}&=&
 \{ b \in \adom{D(\abf_{0}\cdots
\abf_{n-1}} \mid b^{\uparrow} \in \abf_{n}\cap \abf_{n-1}\} \; \cup \\[1mm]
&&\{ a' \mid 
a\in \abf_{n}\setminus \abf_{n-1}\}
\end{array}
$$
such that $b \mapsto b^{\uparrow}$ is an isomorphism from $D(v)$ to $D_{|\abf_{n}}$.


\medskip

Injective homomorphisms play a crucial role in our proof of
Theorem~\ref{the:omq-main-technical-result}. Here, we make
some basic observations pertaining to such homomorphisms.
For a database $D$, a CQ~$q(\bar x)$, and a tuple of distinct
constants~\abf, 
we write $D \models^{io} q(\abf)$ if $D \models
q(\abf)$ and all homomorphisms $h$ from $q$ to $D$ with
$h(\bar x)=\abf$ are injective. 
Note that since the constants in \abf are distinct, the
condition $h(\bar x)=\abf$ does not conflict with injectivity.
 Here, `io' stands for `injectively only'.
The following is a simple observation, see \cite{BFLP19} for a proof.
\begin{restatable}{lemma}{lemiolem}
\label{lem:iolem}
If $D
\models q(\abf)$, for $D$  a potentially infinite database, $q$ a CQ,
and
\abf a tuple of distinct constants, then $D
\models^{io} q_c(\abf)$ for some contraction $q_c$ of $q$.
\end{restatable}
The next lemma is an important ingredient to the proof of
Theorem~\ref{the:omq-main-technical-result}.  We write $D,\abf
\rightarrow E,\bbf$ to indicate that there is a homomorphism $h$ from
database $D$ to database $E$ with $h(\abf)=\bbf$.
\begin{restatable}{lemma}{lemiomin}
\label{lem:iomin}
Let $Q=(\Sbf,\Sigma,q) \in (\class{G \cap \class{FULL}},\class{UCQ})$
and \abf a tuple of distinct constants. 
Then for every $\Sbf$-database $D$ with $D \models Q(\abf)$, there
is an \Sbf-database $\widehat D$ such that the following are
satisfied:
  \begin{enumerate}

  \item $\widehat D \models Q(\abf)$ and $\widehat D,\abf \rightarrow D,\abf$;


  \item if $\chase{\widehat D}{\Sigma} \models^{io} q_c(\abf)$, for $q_c$ a contraction
    of $q$, then there is no $\Sbf$-database
    $D'$ and  contraction $q'_c$ of $q$ where $D',\abf
    \rightarrow \widehat D,\abf$, $\chase{D'}{\Sigma} \models^{io}
    q'_c(\abf)$, and $q_c\neq q'_c$ is a contraction of~$q'_c$.

  \end{enumerate}
\end{restatable}
\begin{proof}
  We start with $\widehat D = D$. Then, of course, Condition~2 is not
  guaranteed to be satisfied. We thus iteratively replace $\widehat D$
  with more suitable databases, as follows.  As long as there
  are an $\Sbf$-database $D'$ and contractions $q_c,q'_c$ of $q$ such
  that $\chase{\widehat D}{\Sigma} \models^{io} q_c(\abf)$,
  $D',\abf \rightarrow \widehat D,\abf$,
  $\chase{D'}{\Sigma} \models^{io} q'_c(\abf)$, and $q_c \neq q'_c$ is
  a contraction of $q'_c$, replace $\widehat D$ by~$D'$.  Informally,
  this iterative process terminates since we lose at least one
  contraction $q_c$ of $q$ with
  $\chase{\widehat D}{\Sigma} \models^{io}q_c(\abf)$ in every step. We
  refer to \cite{BFLP19} for details.

  It is clear that the resulting $\widehat D$ satisfies
  Condition~(2). Condition~(1) is satisfied as well: we have
  $\chase{\widehat D}{\Sigma} \models^{io} q_c(\abf)$ for some contraction
  $q_c$ of $q$, thus $\widehat D \models Q(\abf)$.
\end{proof}

\begin{restatable}{lemma}{lemdecHomPre} \label{lem:decHomPre} Given an
  OMQ $Q=(\Sbf,\Sigma,q)$ from $(\class{G} \cap
  \class{FULL},\class{UCQ})$ and an $\Sbf$-database $\widehat D$, it is
  decidable whether 
  Conditions~1 and~2 from Lemma~\ref{lem:iomin}
  hold. \end{restatable}
\begin{proof}
  As Condition~1 is clearly decidable, we concentrate on
  Condition~2. We start with a central observation. Let $q'_c$ be a
  contraction of $q$ and let $q''_c$ be the UCQ that consists of all
  contractions of $q'_c$ that are distinct from~$q'_c$.  Then checking
  whether $\chase{D'}{\Sigma} \models^{io} q'_c(\abf)$, for some
  $\Sbf$-database $D'$ with $D',\abf \rightarrow \widehat D,\abf$, is
  equivalent to checking that the OMQs $(\Sbf,\Sigma,q'_c)$ and
  $(\Sbf,\Sigma,q''_c)$ are not equivalent on \Sbf-databases $D'$ with
  $D',\abf \rightarrow \widehat D,\abf$.  As a consequence, we can
  decide
  Condition~2 by checking that there are no contractions
  $q_c$, $q'_c$ of $q$ such that $q_c$ is a contraction of $q'_c$,
  $q_c \neq q'_c$, $\chase{\widehat D}{\Sigma} \models q_c(\abf)$ 
  and $(\Sbf,\Sigma,q''_c)$ are not equivalent on \Sbf-databases $D'$
  with $D',\abf \rightarrow \widehat D,\abf$.  Note that in the
  statement `$\chase{\widehat D}{\Sigma} \models q_c(\abf)$', we have
  dropped the $\cdot^{io}$ that is present in Condition~2. This is
  not harmful since if  $\chase{\widehat D}{\Sigma} \models q_c(\abf)$,
  but  $\chase{\widehat D}{\Sigma} \not\models^{io} q_c(\abf)$, we can
  simply replace $q_c$ by a contraction of itself that satisfies this
  stronger condition.

  The \emph{width} of an OMQ $Q=(\Sbf,\Sigma,q)$ from
  $(\class{G} \cap \class{FULL},\class{UCQ})$, denoted $\mn{w}(Q)$, is
  the maximum number of variables in the body of a rule in $\Sigma$ or
  in a CQ in $q$.
  
  To establish decidability of Condition~2, it suffices to show that
  given OMQs $Q_1$ and $Q_2$ from
  $(\class{G} \cap \class{FULL}, \class{UCQ})$ over schema \Sbf, an \Sbf-database
  $\widehat D$, and a tuple of constants \abf, it is decidable whether
  there exists an $\Sbf$-database $D$ such that $D \models Q_1(\abf)$,
  $D \not \models Q_2(\abf)$, and $D,\abf \to \widehat D,\abf$.  We note that this is
  the case iff there exists a database with these properties that is
  of treewidth $\mn{w}(Q_1)$.
 
  For an OMQ $Q=(\Sbf,\Sigma,q)$ from
  $(\class{G} \cap \class{FULL}, \class{UCQ})$ we construct a Guarded
  Second Order (GSO)~\cite{GrHO02}
  formula $\phi_{Q}$ such that for every $\Sbf$-database $D$ and tuple
  of constants \abf, $D \models Q(\abf)$ iff $D \models
  \phi_Q(\abf)$. Let $J_1,\dots,J_n$ be the relation symbols in
  $\Sigma$ and reserve fresh relation symbols $J'_1,\dots,J'_n$
  and $J''_1,\dots,J''_n$ of the same arity. Further let:
  \begin{itemize}
  \item $\phi'_{\Sigma}$ (resp.\ $\phi''_{\Sigma}$) be the conjunction
    of all TGDs in $\Sigma$, viewed as universally quantified FO
    formulas, with each relation symbol $J_i$ replaced by $J'_i$
    (resp.\ $J''_i$);

  \item $\phi_\Sbf$ be a formula which says that if $J_i$ is in \Sbf,
    then the extension of $J_i$ is included in that of $J'_i$;
    
  \item  $\phi_!$ be a formula which says that the extension of each
    $J''_i$ is included in the extension of $J'_i$, with at least one
    inclusion being strict.

    \end{itemize}
    Then $\phi_{Q}$ is the following formula:\footnote{We use
    the variant of GSO whose the syntax is identical to that of Second
    Order Logic, but the second-order quantifiers range only
    over guarded relations, see
    \cite{GrHO02}.}
$$\exists J'_1, \ldots, J'_n\ (\phi_\Sbf \wedge \phi_{\Sigma} \wedge  q \wedge \forall J''_1 \ldots J''_n \ (\phi_! \rightarrow \neg  \phi'_{\Sigma})).$$

Note that the extension of $J'_1,\dots,J'_n$ represents the chase of
the input database with $\Sigma$.

We also take a formula $\phi_{\widehat D}$ that evaluates to true
exactly on those \Sbf-databases $D$ such that $D,\abf \rightarrow\widehat
D,\abf$. Then, we have to decide whether
there exists an $\Sbf$-database $D$ of treewidth at most
$\mn{w}(Q_1)$ such that $D \models \phi_{Q_1} \wedge \neg
\phi_{Q_2} \wedge \phi_{\widehat D}$. Note that on sparse
structures, GSO has the same expressive power as MSO
\cite{Blum10}. As databases of bounded treewidth are sparse \cite{Cour03}, it follows that GSO over such databases is decidable.
\end{proof}
For $k,\ell \geq 1$, the \emph{$k \times \ell$-grid} is the
(undirected) graph with vertex set $\{ (i,j) \mid 1\leq i \leq k
\text{ and } 1 \leq j \leq \ell \}$ and an edge between $(i,j)$ and
$(i',j')$ iff $|i-i'|+|j-j'|=1$. A graph $H$ is a \emph{minor} of a
graph $G$ if $H$ is isomorphic to a graph that can be obtained from a
subgraph of $G$ by contracting edges. Equivalently, $H$ is a minor of
$G$ if there is a \emph{minor map} from $H=(V_H,E_H)$ to
$G=(V_E,E_G)$, that is, a mapping $\mu: V_H \rightarrow 2^{V_G}$ with
the following properties for all $v,w \in V_H$:
\begin{itemize}

\item $\mu(v)$ is non-empty and connected in $G$;

\item $\mu(v)$ and $\mu(w)$ are disjoint whenever $v \neq w$;

\item if $\{v,w\} \in E_H$, then there are $v' \in \mu(v)$ and $w' \in
  \mu(w)$ such that $\{v',w'\} \in E_G$.
  
\end{itemize}
We say that $\mu$ is \emph{onto} if $\bigcup_{v \in V_H} \mu(v)=V_G$.
For a database $D$, we say that $a \in \adom{D}$ is \emph{isolated in}
$D$ if there is only a single atom $R(\abf) \in D$ with $a \in \abf$.
When $k$ is understood from 
the context, we use $K$ to denote~${k \choose 2}$.%

Note that Point~(3) of the following theorem can be viewed as a form
of ontoness requirement for the homomorphism $h_0$ from Point~(1): for
the fact $R(a_1,\dots,a_n) \in D$ mentioned in that condition, it is
possible to find a fact $R(c_1,\dots,c_n) \in D_G$ that maps to it.

\grohetechnew*

\noindent
\begin{proof}
  The construction of $D_G$ is a slight extension of the original
  construction from \cite{Grohe07}. The extension is necessary because
  we have to treat the constants from $\adom{D} \setminus A$ in a
  special way whereas in \cite{Grohe07} there is no set $A$ or, in
  other words, $A=\adom{D}$. Let $\mu$ be a minor map from the
  $k \times K$-grid to $G^D_{|A}$. Since $D$ is connected, we can
  assume w.l.o.g.\ that $\mu$ is onto. To attain Point~(3) of
  Theorem~\ref{thm:grohetechnew}, we actually have to choose $\mu$ a
  bit more carefully. We first proceed with the construction and argue
  that Points~(1) and~(2) are satisfied, and then give details on
  Point~(3). 
  Fix a
  bijection $\rho$ between $[K]$ and the set of unordered pairs over
  $[k]$. For $p \in [K]$, let $i \in p$ be shorthand for
  $i \in \rho(p)$. Let $G^D_{|A}=(V,E)$.


  The domain of $D_G$ is
$$
\begin{array}{@{}l}
(\adom{D} \setminus A) \, \cup \\[1mm]
\{ (v,e,i,p,a)  \in V \times E \times
[k] \times [K] \times A \mid (v \in e \Leftrightarrow i \in p),  a \in
\mu(i,p) \}.
\end{array}
$$
Before we say what the atoms of $D_G$ are, let us define the
homomorphism $h_0: \adom{D_G} \to \adom{D}$ from Point~(1) by taking
the identity on $\adom{D} \setminus A$ and the projection to the last
component for all other constants from $\adom{D_G}$. We extend $h_0$
to tuples over $\adom{D_G}$ in the expected way.

We now define $D_G$ to contain every atom $R(\bbf)$, $\bbf$ a tuple over
$\adom{D_G}$, such that $R(h_0(\bbf)) \in D$ and for all constants
$b,b'$ in \abf that are of the form $b=(v, e, i, p, a)$ and
$b'=(v', e', i', p', a')$,
$$
\begin{array}{rcl}
\text{(C1)  } i=i' \text{ implies } v=v' &\text{and}& \text{(C2)  } p=p' \text{ implies } e=e'. 
\end{array}	
$$
This finishes the construction of $D_G$.

Point~(1) follows directly from the construction. The proof of
Point~(2) is a slight variation of the proofs of Lemma~4.3 (for the
`$\Rightarrow$' direction) and Lemma~4.4 (for the `$\Leftarrow$' direction)
from \cite{Grohe07}. In fact, the homomorphism $h$ constructed in the
proof of Lemma~4.3 simply needs to be extended to be the identity on
$\adom{D} \setminus A$. 
In Lemma~4.4, the assumption that $D$ is a core has been replaced here
with the condition that $h_0(h(\cdots))$ is the identity. This,
however, is exactly what is derived from the assumption that $D$ is a
core in the original proof. No further changes to the proof are
required despite the presence of the elements in $\adom{D} \setminus A$. 
%


For Point~(3), assume that $R(a_1,\dots,a_n) \in D$, $h_0(b_j)=a_j$ for $1 \leq j
  \leq n$, and $S=\{ b_j \mid a_j \text{ non-isolated in } D \}$ is a
  clique in the Gaifman graph of $D_G$.
  As $D$ is connected, $S$ cannot be empty. 
  We can assume w.l.o.g\ that the minor map $\mu$ from the $k \times 
  K$-grid onto $G^D_{|A}$ is such that one of the following holds:
  \begin{enumerate}

  \item there is an $a_\ell$ that is non-isolated in $D$ such that
    $a_\ell \in A$ and if $i,p$ are such that $a_\ell \in \mu(i,p)$, then
    all $a_j \in A$ that are isolated in $D$ are also in $\mu(i,p)$;
    in this case, $b_\ell$ must take the form $(v,e,i,p,a_\ell)$;

  \item there is no $a_\ell$ that is non-isolated in $D$ such that
    $a_\ell \in A$ and there are $i,p$ such that all $a_j \in A$
    that are isolated in $D$ are in $\mu(i,p)$; in this case, choose $v,e$ arbitrarily
    such that $v\in e$ iff $i \in p$.
    
  \end{enumerate}
  %
  Define
  $c_1,\dots,c_n$ by setting
  $$
  c_j := \left \{
    \begin{array}{ll}
       (v,e,i,p,a_j) & \text{ if } a_j\in A \text{ is isolated in } D \\[1mm]
       b_j & \text{ otherwise.}
    \end{array}
    \right .
  $$
  It can be verified that $(v,e,i,p,a_j)$ is a
  constant in $D_G$ if $a_j \in A$ is isolated in $D$, due to our
  assumptions (1) and~(2). Moreover, $h_0(c_i)=h_0(b_i)$ for $1 \leq i
  \leq n$ and Conditions (C1) and (C2) 
  are satisfied for the fact
  $R(c_1,\dots,c_n)$, which is thus in $D_G$. Finally, $h_0(c_i)=a_i$
  for $1 \leq i \leq n$,
  by definition of $h_0$.
\end{proof}

%
%

\subsection{The Reduction}

Recall that our goal is to establish the following lower bound:

\begin{restatable}{theorem}{thmgrohemainlower}
\label{thm:grohemainlower}
%
Fix $r \geq 1$. Let $\class{O}$ be a recursively enumerable class of OMQs from $(\class{G},\class{UCQ})$ over a schema of  arity $r$, and, for each $k \geq 1$, $\class{O} \not\subseteq (\class{G},\class{UCQ})_{k}^{\equiv}$. Then, \text{\rm {\sf p}-{\sf OMQ}-{\sf Evaluation}($\class{O}$)} is \W-hard. 
\end{restatable}
%
%
The proof is by fpt-reduction from {\sf p}-{\sf Clique}, a \W-hard problem.
Assume that $G$ is an undirected
graph and $k \geq 1$ a clique size, given as an input to the
reduction. 
By Robertson and Seymour's Excluded Grid Theorem, there
is an $\ell$ such that every graph of treewidth exceeding $\ell$
contains the $k \times K$-grid as a minor~\cite{EGT}. We may assume, w.l.o.g., that $\ell$ exceeds the fixed arity $r$.  By our assumption
on $\class{O}$, we find an OMQ $Q=(\Sbf,\Sigma,q)$ from $\class{O}$
that is not in $(\class{G},\class{UCQ})^\equiv_\ell$. Since the
choice of $Q$ is independent of $G$ and since it is decidable whether
an OMQ from $(\class{G},\class{UCQ})$ is in
$(\class{G},\class{UCQ})^\equiv_\ell$
by Theorem~\ref{the:ucq-k-equiv-complexity-omq}, we can enumerate the OMQs from $\class{O}$ until we find $Q$.


By Theorem~\ref{thm:elimEx}, we can assume, w.l.o.g., that $Q=(\Sbf,\Sigma,q)$ is from $(\class{G} \cap
\class{FULL},\class{UCQ})$. It might of course be that the OMQ from
$(\class{G} \cap \class{FULL},\class{UCQ})$ is not in $\class{O}$, but
we will see that this is not a problem. 
Let $q=q_1 \vee \cdots \vee q_n$ and let $Q_i = (\Sbf,\Sigma,q_i)$ for
$1 \leq i \leq n$. We can assume w.l.o.g.\ that
$Q_i \not\subseteq Q_j$ for all $i \neq j$ because if this is not the
case, then we can drop $q_j$ from $q$ in $Q$ and the
resulting OMQ is equivalent to $Q$. Also, we exhaustively replace CQs $q_i$ in $Q$ with $(q_i)^a_\ell$ as defined in the context of
$\class{UCQ}_\ell$-approximations (see Definition~\ref{def:ucq-k-apx-via-squids}), whenever the resulting OMQ is equivalent to $Q$.\footnote{We rely on the key property of $\class{UCQ}_\ell$-approximations given by item (2) of Lemma~\ref{lem:ucq-k-apx}.} Since $Q$ is not in
$(\class{G},\class{UCQ})^\equiv_\ell$, the final OMQ must contain a
$q_w$ that has not been replaced by $(q_w)^a_\ell$ and, in particular, we must have $Q \not\subseteq Q'$ where $Q'$ is obtained from
$Q$ by replacing $q_w$ with $(q_w)^a_\ell$. We thus find an
\Sbf-database $D_0$ and tuple of constants $\abf_0$ such that
$D_0 \models Q_w(\abf_0)$, $D_0 \not\models
(Q_w)^a_\ell(\abf_0)$, 
and $D_0 \not\models Q_i(\abf_0)$ for all $i \neq w$. By duplicating
constants, we can achieve that all constants in $\abf_0$ are
distinct.
%
%
%
We also assume
that Condition~2 of Lemma~\ref{lem:iomin} is satisfied for
$Q=Q_w$, $q=q_w$, $\widehat D = D_0$, and $\abf = \abf_0$. If it is
not, we can apply that lemma and replace $D_0$ with the
resulting~$\widehat D_0$. By Condition~(1),
$\widehat D_0 \models Q_w(\abf_0)$, $\widehat D_0 \not\models (Q_w)^a_\ell(\abf_0)$,
and $\widehat D_0 \not\models Q_i(\abf_0)$ for all $i \neq
w$. 

Since the properties of $D_0$ are independent of $G$ and due to
Lemma~\ref{lem:decHomPre}, we can find $D_0$ by
enumeration. However, $D_0$ is still not as required and needs to
be manipulated further to make it suitable for the reduction. Before
we can carry out the actual manipulation, we need some preliminaries.


%

Recall that for a guarded set $\abf \subseteq \adom{D_0}$,
$D_0^{\abf}$ denotes the guarded unraveling of $D_0$ starting at \abf.
An  \emph{atomic query (AQ)} takes the form $R(\bar x)$, i.e., it
has a single atom and no quantified variables.
%
\begin{lemma}
\label{lem:unravcon}
$D_0 \models (\Sbf,\Sigma,p)(\abf')$ iff
$D_0^{\abf} \models (\Sbf,\Sigma,p)(\abf')$ if $p$ is an AQ and all
constants in $\abf'$ are also in \abf. 
%
%
%
%
%
\end{lemma}
%
%
Of course, $D^{\abf}_0$ can be infinite. By compactness, however,
there is a finite $D_{\abf} \subseteq D^{\abf}_0$ such that
Lemma~\ref{lem:unravcon} is satisfied for all (finitely many) AQs $p$
that use only symbols from $\Sigma$. For brevity, we say that $D_{\abf}$
\emph{satisfies Lemma~\ref{lem:unravcon} for all relevant AQs}.  We
can find $D_{\abf}$ by constructing $D^{\abf}_0$ level by level and
deciding after each such extension whether we have found the desired
database, by checking the condition in Lemma~\ref{lem:unravcon} for
all relevant AQs.

\medskip




  

\medskip A \emph{diversification} of $D_0$ is a database $D$ obtained
from $D_0$ by replacing every atom $R(\abf) \in D_0$ with a
(potentially empty) finite set of atoms $R(\abf'_1),\dots,R(\abf'_n)$ such
that each $\abf'_i$ is obtained from \abf by replacing some constants
that do not occur in $\abf_0$ with fresh constants. A constant
$a \in \adom{D}$ is \emph{old} if it already occurs in $D_0$ and
\emph{fresh} otherwise. For a fresh constant $b$, we use $b^\uparrow$
to denote the constant in $D_0$ that it was introduced for. For old
constants $a$, $a^\uparrow = a$. This extends to tuples of constants
from $D$ in a component-wise fashion.  Note that all fresh constants
in $D$ are isolated in $D$ and that all constants from $\abf_0$ still
occur in $D$ as old constants.
If $D_1$ and $D_2$ are diversifications of $D_0$, we write
$D_1 \preceq D_2$ if for every $R(\abf_1) \in D_1$, there is an
$R(\abf_2) \in D_2$ with
$\abf_1 \cap \adom{D} \subseteq \abf_2 \cap \adom{D}$.
\begin{example}
   Assume that $D=\{R(a,b,c),S(a,b,d)\}$ and $\abf_0=()$. Then
    $D_1 = \{R(a,b',c'),S(a,b,d)\}$ and $D_2 = \{R(a,b,c'),S(a,b,d)\}$
    are diversifictions of $D$ and $D_1 \preceq D_2$.
\end{example}
If $D$ is a diversification of $D_0$, we use $D^+$ to denote the
database obtained from $D$ by adding each database $D_{\abf}$ such
that for some atom $R(\abf)$ in $D$, we obtain $D_{\abf}$ from
$D_{\abf^\uparrow}$ by renaming the constants in $\abf^\uparrow$ to
those in~$\abf$.
A constant $a \in \adom{D^+}$ is
\emph{old} if it is an old constant from $\adom{D}$, and \emph{fresh}
otherwise. 

For what follows, we choose a $\preceq$-minimal diversification $D_1$
of $D_0$ such that $D^+_1 \models Q_w(\abf_0)$. 
\begin{example}
  Assume that $q_w$ is a Boolean CQ that takes the form of an $n
  \times m$-grid using the binary relations $X$ and $Y$. Let $\Sbf =
  \{ X',Y' \}$ with both relations ternary and let $\Sigma = \{
  X'(x,y,z) \rightarrow X(x,y), \ Y'(x,y,z) \rightarrow Y(x,y)
  \}$. Further let $D_0$ be the database
  $$
  \begin{array}{l}
     \{ X'(a_{i,j},a_{i,j+1},b) \mid 1 \leq i \leq m \text{ and } 1
     \leq j < n \} \; \cup \\[1mm]
     \{ Y'(a_{i,j},a_{i+1,j},b) \mid 1 \leq i < m \text{ and } 1
     \leq j \leq n \} 
  \end{array}
  $$
  Then the following database is a $\preceq$-minimal diversification
  of $D_0$ such that $D^+_1 \models Q_w(\abf_0)$:
  $$
  \begin{array}{l}
     \{ X'(a_{i,j},a_{i,j+1},b_{ij}) \mid 1 \leq i \leq m \text{ and } 1
     \leq j < n \} \; \cup \\[1mm]
     \{ Y'(a_{i,j},a_{i+1,j},b'_{ij}) \mid 1 \leq i < m \text{ and } 1
     \leq j \leq n \}.
  \end{array}
  $$

\end{example}
\noindent
We next observe some basic properties of~$D_1$. 
\begin{restatable}{lemma}{lemAprop}\label{lem:Aprop}~\\[-4mm]
  \begin{enumerate}

  \item $D_1^+ \models Q_w(\abf_0)$ and 
    $D_1^+ \not\models Q_i(\abf_0)$ for all $i \neq w$;

    
    
    \item $D_1$ has treewidth
    exceeding $\ell$ up to $\abf_0$;

  \item $D_1^+$ satisfies
    Condition~2 of Lemma~\ref{lem:iomin}.


  \end{enumerate}
\end{restatable}
\begin{proof} For Point~1, we have $D_1^+ \models Q_w(\abf_0)$ by
  choice of $D_1^+$. By construction of $D_1^+$, we have
  $D_1^+,\abf_0 \rightarrow D_0,\abf_0$ and thus
  $D_1^+ \not \models Q_i(\abf_0)$ for all $i \neq w$.
  We proceed with Point~2, noting that we can argue in the same
  way to show that
  $D_1^+ \not \models (Q_w)^a_\ell(\abf_0)$. Thus
  Point~2 of Lemma~\ref{lem:ucq-k-apx} implies that the treewidth of
  $(D_{1}^{+})$ up to $\abf_0$ exceeds $\ell$. Note that the treewidth of $D_1^+ \setminus D_1$ is bounded by the fixed arity
  $r$.\footnote{This is no longer true of $\Sigma$ is
    formulated in $\class{FG}$ and this is in fact the reason why our
    proof does not extended from $\class{G}$ to $\class{FG}$.} As we
  assume $\ell > r$, $D_1$ has
  treewidth exceeding $\ell$ up to $\abf_0$. 

  \smallskip

 Now for Point~3, that is, $D_1^+$ satisfies Condition~2 of
  Lemma~\ref{lem:iomin}. Assume to the contrary that there are an
  $\Sbf$-database $D'$ and contractions $q_c$, $q'_c$ of $q_w$ such
  that $\chase{D_1^+}{\Sigma} \models^{io} q_c(\abf_0)$, $D',\abf_0 \rightarrow
  D_1^+,\abf_0$, $\chase{D'}{\Sigma} \models^{io} q'_c(\abf_0)$, and $q_c \neq
  q'_c$ is a contraction of $q'_c$. We use a case distinction:
  \begin{itemize}

  \item $\chase{D_0}{\Sigma} \models^{io} q_c(\abf_0)$.

    From $D',\abf_0 \rightarrow D_1^+,\abf_0$ and
    $D_1^+,\abf_0 \rightarrow D_0,\abf_0$, we obtain
    $D',\abf_0 \rightarrow D_0,\abf_0$. This together with
    $\chase{D_0}{\Sigma} \models^{io} q_c(\abf_0)$ and
    $\chase{D'}{\Sigma} \models^{io} q'_c(\abf_0)$ yields a
    contradiction to $D_0$ satisfying Condition~2 of
    Lemma~\ref{lem:iomin}.

  \item $\chase{D_0}{\Sigma} \not\models^{io} q_c(\abf_0)$. 
  
    From $\chase{D_1^+}{\Sigma} \models^{io} q_c(\abf_0)$ and $D_1^+,\abf_0 \rightarrow
    D_0,\abf_0$, we obtain $\chase{D_0}{\Sigma} \models q_c(\abf_0)$.  By
    Lemma~\ref{lem:iolem}, there is a contraction $q'_c$ of $q_c$ with
    $\chase{D_0}{\Sigma} \models^{io} q'_c(\abf_0)$. But then we must have
    $q'_c = q_c$ as otherwise we obtain a contradiction to $D_0$
    satisfying Condition~2 of Lemma~\ref{lem:iomin} (instantiated with
    $D=\widehat D=D_0)$. Contradiction.

  \end{itemize}
	This completes the proof of Lemma~\ref{lem:Aprop}.
\end{proof}
The next lemma states a more intricate property of~$D_1$ that is
crucial for the remaining proof.
\begin{lemma}
\label{lem:point2lem}
If $h$ is a homomorphism from $q_w(\bar x_0)$ to
    $\chase{D_1}{\Sigma}$ with $h(\xbf)=\abf_0$,
    $R(\abf) \in D_1$, and $a_1,a_2$ are old constants in \abf, then
    there is a path $x_1,\dots,x_m$ in the Gaifman graph of $D[q_w]$
    such that $h(x_1)=a_1$, $h(x_m)=a_2$, and
    $h(x_{2}),\dots,h(x_{m-1})$  
    are in $\adom{D_{\abf}}$.\footnote{A special case is that there is
      an atom $S(\bar x)$ in $q_w$ such that $a_1$ and $a_2$ are both
      in~$h(\bar x)$. Also note that it follows that all old
      constants are in the range of $h$.}
\end{lemma}
\begin{proof}
  Assume to the contrary that there is a homomorphism $h$ from
  $q_w(\bar x_0)$ to $\chase{D_1^+}{\Sigma}$ with $h(\bar x_0)=\abf_0$, an
  $R(\abf) \in D_1$, and old constants $a_1,a_2$ in \abf such that
  there is no path in the Gaifman graph $G^{D[q_w]}$ of $D[q_w]$ as
  required by Lemma~\ref{lem:point2lem}.  We argue that we can then
  find an \Sbf-database $D_2$ with $D_2 \prec D_1$ and
  $D^+_2 \models Q_w(\abf_0)$, in contradiction to the choice of
  $D_1$.


  Let $D_2$ be obtained from $D_1$ by replacing $R(\abf)$ with the
  atoms $R(\abf_1), R(\abf_2)$ where $\abf_i$ is obtained from \abf by
  replacing $a_i$ with a fresh constant $b_i$ and each fresh constant
with a new fresh constant. Clearly, $D_2$ is a
  diversification of $D_0$ and $D_2 \prec D_1$. We
  convert $h$ into a homomorphism $h'$ from $q_w$ to
  $\chase{D_2^+}{\Sigma}$ such that $h'(\bar x)=\abf_0$, which
  contradicts the choice of $D_1$.

  As a preparation, let us discuss the construction of $D_1^+$ and
  $D_2^+$. During the construction of $D^+_1$, we have attached a
  copy $D_{\abf}$ of $D_{\abf^\uparrow}$ at \abf. In the construction
  of $D^+_2$, we have attached copies $D_{\abf_1}$ and $D_{\abf_1}$
  of the same $D_{\abf^\uparrow}$ at $\abf_1$ and $\abf_2$. The
  constants in $\adom{D_{\abf}} \setminus \adom{D_1}$ are the only
  constants in $D^+_1$ that may not be in $D^+_2$.
%
  %
  For $i \in \{1,2\}$, fix an isomorphism $\iota_i$ between
  $D_{\abf}$ and $D_{\abf_i}$. Note that $\iota_{i}$ is the identity
  on $a_{2-i}$ and in fact on all old constants in \abf
  except $a_i$.

  Now define $h'(x)$ as follows, for every variable $x$ in $q$:
  \begin{enumerate}
  \item if $h(x) \in \adom{D^+_1} \setminus \adom{D_{\abf}}$, then $h'(x)=h(x)$
  \item if $h(x) \in \adom{D_{\abf}}$: 
  	\begin{enumerate}
  	\item if $G^{D[q_w]}$ contains a path $y_1,\dots,y_p$ such
          that $h(y_1)=a_i$ for $i \in \{1,2\}$, $y_p=x$, and
          $h(x_{2}),\dots,h(x_{p-1})$ are 
          in $\adom{D_{\abf}}$, then $h'(x)=\iota_{2-i}(h(x))$; 
  	\item if $h(x)$ is still undefined, then $h'(x)=\iota_{1}(h(x))$.
  	\end{enumerate}
  \end{enumerate}
  Note that $h'$ is defined properly: due to our assumption that there
  is no path as in Lemma~\ref{lem:point2lem}, the precondition of 2b
  cannot be satisfied both for $i=1$ and for $i=2$.  Also note that
if $S(\bar x) \in q_w$ with
  $h(\bar x) \subseteq \mn{adom}(D_{\abf})$, then 
  $h'(x)=\iota_{i}(h(x))$,
  for some $i \in \{1,2\}$.


  It remains to show that $h'$ is a homomorphism from $q_w$ to
  $\chase{D^+_2}{\Sigma}$. Let us first observe three helpful claims.
  Recall that $\cdot^\uparrow$ maps each constant in a diversification
  of $D_0$, such as $D_1$ and $D_2$, to the original constant that it
  was introduced for.
  The first claim is a consequence of Lemma~\ref{lem:unravcon} and the
  fact that since $D_1$ and $D_2$ are diversifications of $D_0$, 
  $\cdot^\uparrow$ defines a homomorphism from $D_j$ to $D_0$ each
  $j \in \{1,2\}$ that can be extended to a homomorphism from $D_j^+$
  to $D_0$ and further from $\chase{D^+_j}{\Sigma}$ to
  $\chase{D_0}{\Sigma}$.
  \\[2mm]
  {\bf Claim~1.}  Let 
 \bbf be a guarded set in~$D_j$,
  $j \in \{1,2\}$.  Then $R(\bbf) \in \chase{D_j^+}{\Sigma}$ iff
  $R(\bbf^\uparrow) \in \chase{D^+_0}{\Sigma}$ for all facts $R(\bbf)$.
  \\[2mm]
  Here is the second claim.
  \\[2mm]
  {\bf Claim~2.}  Let $\bbf$ be a guarded set in $D_{\abf}$. 
  Then $R(\bbf) \in
  \chase{D_1^+}{\Sigma}$ implies $R(\iota_i(\bbf)) \in
  \chase{D^+_2}{\Sigma}$ for all facts $R(\bbf)$ and $i \in \{1,2\}$.
  \\[2mm]
  Assume to the contrary that
  $R(\iota_i(\bbf)) \notin \chase{D^+_2}{\Sigma}$. Set
  $$
  \begin{array}{rcl}
  D &=& \chase{D_1^+}{\Sigma}|_{\adom{D^+_1} \setminus
        \adom{D_{\abf}}} \; \cup\\[1mm]
  && 
  \{ R(\iota_i^-(\cbf)) \mid R(\cbf) \in \chase{D^+_2}{\Sigma}|_{\adom{D_{\abf_i}}} \}.
  \end{array}
  $$
  Using the fact that $\Sigma$ is guarded, it can be shown
  that $D$ is a model of $\Sigma$ and of $D_1^+$. Moreover,
  $R(\bbf) \notin D$ and, consequently, $R(\bbf) \not\in
  \chase{D_1^+}{\Sigma}$,
  which finishes the proof of the claim.

  \medskip
  The proof of the third claim is identical to that of Claim~2.
  \\[2mm]
  {\bf Claim~3.}  Let $\cbf$ be a guarded set in some tree $D_{\bbf}$ in
  $D^+_1$, $\bbf \neq \abf$. Then $R(\cbf) \in \chase{D_1^+}{\Sigma}$
  implies $R(\cbf) \in \chase{D^+_2}{\Sigma}$ for all atoms
  $R(\cbf)$.

\medskip

  %
  %



%
  %
Now for the proof that $h'$ is a homomorphism from $q_w$ to
$\chase{D^+_2}{\Sigma}$. Let $S(\bar x) \in q_w$.

First assume that $h(\bar x)$ contains only constants from $\adom{D_1}$.
It follows from (two applications of) Claim~1 that $S(h(\bar x)) \in
\chase{D^+_1}{\Sigma}$ implies $S(h'(\bar x)) \in
\chase{D^+_2}{\Sigma}$.

Next assume that $h(\bar x)$ contains a constant from a tree
$\adom{D_{\cbf}}$, $\cbf \neq \abf$, that is not in $\adom{D_1}$.  By
construction of $D^+_1$ and by definition of the chase, this implies that
$h(\bar x) \subseteq \mn{adom}(D_{\cbf})$. We can thus apply Claim~3 to
infer that $S(h(\bar x)) \in \chase{D^+_1}{\Sigma}$ implies
$S(h'(\bar x)) \in \chase{D^+_2}{\Sigma}$.

Finally assume that $h(\bar x)$ contains a constant from $\adom{D_{\abf}}
\setminus \adom{D_1}$. Then
  $h(\bar x) \subseteq \mn{adom}(D_{\abf})$.
Thus 
$h'(\bar x)=\iota_i(h(\bar x))$
for some $i \in \{1,2\}$. We can thus apply
Claim~2 to infer that
$S(h(\bar x)) \in \chase{D^+_1}{\Sigma}$ implies
$S(h'(\bar x)) \in \chase{D^+_2}{\Sigma}$.
\end{proof}
Clearly, neither $D_1$ nor $D_1^+$ are guaranteed to be connected. In
particular, $D_0$ might be disconnected and even if it is connected
the diversification might replace all constants in an atom with fresh
constants and in this way make $D_1$ disconnected. By Point~2 of
Lemma~\ref{lem:Aprop}, however, the restriction of $G^{D_1}$ to
$\adom{D_1} \setminus \abf_0$ has a maximal connected component that
is of treewidth exceeding $\ell$. Let $A$ be the set of constants in
that component.  Moreover, let $D$ be the maximal connected component
of $D_1$ that contains all constants from $A$ (it might additionally
contain constants from $\abf_0$) and let $D^+_{\mn{dis}}$ denote the
part of $D_1^+$ that is disconnected from $D \subseteq D_1^+$.  By
choice of $D$, $A$, and $\ell$, the restriction of $G^D$ to $A$
contains the $k \times K$-grid as a minor. We can thus apply
Theorem~\ref{thm:grohetechnew} to $G$, $k$, and $D$, and $A$,
obtaining an $\Sbf$-database $D_G$ and a homomorphism $h_0$ from $D_G$
to $D$ such that Points~1 to~3 of that theorem are satisfied.  Let
\begin{enumerate}

\item $D^+_G$ be obtained by starting with $D_G$
  and then disjointly adding, for each guarded set $\abf$ in  $D_G$
  such that the restriction of $h_0$ to \abf is injective, a
    copy of the database $D_{h_0(\abf)}$ that was attached to
    $R(h_0(\abf)) \in D_1$ during the construction of $D_1^+$,
    identifying the root of this copy with~$\abf$;\footnote{The
      injectivity requirement implies that we can add these copies as
      they are, without duplicating any constants in them.}

\item $D^{*}_G$ be obtained by disjointly adding
  $D^+_{\mn{dis}}$. 

\end{enumerate}
The fpt-reduction of {\sf p}-{\sf Clique} then consists in computing $Q$ and $D^{*}_G$ from $G$ and $k \geq 1$. We mean here the original $Q$ that is guaranteed to be in $\class{O}$ and where $\Sigma$ is from $\class{G}$ but not necessarily from $\class{G} \cap \class{FULL}$.

\subsection{Correctness of the Reduction}

We show in the subsequent lemma that $D^{*}_G \models Q(\abf_0)$ if
and only if $G$ has a $k$-clique. Clearly, we can work with the
version of $Q$ here in which $\Sigma$ is from
$\class{G} \cap \class{FULL}$ since it is equivalent to the original~$Q$.

%
\begin{restatable}{lemma}{lemCliqueRed}
\label{lem:CliqueRed}
  $G$ has a $k$-clique iff $D^{*}_G \models Q(\abf_0)$.
\end{restatable}
\begin{proof}\ The `only if' direction is easy. If $G$ has a
  $k$-clique, then by Point~2 of Theorem~\ref{thm:grohetechnew} and
  since $A$ contains none of the constants from $\abf_0$, there is a
  homomorphism $h$ from $D$ to $D_G$ such that $h(\abf_0)=\abf_0$ and
  $h_0(h(\cdot))$ is the identity.  It is straightforward to extend
  $h$ to a homomorphism from
  %
  %
  $D_1^+$ to $D^*_G = D^+_G \uplus D^+_{\mn{dis}}$; in particular, if
  $R(\abf)$ is an atom in $D$, then the restriction of $h_0$ to $\abf$
  must be injective since $h_0(h(\cdot))$ is the identity and thus a
  copy of the subtree that was attached to $h(\abf)$ in the construction of
  $D_1^+$ from $D_1$ was attached to $h(\abf)$ in the construction of
  $D^+_G$ from~$D_G$.  Consequently, $D_1^+ \models Q_w(\abf_0)$
  implies $D^*_G \models Q_w(\abf_0)$.

  \smallskip

  For the `if' direction, assume that $D^{*}_G \models Q(\bar a_0)$. 
  The surjective homomorphism $h_0$ from $D_G$ to $D$ given by
  Theorem~\ref{thm:grohetechnew} can be extended to a homomorphism
  from $D^*_G=D^+_G \uplus D^+_{\mn{dis}}$ to $D_1^+$ and from
  $\chase{D^*_G }{\Sigma}$ to $\chase{D_1^+}{\Sigma}$.\footnote{The extension
  needs not be surjective.} For brevity, we denote this extension also
  with~$h_0$. For later use, we observe that
  \begin{itemize}

  \item[({\sf P1})] $b \in \adom{D_G}$ iff $h_0(b) \in \adom{D}$; 


  \item[({\sf P2})] $b_1,b_2$ are in the same tree of $D^+_G$ iff 
    $h_0(b_1),h_0(b_2)$ are in the same tree of $D^+_1$.
    
  \end{itemize}
  So $h_0$ witnesses $D^*_G,\abf_0 \rightarrow D_1^+,\abf_0$.  Thus
  $D_1^+ \not\models Q_i(\abf_0)$ for all $i \neq w$ implies
  $D^*_G \not\models Q_i(\abf_0)$ for all
  $i \neq w$. It follows that
  $D^*_G \models Q_w(\abf_0)$ and thus
  $\chase{D^*_G}{\Sigma} \models q_w(\abf_0)$.

By Lemma~\ref{lem:iolem}, we find a contraction $q_c$ of $q_w$ such
that $\chase{D^*_G}{\Sigma} \models^{io}
q_c(\abf_0)$. As witnessed by $h_0$, $\chase{D^*_G}{\Sigma},\abf_0 \rightarrow
\chase{D_1^+}{\Sigma},\abf_0$. Consequently $\chase{D_1^+}{\Sigma} \models
q_c(\abf_0)$. Thus we find a contraction $q'_c$ of $q_c$ such that
$\chase{D_1^+}{\Sigma} \models^{io} q'_c(\abf_0)$. We must have $q_c=q'_c$
since $D_1^+$ satisfies Condition~2 of Lemma~\ref{lem:iomin}, via
Lemma~\ref{lem:Aprop}.

Let $h$ be a homomorphism from $q_c$ to $\chase{D^*_G}{\Sigma}$.  The
composition $g=h_0(h(\cdot))$ is a homomorphism from $q_c$ to
$\chase{D^+_1}{\Sigma}$. Since
$\chase{D^+_1}{\Sigma} \models^{io} q_c$, this homomorphism must be
injective. Consequently, also $h_0$ is injective. 

  \smallskip

  Since $h_0$ is injective, we can construct a function $h_0^-$ as follows:
  \begin{enumerate}

  \item start with the restriction of $h_0$ to the range of $h$;

  \item take the inverse;

  \item restrict to old constants in $D$.

  \end{enumerate}
  It follows from Lemma~\ref{lem:point2lem} that the range of
  $h_0$ includes all old constants in $D$, and thus so does the domain
  of $h_0^-$.  In particular, it thus also contains all constants in
  $\abf_0$.  Trivially, $h_0(h_0^-(\cdot))$ is the identity and
  $h_0^-$ is the identity on all constants from
  $\mn{dom}(D) \setminus A \subseteq \abf_0$ (since $h_0$ is the
  identity on such constants).

  We are going to show that a certain extension of $h_0^-$ is a
  homomorphism from $D$ to $D_G$ that is the identity on
  $\adom{D} \setminus A$, and thus
  Point~2 of Theorem~\ref{thm:grohetechnew} yields that $G$ contains a
  $k$-clique.
  
  Let $R(\abf) \in D$.  By definition of $g$ and since the range of
  $g$
  contains all old constants in $D$, we find for each
  old $a \in \abf$ a variable $x_a \in \bar x$ and a
  $b_a \in \adom{D^*_G}$ such that $h(x_a)=b_a$ and $h_0(b_a)=a$.
  Let $\Gamma$ be the set of all $b_a$.
  By Lemma~\ref{lem:point2lem}, for any
  two old constants $a_1,a_2$ in~\abf, there is a path $x_1,\dots,x_n$
  in the Gaifman graph of $D[q_w]$ such that $g(x_1)=a_1$,
  $g(x_m)=a_2$, and $g(x_2),\dots,g(x_{m-1})$ are in
  $\mn{dom}(D_{\abf})$.
  By Properties $\sf P1$ and
  $\sf P2$, it follows that for any two $b_1,b_2$ in $\Gamma$, there is a path
  $x_1,\dots,x_n$ in the Gaifman graph of $D[q_w]$ such that
  $h(x_1)=b_1$, $h(x_m)=b_2$, and $h(x_2),\dots,h(x_{m-1})$ are in the
  same tree of $D^+_G$.  The fact that $h_0$ maps all elements of
  $\Gamma$ to constants in $D$ and Property~$\sf P1$ further imply
  $b_1,b_2 \in \mn{dom}(D_G)$. By construction of $D^+_G$, there must
  thus be a fact in $D_G$ that contains both $b_1$ and $b_2$ (the one
  for which the tree was added that $h(x_2),\dots,h(x_{m-1})$ are in).
  It follows that $\Gamma$ forms a clique in the Gaifman graph of
  $D_G$. Since (the non-extended) $h_0$ is surjective, for each fresh
  $a \in \abf$ we find some $b_a \in \adom{D_G}$ such that
  $h_0(b_a)=a$. We can now apply Point~3 of
  Theorem~\ref{thm:grohetechnew} to $R(\abf)$ and the $b_a$ to obtain
  an extension of $h_0^-$ to the fresh constants in \abf such that
  $R(h_0^-(\abf))\in D_G$. Note that the set
  $\{b_a \mid a \text{ non-isolated in } D\}$ can only contain old
  constants since all fresh constants are isolated in $D$; it is thus 
  a subset of $\Gamma$ and forms a clique in the Gaifman graph of
  $D_G$.  The extension is such that if $a$ is fresh and covered by
  the extension, then $h_0^-(a)=c$ for some $c$ with $h_0(c)=a$. Thus,
  $h_0(h_0^-(\cdot))$ is still the identity and $h^-_0$ is still the
  identity on $\adom{D} \setminus A$.  Also note that this extension
  can be done independently for all $R(\abf) \in D$ because all fresh
  constants in $D$ are isolated in $D$.

   Finally, 
Point~2 of
  Theorem~\ref{thm:grohetechnew} yields that $G$ has a $k$-clique.
\end{proof}



\section{Proof of Proposition~\ref{pro:from-omq-to-cqs}}
%

Recall that, given two CQSs $S= (\dep,q)$ and $S' = (\dep,q')$ over $\ins{S}$, we write $S \subseteq S'$ if $q(D) \subseteq q'(D)$ for every $\ins{S}$-database $D$ that satisfies $\dep$. We also write $S \equiv S'$ if $S \subseteq S'$ and $S' \subseteq S$.
Analogously, we define the notions that consider unrestricted (not necessarily finite) instances. We write $S \subseteq^{\mathsf{unr}} S'$ if $q(I) \subseteq q'(I)$ for every (possibly infinite) $\ins{S}$-instance $I$ that satisfies $\dep$, and $S \equiv^{\mathsf{unr}} S'$ if $S \subseteq^{\mathsf{unr}} S'$ and $S' \subseteq^{\mathsf{unr}} S$.
The next simple result states that for guarded TGDs the notions $\equiv$ and $\equiv^{\mathsf{unr}}$ coincide, which, unsurprisingly, relies on the fact that guarded TGDs are finitely controllable.

\begin{lemma}\label{lem:cqs-finte-vs-unrestricted-equiv}
	Consider two CQSs $S = (\dep,q(\bar x))$ and $S' = (\dep,q'(\bar x'))$ from $(\class{G},\class{UCQ})$. Then, $S \equiv S'$ iff $S \equiv^{\mathsf{unr}} S'$.
\end{lemma}

\begin{proof}
The $(\Leftarrow)$ direction holds trivially. For the $(\Rightarrow)$ direction, assume that $S \not\equiv^{\mathsf{unr}} S'$. There are two cases: (i) $S \not\subseteq^{\mathsf{unr}} S'$, or (ii) $S' \not\subseteq^{\mathsf{unr}} S$. We consider only the first case; the second case is shown analogously.
Since $S \not\subseteq^{\mathsf{unr}} S'$, we get that $\bar x \not\in q'(\chase{q}{\dep})$. With $Q'$ being the OMQ $\mathsf{omq}(S')$, we get that $\bar x \not\in Q'(D[q])$. Since, by Lemma~\ref{lem:fc-vs-sfc} and Theorem~\ref{the:fg-sfc}, guarded TGDs are finitely controllable, we conclude that there exists a finite model $M$ of $D[q]$ and $\dep$ such that $\bar x \not\in q'(M)$. However, it holds trivially that $\bar x \in q(M)$, and thus, $S \not\subseteq S'$, which in turn implies that $S \not\equiv S'$, as needed.
\end{proof}

Having the above lemma in place, we can now give the proof of Proposition~\ref{pro:from-omq-to-cqs}. Assume that $S = (\dep,q)$ is over a schema $\ins{S}$. Since, by Proposition~\ref{pro:guarded-ucq-k-approximation}, an OMQ from $(\class{G},\class{UCQ})$ is $\class{UCQ}_k$-equivalent iff it is uniformly $\class{UCQ}_k$-equivalent, it suffices to show that, for each $k \geq 1$, the following are equivalent:
\begin{enumerate}
	\item There is $q' \in \class{UCQ}_k$ over $\ins{S}$ such that $S \equiv (\dep,q')$.
	\item There is $q' \in \class{UCQ}_k$ over $\ins{S}$ such that $\mathsf{omq}(S) \equiv (\ins{S},\dep,q')$.
\end{enumerate}

\medskip

\underline{$(1) \Rightarrow (2).$} By hypothesis, there exists $q' \in \class{UCQ}_k$ over $\ins{S}$ such that $S \equiv (\dep,q')$. It is easy to show that $\mathsf{omq}(S) \equiv (\ins{S},\dep,q')$; for brevity, let $Q = \mathsf{omq}(S)$ and $Q' = (\ins{S},\dep,q')$.
Consider an arbitrary $\ins{S}$-database $D$. Since, by construction, $\chase{D}{\dep}$ is a model of $\dep$, and, by Lemma~\ref{lem:cqs-finte-vs-unrestricted-equiv}, $S \equiv^{\mathsf{unr}} (\dep,q')$, we get that $q(\chase{D}{\dep}) = q'(\chase{D}{\dep'})$. Therefore, $Q(D) = Q'(D)$, as needed.

\medskip

\underline{$(2) \Rightarrow (1).$} By hypothesis, there exists $q' \in \class{UCQ}_k$ over $\ins{S}$ such that $\mathsf{omq}(S) \equiv (\ins{S},\dep,q')$; again, let $Q = \mathsf{omq}(S)$ and $Q' = (\ins{S},\dep,q')$. Consider an arbitrary $\ins{S}$-database $D$ such that $D$ is a model of $\dep$. Observe that $q(D) = Q(D)$ and $q'(D) = Q'(D)$. Since $Q \equiv Q'$, we get that $q(D) = q'(D)$, and the claim follows.





\section{Proof of Proposition~\ref{pro:fpt-reduction-omq-to-cqs}}
%

%
Consider an OMQ $Q = (\ins{S},\dep,q(\bar x))$ from $\mathsf{omq}(\class{O})$, an $\ins{S}$-database $D$, and a tuple $\bar c \in \adom{D}^{|\bar x|}$. We are going to construct an $\ins{S}$-database $D^*$ such that $D^* \models \dep$, and $\bar c \in Q(D)$ iff $\bar c \in q(D^*)$, or, equivalently (due to Proposition~\ref{pro:chase}), $\bar c \in \chase{D}{\dep}$ iff $\bar c \in q(D^*)$.

\subsection{Construction of the Database $D^*$}
Note that the definition of $D^*$ has been already given in the main body of the paper (Section~\ref{sec:guarded-cqs-details}).
However, we repeat it here for the sake of readability.
We first define $D^+$ as the database 
\[
D\ \cup\ \{R(\bar a) \in \chase{D}{\dep} \mid \bar a \subseteq \adom{D}\}.
\]
%
Let $A$ be the family of all maximal tuples $\bar a$ over $\adom{D}$
that are guarded in $D^+$, i.e., there is an atom $R(\bar b) \in D^+$ such that $\bar a \subseteq \bar b$.
Fix an arbitrary tuple $\bar a \in A$. Since, by Theorem~\ref{the:fg-sfc}, the class $\class{G}$ is strongly finitely controllable, and also finite witnesses are realizable, we can compute an instance $M(D^{+}_{|\bar a},\dep,n) \in \fmods{D^{+}_{|\bar a}}{\dep}$,
where $n$ is the number of variables in $q$, such that for each CQ $q'$ of arity $|\bar a|$ with at most $n$ variables, it holds that
\begin{itemize}
	
	\item[($*$)] $\bar a \in q'(M(D^{+}_{|\bar a},\dep,n))\ \Longrightarrow\ \bar a \in q'(\chase{D^{+}_{|\bar a}}{\dep})$.
	
\end{itemize}
%
%
W.l.o.g., we assume that $\adom{M(D^{+}_{|\bar a},\dep,n)} \cap \adom{M(D^{+}_{|\bar b},\dep,n)} \subseteq \adom{D}$, for every two distinct tuples $\bar a,\bar b \in A$. The database $D^*$ is 
\[
D^{+}\ \cup\ \bigcup_{\bar a \in A} M(D^{+}_{|\bar a},\dep,n).
\]

\subsection{Correctness of the Reduction}

The next lemma, which is actually Lemma~\ref{lem:reduction-correctness} in the main body of the paper, shows that $D^*$ is the desired database.

\begin{lemma}\label{lem:app-reduction-correctness}
	It holds that:
	\begin{enumerate}
		\item $D^* \models \dep$.
		\item $\bar c \in \chase{D}{\dep}$ iff $\bar c \in q(D^*)$.
		\item There exists a computable function $f : \mathbb{N} \to \mathbb{N}$ such that $D^{*}$ can be constructed in time $\size{D}^{O(1)} \cdot f(\size{Q})$.
	\end{enumerate}
\end{lemma}

\begin{proof}
	\underline{Item (1).} Consider a TGD $\sigma \in \dep$ of the form $\phi(\bar x,\bar y) \ra \exists \bar z \, \psi(\bar x,\bar z)$, with $\guard{\sigma} = R(\bar u)$. Let $q_{\phi} = \exists \bar y \, \phi(\bar x,\bar y)$ and $q_{\psi} = \exists \bar z \, \psi(\bar x,\bar z)$. Fix an arbitrary tuple $\bar c \in \adom{D^*}^{|\bar x|}$, and assume that $\bar c \in q_\phi(D^*)$. We need to show that $\bar c \in q_\psi(D^*)$.
	Since $\bar c \in q_\phi(D^*)$, we get that $\phi(\bar x,\bar y) \ra D^*$ via a homomorphism $h$ such that $h(\bar x) = \bar c$. Clearly, $R(h(\bar u)) \in M(D_{|\bar a},\dep,n)$ for some tuple $\bar a \in A$ that contains all the constants $h(\bar u) \cap \adom{D}$. It is then easy to see that, due to guardedness, $h(\phi(\bar x,\bar y)) \subseteq M(D_{|\bar a},\dep,n)$. Since $M(D_{|\bar a},\dep,n)$ is model of $\dep$, we get that $\bar c \in q_\psi(M(D_{|\bar a},\dep,n))$. Therefore, due to the monotonicity of CQs, $\bar c \in q_\psi(D^*)$, as needed.
	
	\medskip
	
	\underline{Item (2).} Let us first concentrate on the $(\Rightarrow)$ direction. Since $D^* \models  \dep$, we get that $\chase{D}{\dep} \ra D^*$ via a homomorphism $h$ that is the identity on $\adom{D}$; the latter holds due to Proposition~\ref{pro:chase}. Moreover, by hypothesis, $q \ra \chase{D}{\dep}$ via a homomorphism $\mu$ that maps the answer variables $\bar x$ of $q$ to $\bar c$. Therefore, $\mu \circ h$ maps $q$ to $D^*$ and $\bar x$ to $\bar c$, which in turn implies that $\bar c \in q(D^*)$.
	We now show the $(\Leftarrow)$ direction. By hypothesis, $q \ra D^*$ via a homomorphism $h$ such that $h(\bar x) = \bar c$. For each tuple $\bar a \in A$, let $q_{\bar a}(\bar x_{\bar a})$ be the maximal subquery of $q$ such that $h(q_{\bar a}) \subseteq M(D_{|\bar a},\dep,n)$ and $h_{\bar a}(\bar x_{\bar a}) = \bar a$, i.e., $\bar x_{\bar a}$ are the variables of $q_{\bar a}$ that are mapped to constants of $\adom{D}$. By $(*)$, we get that $q_{\bar a} \ra \chase{D_{|\bar a}}{\dep}$ via a homomorphism $h_{\bar a}$ that maps $\bar x_{\bar a}$ to $\bar a$. It should be clear that
	\[
	\mu\ =\ \bigcup_{\bar a \in A} h_{\bar a}
	\]
	is well-defined mapping that maps $q$ to $\chase{D}{\dep}$ with $\mu(\bar x) = \bar c$. This implies that $\bar c \in q(\chase{D}{\dep})$, and the claim follows.
	
	\medskip
	
	\underline{Item (3).} This is shown by exploiting the fact that the database $D^{+}$ can be constructed in time $\size{D}^{O(1)} \cdot g(\size{Q})$ for some computable function $g : \mathbb{N} \to \mathbb{N}$, and the fact that the cardinality of $A$, as well as the size of $D^{+}_{|\bar a}$ for some $\bar a \in A$, do not depend on $D$.
	For showing that $D^{+}$ can be constructed in the claimed time, we exploit Lemma~\ref{lem:GDlog-fpt}, and the fact that we can construct a set $\xi(\dep) \in \class{G} \cap \class{FULL}$ such that $D^+ = \chase{D}{\xi(\dep)}$; the latter is inherited from~\cite{GoRS14}.
\end{proof}

\section{Proof of Proposition~\ref{pro:fr-guarded-ucq-k-approximation-cqs}}

Let $S = (\dep,q(\bar x))$ be a CQS from $(\class{FG}_m,\class{UCQ})$ over a schema $\ins{S}$ of arity $r$. Via a proof similar to that of Proposition~\ref{pro:from-omq-to-cqs}, we get that:
\begin{itemize}
	\item $S$ is uniformly $\class{UCQ}_k$-equivalent iff $\omq{S}$ is.
	\item $S \equiv S_{k}^{a}$ iff $\omq{S} \equiv \omq{S_{k}^{a}}$.
\end{itemize}
Let $Q_S = \omq{S}$ and $Q_{S}^{k} = \omq{S_{k}^{a}}$.
Thus, for showing Proposition~\ref{pro:fr-guarded-ucq-k-approximation-cqs}, it suffices to show that the following are equivalent:
\begin{enumerate}
	\item $Q_S$ is uniformly $\class{UCQ}_k$-equivalent.
	\item $Q_S \equiv Q_{S}^{k}$.
\end{enumerate}

The proof of the next lemma is along the lines of the proof of Lemma~\ref{lem:ucq-k-apx}. The only difference is that we rely on the following fact: given a database $D$ of treewidth at most $k$ up to $\bar c$, $\chase{D}{\dep}$ has treewdith at most $k$ up to $\bar c$. The latter heavily exploits the bound $m$ on the number of head atoms in the TGDs of $\dep$.

\begin{lemma}\label{lem:ucq-k-apx-fg}
	For every OMQ $Q' = (\ins{S},\dep,q')$ from $(\class{FG}_m,\class{UCQ})$ such that $Q'  \subseteq Q_S$, it holds that $Q' \subseteq Q_{S}^{k}$.
	%
\end{lemma}

Having Lemma~\ref{lem:ucq-k-apx-fg} in place, it is now easy to establish the above equivalence, i.e., $Q_S$ is uniformly $\class{UCQ}_k$-equivalent iff $Q_S \equiv Q_{S}^{k}$.
Observe that the $(\Leftarrow)$ direction holds trivially. It remains to show the $(\Rightarrow)$ direction. 
By hypothesis, there exists an OMQ $Q' = (\ins{S},\dep,q')$ from $(\class{FG}_m,\class{UCQ}_k)$ such that $Q_S \equiv Q'$. By Lemma~\ref{lem:ucq-k-apx-fg}, $Q' \subseteq Q_{S}^{k}$, which in turn implies that $Q_S \subseteq Q_{S}^{k}$. The fact that $Q_{S}^{k} \subseteq Q_S$ holds by construction, and the claim follows.


\section{Proof of Theorem~ \ref{the:cqs-frontier-guadred-technical-result}}

We provide a complete proof of Theorem \ref{the:cqs-frontier-guadred-technical-result} here. Our proof does not use Grohe's database 
as defined in \cite{Grohe07}, but a 
variant 
of it which we define below. We start by introducing several important notions.

%
%
%
%
%

%
%
%
%
%
 Let $G = (V,E)$ and $H = (V',E')$ be simple graphs. A {\em minor map} from $H$ to $G$ is a mapping $\mu : V' \to 2^V$ such that:  
 \begin{enumerate} 
\item
For every $v' \in V'$, $\mu(v')$ is nonempty and connected in $G$. 
\item
For every $v',u' \in V'$, with $v' \neq u'$, $\mu(v')$, $\mu(u')$ are disjoint.
The sets of the form $\mu(v')$, for $v' \in V'$, are pairwise disjoint. 
\item 
For every edge $\{v',u'\} \in E'$, there are nodes $v \in \mu(v')$ and $u \in \mu(u')$ such that 
$\{v,u\} \in E$. 
\end{enumerate} 
Such a minor map is said to be {\em onto} if, in addition, $\bigcup_{v' \in V'} \mu(v') = V$.

Let $G$ and $H$ be simple graphs. It is easy to see that 
$H$ is a minor of $G$ iff there is a minor map $\mu$ from $H$ to $G$. If, in addition, $G$ is connected, 
then there exists an onto minor map $\mu$ from $H$ to $G$. 


\subsection{A Variation of Grohe's Database}

Let  $G = (V,E)$ be an undirected graph, $k \geq 1$ be an integer, and $K=\binom{k}{2}$. 
Suppose that $D$ and $D'$ are $\ins{S}$-databases with $D \subseteq D'$ such that there exists a set $A \subseteq\adom{D}$  for which
there is a minor map $\mu$ from the $(k\times K)$-grid onto $G^{D}_{|A}$. 
We fix a bijection $\chi$ that assigns to each $2$-element subset of $[k]$ a value in $[K]$.
We define a database 
$D^* := D^*(G,D,D',A,\mu)$ as follows. 


\begin{description}
        \item \underline{{\em The domain.}}  Every element of $D^*$ is either an element in $\adom{D'} \setminus A$ or
a tuple of the form $(v,e,i,p,z)$ where $v\in V$, $e\in E$, $i\in [k]$, $p$ is a two-element subset of $[k]$, and 
$z \in A$. We note here that 
$\adom{D^*}$ does not necessarily contain all elements of the previous form since some of them might not appear in any of its atoms.

        \item \underline{{\em The facts.}}  A {\em labelled clique} in $G$ is any partial mapping $\eta$ from $[k]$ to $V$ such that, for every different 
        $i,j$ in the domain of 
$\eta$, we have that 
$\eta(i)$ and $\eta(j)$ are adjacent in $G$. We say that an element $z\in\adom{D'}$ is 
{\em covered} by a labelled clique $\eta$ if $z\not\in A$, or there exist $i,j,\ell$ in domain of $\eta$ such
that $z\in\mu(i,\chi(\{j,\ell\}))$. Then, for every fact $R(\bar z) \in D'$ and every labelled clique $\eta$ covering every element mentioned in $\bar z$, it is the case that $D^*$ contains the fact $R(\bar z_{\eta})$ where $\bar z_{\eta}$ is obtained by replacing, in $\bar z$, every element $z\in A$ with $$(\eta(i),\{\eta(j),\eta(\ell)\},i,\{j,\ell\},z),$$ where $\mu(i,\chi\{j,\ell\})=z$.

    \end{description}

Additionally, we define the projection 
$$h_0 \, : \, \adom{D^*} \to \adom{D'}$$ such that 
$$\begin{cases} 
h_0(v,e,i,p,z)=z, & \text{ if $z \in A$, and } \\ 
h_0(y) = y, & \text{otherwise.} 
\end{cases} 
$$

Although we avoid to introduce Grohe's original database~\cite{Grohe07}, it is useful to compare it with ours. In particular,
consider the case when $D=D'$ and $A=\adom{D}$, and thus we have 
a database $D^* = D(G,D,D,\adom{D},\mu)$. Grohe's database $D_{\rm Grohe} = D_{\rm Grohe}(G,D,\mu)$ 
then satisfies the following: 
\begin{itemize} 
\item 
$\adom{D_{\rm Grohe}}$ is contained in the set of tuples of the form $(v,e,i,p,z)$, for $v\in V$, $e\in E$, $i\in [k]$, 
$p$ is a two-element subset of $[k]$, and $z\in\adom{D}$. Therefore, $h_0$ is well-defined on $\adom{D_{\rm Grohe}}$.  
\item 
Also, $D^* \subseteq D_{\rm Grohe}$. 
\end{itemize} 
In addition, the following crucial property holds for $D_{\rm Grohe}$. 

\begin{lemma}{\em (\cite{Grohe07})} 
\label{le:fromGrohe}
Let $h$ be an homomorphism from $D$ to $D_{\rm Grohe}$ such that $h_0 \circ h$ is the identity. Then $G$ contains a $k$-clique.
\end{lemma}

Our variant still enjoys many of the good properties of Grohe's database, and some additional properties concerning the satisfaction of frontier-guarded TGDs, as shown in the following Theorem (which is essentially an equivalent reformulation of Theorem~\ref{lemma:grohe-correct-fg-body}) for the case when CQs allow for free variables. 

 \begin{lemma}
\label{lemma:grohe-correct}
Let $G$, $k$, $D$, $D'$, $A$, $D^* = D^*(G,D,D',A,\mu)$, and $h_0$ be as defined above. The following statements hold:
	   \begin{enumerate} 
		 \item There is a polynomial-time algorithm that, given $G$, $D$, $D'$, $A$, and $\mu$, computes $D^*$. 
		
		\item $h_0$ is a surjective homomorphism from $D^*$ to $D'$. 
		
		\item $G$ contains a $k$-clique iff there is a homomorphism $h$ from $D$ to $D^*$ such that $h_0 (h(\cdot))$ is the identity on $A$. 
		
		\item If $D' \models \Sigma$, and every clique of size at most 
		$3 \cdot r$ in $G$ is contained in a clique of size $3 \cdot r \cdot m$, then $D^* \models \Sigma$.
	\end{enumerate} 
\end{lemma}

\begin{proof}
Items (1) and (2) follow directly from the definition of $D^*$ and $h_0$. However, for item (1) it is necessary to note that since the 
arity of the schema $\ins{S}$ is bounded by the fixed integer $r$, 
in the construction of $D$ it is enough to consider only all labelled cliques of size at most $r$.

Let us prove $(3)$. Although ($\Leftarrow$) can be obtained easily by mimicking the proof in \cite{Grohe07}, we shall give a proof using Lemma~\ref{le:fromGrohe}. Assume that $h$ is a homomorphism from $D$ to $D^*$ such that $h_0 \circ h$ is the identity on $A$. 
Then the restriction $h_{|A}$ of $h$ to $A$ defines a homomorphism from $D_{|A}$ to $D^*$ such that 
$h_0 \circ h_{|A}$ is the identity. But then 
$h_{|A}$ corresponds to a homomorphism from $D_{|A}$ to $D(G,D_{|A},D_{|A},\adom{D_{|A}},\mu)$ with $h_0 \circ h_{|A}$ being the identity. Since 
$D(G,D_{|A},D_{|A},A,\mu)$ is contained in the original Grohe's database, $D_{\rm Grohe}(G,D_{|A},\mu)$, it follows from Lemma~\ref{le:fromGrohe} that $G$ contains a $k$-clique. For the converse ($\Rightarrow$) assume that $G$ has a $k$-clique, or alternatively, that there is a labelled clique $\eta$ in $G$ with domain $[k]$.
Then $\eta$ covers all elements in $\adom{D'}$ and, hence, 
the mapping $h : \adom{D} \to \adom{D^*}$ with $h(z) = z_{\eta}$, for every 
$z \in \adom{D}$, defines a homomorphism from $D$ to $D^*$ satisfying the conditions.


Let us finally show (4). Consider an arbitrary frontier-guarded TGD $\phi(\bar x,\bar y) \rightarrow \exists \bar z\psi(\bar x,\bar z)$ in $\Sigma$, 
and let $(\bar a,\bar b)$ be an arbitrary tuple in the evaluation of $\phi(\bar x,\bar y)$ over $D^*$. Since $h_0$ is a homomorphism 
from $D^*$ to $D'$, we have that $D'$ contains 
all facts in $\phi(h_0(\bar a,\bar b))$.
Therefore, since $D' \models \Sigma$, there is a tuple $\bar d$ of elements in $\adom{D'}$, for 
$|\bar d| = |\bar z|$, 
such that $D'$ contains every atom in 
$\psi(h_0(\bar a),\bar d)$. Since there is an atom guarding $\bar{x}$ in $\phi$, it follows from the definition of $D^*$ that $\bar a=h_0(\bar a)_{\eta}$ for some labelled clique $\eta$ covering $h_0(\bar a)$. In addition, $\Sigma \in \class{FG}_m$ and hence $(h_0(\bar a),\bar d)$ can have at most $r \cdot m$ elements. It follows from the assumptions of item (4) that there is labelled clique $\eta'$ in $G$, extending $\eta$, that covers all elements mentioned by $\bar d$ (in addition to all elements in $h_0(\bar a)$ already covered by $\eta$).
Thus, $D^*$ contains every atom in $\psi(h_0(\bar a)_{\eta'},\bar d_{\eta'})=\psi(\bar a,\bar d_{\eta'})$.
\end{proof}

\subsection{A Crucial Lemma}
%

Our proof borrows several ideas and techniques from the proof of Grohe's Theorem. However, the adaptation of such techniques is non-trivial for the reasons that we have already discussed in Section~\ref{sec:frontier-guarded-results}, and we briefly recall here again. 
First, we cannot construct an arbitrary database $D$, but we need to construct one that satisfies the given set $\dep$ of frontier-guarded TGDs.
Second, the notion of core of a CQ, which is crucial for the proof of Grohe, cannot be directly used 
in our context. 
%
Instead, we need to develop a technical lemma, which, intuitively speaking, states that some subsets of our CQs 
behave like cores for the sake of our proof.


\begin{restatable}{lemma}{lemmamain}	\label{le:main} 
	Fix $\ell \geq r \cdot m$. There is a computable function that takes as input a CQS $S = (\dep,q(\bar x))$ from $(\class{FG}_m,\class{CQ})$ that is {\em not} uniformly $\class{CQ}_\ell$-equivalent, and a positive integer $s \geq 1$, and outputs a CQ $p(\bar x)$, a subset $X$ of the 
	existentially quantified 
	variables of $p$, and a
	CQ $p'(\bar x)$, such that:
	\begin{enumerate}
		\item $q \equiv_\Sigma p$.
		\item $D[p'] \models \Sigma$.
		\item $D[p] \subseteq D[p']$.
		\item For every homomorphism $h$ from $p$ to $p'$ with $h(\bar x) = \bar x$, we have that $h(X) = X$.
		\item The treewidth of $G^p_{|X}$ is larger than $\ell$.
		\item For every CQ $p''(\bar x)$ with at most $s$ variables, if $\bar x \in p''(D[p'])$ then also $\bar x \in 
		p''(\chase{D[p]}{\Sigma})$. 
	\end{enumerate}
\end{restatable}


\begin{proof} 
Let us assume that $S$ is defined over schema $\ins{S}$, that $q(\bar x)$ has $n > 0$ variables, and 
that $t$ is the maximum number of variables appearing in the 
body 
of some TGD in $\Sigma$. 
We set $s' = \max{\{s,n,t\}}$.

First of all, we can compute by exhaustive search a CQ $p(\bar x)$ with a minimum number of variables that satisfies 
$q \equiv_\Sigma p$. This is because the notion $\equiv_\Sigma$ is decidable, for $\Sigma \in  \class{FG}$. In fact, to check whether $p \subseteq_\Sigma q$ one can  construct the finite instance 
$M = M(D[p],\Sigma,n)$ 
from Definition \ref{def:sfc}, which exists by the strong finite controllability of $\class{FG}$ as stated in Theorem  \ref{the:fg-sfc}, 
and then check whether $\bar x \in q(M)$. We know that the latter holds iff  $\bar x \in q(\chase{D[p]}{\Sigma})$. 
From Proposition~\ref{pro:cqs-cont-chase} we get that $p \subseteq_\Sigma q$. Analogously, one can check $q \subseteq_\Sigma p$.

We have then that item (1) holds by definition.  We now explain how to build $p'$ and $X$.
We need some preparation. First, we shall assume that for every TGD $\sigma$ in $\Sigma$ the new TGD $\sigma'$ obtained by identifying two variables in the body of $\sigma$ belongs also to $\Sigma$. We can assume this without loss of generality because $\sigma'$ is implied by $\sigma$ and, furthermore, there is an easy procedure that iteratively adds to $\Sigma$ all new TGDs that can be obtained by identifying two variables in the body.  Also, we build a new relational schema $\ins{S'}$ and set $\Sigma'$ of frontier-guarded TGDs in the following way. Initially, set 
$\ins{S'}=\ins{S}$ and $\Sigma'=\Sigma$. Then, for each $\sigma=\phi(\bar x',\bar y)\rightarrow\exists\bar w\psi(\bar x',\bar w)$ in $\Sigma$, for every partition $\chi\wedge\chi'=\phi$ of the atoms in $\phi$, and for every CQ $q'(\bar z)=\exists \bar v\chi(\bar v,\bar z)$ obtained by quantifying existentially some variables in $\chi$, we include in $\ins{S'}$ a 
new predicate $T_{q'}$ with the same arity than $\bar z$. Furthemore, for every atom $R(\bar u)$ where $R \in \ins{S}$ and $\bar u$ might contain variables occurring in $\phi$ and new fresh variables, we include in $\Sigma'$ two new TGDs $\sigma'$ and $\sigma''$ (provided they are in $\class{FG}$). TGD $\sigma'$ is obtained by replacing, in $\phi$, $\chi(\bar v,\bar z)$ by $T_{q'}(\bar z)\wedge R(\bar u)$, whereas $\sigma''$ is the TGD $\phi(\bar v,\bar z)\wedge R(\bar u)\rightarrow T_{q'}(\bar z)$.

The following claim follows easily from the definition.

\begin{claim} \label{claim:same} 
For every $\ins{S}$-instance $D$, and every atom $R(\bar c)$ where $R\in\ins{S}$ and $\bar c$ only mentions elements from $\adom{D}$, $$R(\bar c) \in \chase{D}{\Sigma} \ \ \Longleftrightarrow \ \ R(\bar c) \in \chase{D}{\Sigma'}.$$ 
\end{claim}

We use ${\mathcal A}$ to denote the collection of all 
 subsets $A$ from $\adom{p}$ that are guarded, i.e., 
 there exists an atom $R(\bar a)$ of $p$ such that 
 $\bar a$ mentions all elements in $A$. For every $A\in{\mathcal A}$, let $E_A$ be the $\ins{S'}$-database that contains all facts $R(\bar b) \in \chase{D[p]}{\Sigma'}$ for which it is the case that all elements from $\bar b$ appear in $A$. 
 Notice that, by Theorem~\ref{the:fg-sfc}, each $E_A$ can be constructed from $D[p]$, $\Sigma'$ and $A$.
 It follows also
  from Theorem \ref{the:fg-sfc} that we can construct a finite 
 $\ins{S}'$-instance 
$M_A = M(E_A,\Sigma',s')$  
 with the property that, for every CQ $p''(\bar x'')$ with at most $s'$ variables, $$\bar a \in p''(M_A) \ \ \Longleftrightarrow \ \ \bar a \in  p''(\chase{E_A}{\Sigma'}),$$ 
for every tuple $\bar a$ of elements in $\adom{E_A}$ of the same arity as $\bar x''$.

Let us define $D_A$ by removing from $M_A$ those atoms that are not over $\ins{S}$. 
We can assume (renaming elements if necessary) that  
 $\adom{D_A}\cap\adom{D_{A'}}=A\cap A'$, for every $A,A'\in{\mathcal A}$.
We then define $p'$ as $\cup_{A\in {\mathcal A}} D_A$. 
We have by construction that $D[p]\subseteq D[p']$, i.e., item (3) holds. 

We now show item (6). We actually prove something stronger:

\begin{claim} 
\label{claim:aux1}
Let $p''(\bar x)$ be a CQ with at most $s$ variables 
and $\bar c \in \adom{p}^{|\bar x|}$  
 such that $\bar c \in p''(D[p'])$. 
Then 
$\bar c \in p''(\chase{D[p]}{\Sigma})$. 
\end{claim} 

\begin{proof}
Assume that $\bar c \in p''(D[p'])$, i.e., there exists a homomorphism $h$ from $p''$ to $p'$ with $h(\bar x) = \bar c$. 
Let $p''(\bar x)=\exists\bar y\phi(\bar x,\bar y)$ and for every $A\in{\mathcal A}$,
let $\phi_A$ be the conjunction of all atoms, $R(\bar z)$, in $\phi$ such that $h(\bar z)$ is entirely contained in $\adom{D_A}$. 
Notice that if we define $p''_A$ as the CQ $\exists \bar y_A\phi_A(\bar x_A,\bar y_A)$, 
where $\bar x_A$ and $\bar y_A$ are  all variables of $\phi_A$ occurring in $\bar x$ and $\bar y$, respectively, then 
$p'' \equiv \bigwedge_{A\in{\mathcal A}} p''_A$. 


Since $h(\bar x_A) \in p''_A(D_A)$, we have by definition that also $h(\bar x_A) \in p''_A(M_A)$. Since 
$p''_A$ has at most $s\leq s'$ variables, it follows then from the definition of $M_A$  
that $h(\bar x_A) \in p''_A(\chase{E_A}{\Sigma'})$, and hence 
$h(\bar x_A) \in p''_A(\chase{D[p]}{\Sigma'})$. 
%
Since $A$ was arbitrarily chosen, we can then conclude that 
$h(\bar x) \in p''(\chase{D[p]}{\Sigma'})$. 
Consequently, since $p''$ only contains predicates in $\ins{S}$ it follows directly from Claim~\ref{claim:same} 
that 
$h(\bar x) \in p''_A(\chase{D[p]}{\Sigma})$. 
\end{proof}

We then show that item (2) holds.  
\begin{claim} 
It is the case that $D[p'] \models \Sigma$. 
\end{claim} 

\begin{proof} 
Let $\sigma=\phi(\bar x',\bar y) \rightarrow \exists \bar w \, \psi(\bar x',\bar w)$ be any TGD in $\Sigma$ and let $h$ be any mapping of the variables in $\phi$ to elements in $\adom{p'}$ such that $\phi(h(\bar x'),h(\bar y))\subseteq p'$. Our goal is to show that $h(\bar x)$ belongs to the evaluation of $\exists \bar w \, \psi(\bar x',\bar w)$ in $D[p']$. We can assume that 
$h$ is injective (no two variables from $\phi$ are mapped to the same element in $\dom(p')$) since the non-injective case is reduced to the injective one. Indeed, if $h$ is non-injective we only need to replace $\sigma$ by the TGD obtained by identifying variables in the body that are mapped by $h$ to the same element (which by assumption also belongs to $\Sigma'$) and also modify $h$ accordingly.

Since $\Sigma \in \class{FG}$, it follows that there exists some $A\in {\mathcal A}$ such that
$h(\bar x')$ is entirely contained in $\adom{D_A}$. If $h(\bar y)$ is also contained in $\adom{D_A}$ then there is nothing else to prove since, by construction, $D_A$ satisfies $\Sigma$. 

Assume, then, that $h(\bar y)$ is not contained in $\adom{D_A}$. Let $R(\bar a)$ be the atom in $p$ that guards $A$.  Let $\bar a=(a_1,\dots,a_s)$ and let $\bar u=(u_1,\dots,u_s)$ be a tuple of variables where for every $1\leq i\leq s$, $u_i$ is a new fresh variable if $a_i$ does not belong to the image of $h$, and $u_i=h^{-1}(a_i)$ otherwise (note that here we are using the fact that $h$ is injective). 

Let $\bar y'$ be a tuple  containing all variables in $\bar y$ that are not mapped by $h$ inside $\adom{D_A}$, let $\chi(\bar y',\bar z)$ be the conjunction of all atoms in $\phi$ containing some variable in $\bar y'$ (where $\bar z$ is a tuple that contains the rest of the variables in $\chi$), and let $q'(\bar z)=\exists \bar y' \chi(\bar y',\bar z)$. Note that $h(\bar z) \subseteq A$ as, by construction, if a tuple in $p'$ contains elements in 
$\adom{D_A} \setminus A$ then all of its elements must necessarily be contained in $A$. Since every element mentioned in $h(\bar z)$ belongs to $A$, it follows that the TGD $\sigma'=\chi(\bar z,\bar y)\wedge R(\bar u)\rightarrow R_{q'}(\bar z)$ is frontier-guared by $R(\bar u)$ and, consequently, it belongs to $\Sigma'$. Let $h'$ be the extension of $h$ so that every fresh variable $u_i$ in $\bar u$ is mapped by $h'$ to $a_i$. Note that the image of $h'$ is contained in  $\adom{p'}$ and hence it follows, by applying TGD $\sigma$ with mapping $h'$, that  $R_{q'}(h(\bar z))\in (\chase{D[p]}{\Sigma'})$. Consequently, $R_{q'}(h(\bar z))$ belonts to $E_A$ and, hence, to $M_A$, as well.

Consider the TGD $\sigma''$ obtained by replacing $\chi(\bar y',\bar z)$ by $R_{q'}(\bar z)\wedge R(\bar u)$ in the body, $\phi(\bar x',\bar y)$, of $\sigma$. That is, $\sigma'' = \chi'(\bar x',\bar z) \wedge R_{q'}(\bar z) \wedge R(\bar u)\rightarrow \exists \bar w \, \psi(\bar x',\bar w)$, where $\chi'$ is the set of all atoms in $\phi$ that are not in $\chi$. We shall prove that $\sigma''$ is frontier-guarded (and hence, it belongs to $\Sigma'$). Clearly $\bar x'$ is guarded by some atom in $\chi$ or $\chi'$. In the latter case there is nothing else to prove so we can assume that $\chi$ contains some atom guarding $\bar x'$. By definition, this atom contains some elements not in $D_A$ which implies that $\bar x'$ is entirely contained in $A$.  Consequently $\bar x'$ is guarded by $R(\bar u)$. 

Now, since $R_{q'}(h(\bar z))\in M_A$ it follows, by applying $\sigma''$ with mapping $h'$, that $h(\bar x')$ belongs to the evaluation of 
$\exists \bar w \psi(\bar x',\bar w)$ over $M_A$. Since all atoms in $\psi$ involve only predicates from $\ins{S}$, it follows that $h(\bar x')$ belongs to the evaluation 
of $\exists \bar w \psi(\bar x',\bar w)$ over $D_A$, and hence over $D[p']$.
\end{proof} 

Before continuing, we prove the following important 
properties of the homomorphisms from $p$ to $p'$. 
\begin{claim}
\label{cl:Notwo}
Let $h$ be any homomorphism from $p$ to $p'$ with $h(\bar x) = \bar x$. 
The following statements hold: 
\begin{enumerate}[label=(\alph*)]
\item It is the case that 
$p \equiv_\Sigma p'_{|\adom{h(p)}}$. 
\item $h$ is injective. 
\item For every $A\in{\mathcal A}$,  it is the case that 
$(D_A)_{|\adom{h(p)}\cap\adom{D_A}}$ has treewidth at most $r \cdot m$. 
\item For every  
$Y\subseteq \adom{p} \setminus \bar x$, 
 the treewidth of $G^p_{|Y}$ is at most the maximum between $r \cdot m$ and the treewidth of $G^{p}_{|h(Y)\cap \adom{p}}$.  
\end{enumerate} 
\end{claim}

\begin{proof}
Consider first item (a). Clearly $p'_{|\adom{h(p)}} \subseteq_\Sigma p$ since $h$ defines a homomorphism from $p$ to $p'_{|\adom{h(p)}}$. For 
the converse, we have $\bar x \in p'(D[p'])$ and, therefore, $\bar x \in p'_{|\adom{h(p})}(D[p']) $. Since $p$, and thus 
$p'_{|\adom{h(p)}}$, has at most has at most $n \leq s'$ variables, it follows from  
Claim \ref{claim:aux1}
that $\bar x \in p'(\chase{D[p]}{\Sigma})$, i.e., $p \subseteq_\Sigma p'_{|\adom{h(p)}}$.  
%
 Item (b) holds due to the fact that $p \equiv_\Sigma p'_{|\adom{h(p)}}$ and the minimality 
 condition we have imposed in the definition of $p$. 

Let us prove item (c). Let $A\in{\mathcal A}$. Since $h(p)$ is of size at most $n \leq s'$, 
it follows directly from the construction of $D_A$ that 
there exists a homomorphism $g_A$ from $(D_A)_{|\adom{h(p)}\cap\adom{D_A}}$ to $\chase{E_A}{\Sigma'}$ 
that acts as the identity on $A \cap \adom{h(p)}$ (simply take the CQ induced in $D_A$ by $h(p)$ with all variables corresponding 
to elements in $A \cap \adom{h(p)}$ being free). 
Let us denote by $F_A$ the subinstance of $\chase{E_A}{\Sigma'}$ that contains only those 
atoms $R(\bar a)$, for $R\in\ins{S}$, such that all the elements mentioned in 
$\bar a$ belong to the image of $g_A$. 
We show next that there is a homomorphism $f_A$ from  $F_A$ to $\chase{p}{\Sigma}$ that acts as the identify on $A$. 

Let $p_A(\bar x_A)$ be a CQ defined as the conjunction of 
atoms in $F_A$, where every constant $c$ is replaced by a variable $x_c$ and 
$\bar x_A$ is the tuple of variables that represent the elements $\bar a$ in 
%
%
$\adom{F_A}\cap A$. 
It is then the case that $\bar a \in p_A(\chase{E_A}{\Sigma'})$. 
By definition of $E_A$, this means that $\bar a \in p_A(\chase{p}{\Sigma'})$. 
But $p_A$ only mentions relation symbols in $\ins{S}$, and hence it is the case that 
$\bar a \in p_A(\chase{p}{\Sigma})$ from Claim \ref{claim:same}.
This implies our claim for the existence of $f_A$ as described above. 


Consequently, the mapping $f$ that sends every element $a$ in $\adom{p}$ to 
$f_A\circ g_A\circ h(a)$, where $A$ is any set in ${\mathcal A}$ with $a\in A$, is well-defined and corresponds to a 
homomorphism from $p$ to 
$\chase{p}{\Sigma}$. Clearly, $\chase{p}{\Sigma}_{|\adom{f(p)}}\equiv_\Sigma p$ and, hence, by the 
minimality of $p$ it follows that $f$ is injective. 

Let $A\in{\mathcal A}$. Since $|\adom{E_A}| \leq r$ and $\Sigma'\in\class{FG}_m$, it follows directly 
that $\chase{E_A}{\Sigma'}$, and thus $F_A$, has a treewidth at most $r \cdot m$. 
 Since $f$ is injective, then 
$g_A$ is also injective. Therefore, $(D_A)_{|\adom{h(p)}\cap\adom{D_A}}$, which is the preimage of $g_A$, also 
has treewidth at most 
$r \cdot m$. This establishes item (c). 

We now prove item (d). 
Let $(T,\chi)$ be a tree decomposition of $G^{p'}_{|h(Y)\cap \adom{p}}$. 
We shall extend it into a tree decomposition of 
 $G^{p'}_{|h(Y)}$ in the following way. Let $A\in{\mathcal A}$. We know from item (c) that  
 $(D_A)_{|\adom{h(p)}\cap\adom{D_A}}$  
has treewidth at most 
$r \cdot m$. This implies 
 $$H_A \, := \, G^{p'}_{|h(Y) \cap \adom{D_A}}$$
has a tree decomposition $(T_A,\chi_A)$ of width at most $r \cdot m$ (as $A$ is guarded, and thus every 
pair of elements 
in $G^{p'} \cap A$ is linked by an edge). Moreover, 
since $A$ is a guarded in $\adom{p}$,   
the elements in $A \cap h(Y)$ define a clique on $H_A$. Hence, there is a node 
$v \in T_A$ such that $A \cap h(Y) \subseteq \chi_A(v)$. Likewise, there is $u \in T$ 
such that $A \cap h(Y) \subseteq \chi(u)$. 
%
Thus, we can construct a new tree decomposition by joining $u$ and $v$ with an edge. Iterating this process for all $A\in{\mathcal A}$ we get a tree decomposition of $G^{p'}_{|h(Y)}$ whose width is the maximum between $r \cdot m$ and the width of $T$. Since $h$ is injective by item (b), $G^p_{|Y}$ has a tree decomposition of the same width as $G^{p'}_{|h(Y)}$.
\end{proof}

Let $k$ be the treewidth of $G^p$. Since $q\equiv_\Sigma p$ and $q$ is not uniformly $\class{CQ}_\ell$-equivalent, it
follows that $k> \ell$. A set  $Y\subseteq\adom{p} \setminus \bar x$ is {\em nice} if $G^p_{|Y}$ has treewidth $k$ 
and $|Y|$ is minimal with such property. We construct $X$ as the set containing all elements that belong to some nice set. 
It follows directly from the definition that property (5) holds. To finish, 
we show that property (4) is satisfied.

\begin{claim}
Let $h$ be a homomorphism from $p$ to $p'$ with $h(\bar x) = \bar x$. Then $h(X)=X$
\end{claim}

\begin{proof}
Let $Y$ be any nice set. From Claim \ref{cl:Notwo}(c), the treewidth of $G^p_{|Y}$ is at most the maximum between $r$, $m$, and the treewidth of $G^{p}_{|h(Y)\cap \adom{p}}$. Since $G^p_Y$ has treewidth $k$ and $k>\ell\geq r \cdot m$, the treewidth of $G^{p}_{|h(Y)\cap \adom{p}}$ is at least $k$, implying that it is, in fact, $k$ since the treewidth of $G^p$ is $k$. 
By the minimality of $|Y|$, 
it must be the case then that $h(Y) \subseteq \adom{p}$. Therefore, $h(Y)$ is also a nice set and, hence, $h(Y) \subseteq X$. Since $h$ is injective from Claim \ref{cl:Notwo}(b), it must be the case that $h(X) = X$. 
\end{proof}

 This finishes the proof of Lemma \ref{le:main}. 
 \end{proof}

We now proceed to explain how Lemma~\ref{le:main} is applied in order to prove Theorem~\ref{the:cqs-frontier-guadred-technical-result}.

\subsection{The FPT-reduction}
We have all the ingredients needed 
to define the fpt-reduction used in the proof of Theorem~\ref{the:cqs-frontier-guadred-technical-result}. 
Let $(G,k)$ be an instance of {\sf p}-{\sf Clique}. It is easy to see that 
we can assume, without loss of generality, that every clique of size at most $3 \cdot r$ in $G$ is contained in a clique of size $3 \cdot r \cdot m$; recall that $r$ is the maximum arity of the predicates occurring in CQSs of $\class{O}$, while $m$ is the maximum number of atoms in the head of the TGDs occurring in CQSs of $\class{O}$. 

\begin{claim} \label{prop:wlog} 
{\sf p}-{\sf Clique} is {\sf W}$[1]$-hard, even if restricted to instances that satisfy the above restriction. 
\end{claim} 

Henceforth, for $k \geq 1$ we let $K = \binom{k}{2}$.
From the Excluded Grid Theorem \cite{EGT}, 
%
there is a computable function 
$F : {\mathbb N} \to {\mathbb N}$ such that, for each $k \geq 1$ and simple graph $G$ of treewidth at least $F(k)$, we have that at least some connected component of $G$ contains a $(k \times K)$-grid as a minor.
By hypothesis on the class $\class{O}$, 
there exists a CQS $S = (\Sigma,q'(\bar x))$ from $\class{O}$ such that $S \not\in (\class{FG}_m,\class{UCQ})_{F(k)}^{\equiv}$, where $F$ is as 
defined above. W.l.o.g., we assume that $F(k) \geq r \cdot m$. 
We build from $(G,k)$ an instance $(D^*,\Sigma,q'(\bar x))$ of $\text{\rm {\sf p}-{\sf CQS}-{\sf Evaluation}($\class{O}$)}$, 
where $D^*$ is a database defined as follows.

%

First, observe that there must be a CQ $q(\bar x)$ in $q'(\bar x)$ such that $(\Sigma,q(\bar x)) \not\in (\class{FG}_m,\class{CQ})_{F(k)}^{\equiv}$, while $q \not\subseteq_\Sigma \hat{q}$ for every disjunct $\hat{q}$ of $q'$ other than $q$. In fact, if we remove from $q'$ every disjunct that is not maximal with respect to $\subseteq_\Sigma$, then we obtain an equivalent UCQ $q''(\bar x)$ under $\Sigma$. Therefore, $(\Sigma,q''(\bar x)) \not\in (\class{FG}_m,\class{UCQ})_{F(k)}^{\equiv}$, and hence, there is at least one disjunct $q$ in $q''$ such that $(\Sigma,q(\bar x)) \not\in (\class{FG}_m,\class{CQ})_{F(k)}^{\equiv}$.   

Since, by assumption, 
$F(k) \geq r \cdot m$, 
it is possible to compute from $q(\bar x)$ and $s$, where $s$ is the maximum number of variables over all CQs of $q'$, a 
CQ $p(\bar x)$,  
a subset $X$ of the existentially quantified 
variables of $p$, and a
CQ $p'(\bar x)$, that satisfy the properties stated in Lemma~\ref{le:main} for $\ell = F(k)$. 
In particular, the treewidth of $G^p_{|X}$ is at least $F(k)$, 
and hence, by the Excluded Grid Theorem,
there exists a connected component $H$ of $G^p_{|X}$ and a 
minor map $\mu$ from the $(k \times K)$-grid onto $H$. 
%
We then define $D^*$ as $D^*(G,D[p],D[p'],X,\mu)$, 
where $D^*(G,D[p],D[p'],X,\mu)$ is our modified version of Grohe's database that satisfies the properties 
stated in Lemma~\ref{lemma:grohe-correct}.

\subsection{Correctness of the Reduction}

It remains to show that the above is an fpt-reduction from {\sf p}-{\sf Clique} to {\sf CQS}-{\sf Evaluation}($\class{O}$). To this end, we need to show the following. 

\begin{restatable}{lemma}{lemmacorrectness}
\label{lem:fg-cqs-correcteness}
	The following statements hold:
\begin{enumerate} 
\item $D^* \models \Sigma$. 

\item $G$ has a $k$-clique iff $\bar x \in q'(D^*)$. 

\item There are computable functions $f,g:\mathbb{N}\to\mathbb{N}$ such that  $(D^*,\Sigma,q'(\bar x))$ 
can be constructed in time $f(k) \cdot \size{G}^{O(1)}$ and $(\size{q'} + \size{\Sigma}) \leq g(k)$.
\end{enumerate} 
\end{restatable}

%
%
%

\begin{proof} 
The proof of (1) follows from the last item in Theorem~\ref{lemma:grohe-correct} since 
$D[p'] \models \Sigma$ and every clique of size at most $3 \cdot r$ 
in $G$ is contained in a clique of size $3 \cdot r \cdot m$.  It is worth noticing though that this holds for our variant 
of Grohe's database, but not necessarily for the one that Grohe originally defined in~\cite{Grohe07}. 

We now proceed to show item (2).
\begin{itemize} 
\item[]
($\Rightarrow$) Assume that $G$ has a $k$-clique. Then, by Lemma~\ref{lemma:grohe-correct}, there is a homomorphism 
$h$ from $p$ to $D^*$ such that $h_0 \circ h$ is the identity on $X$. 
It follows directly from the definition of $D^*$ that
the mapping $h'$ defined as $h'(x)=h(x)$, when $x \in X$, and $h'(x)=x$, otherwise, is also a homomorphism from $p$ to $D^*$.
In particular, since $X$ contains only quantified variables we have that $h'(\bar x) = \bar x$, and hence 
$\bar x \in p(D^*)$. 
%
%
%
But $q \equiv_\Sigma p$, and thus, $\bar x \in q(D^*)$ since, by item (1), $D^* \models \Sigma$. Therefore, $\bar x \in q'(D^*)$ as $q$ is a disjunct of $q'$. 

\medskip

\item[] 
($\Leftarrow$) Conversely, assume that $\bar x \in q'(D^*)$. First observe that $\bar x \in q(D^*)$. By contradiction, assume otherwise. Then, there is a CQ $\hat{q}(\bar x)$ in 
$q'(\bar x)$, other than $q$, such that $\bar x \in \hat{q}(D^*)$. Clearly, $D^*$ homomorphically maps to $p'$ via the mapping $h_0$, 
and hence, there is a homomorphism from $\hat{q}$ to $p'$ 
mapping $\bar x$ to $\bar x$. By Lemma~\ref{le:main}, $\bar x \in \hat{q}(\chase{p}{\Sigma})$, which in turn implies that $p \subseteq_\Sigma \hat{q}$; the latter holds due to Proposition~\ref{pro:cqs-cont-chase}.
%
Hence, $q \subseteq_\Sigma \hat{q}$ as $p \equiv_\Sigma q$, which contradicts the way in which $q$ has been chosen. 

Now, since $\bar x \in q(D^*)$, we have that $\bar x \in p(D^*)$ because $q \equiv_\Sigma p$ and $D^* \models \Sigma$. Then, there is a homomorphism $h$ from $p$ to $D^*$ with $h(\bar x) = \bar x$. 
It follows that $h_0 \circ h$ is a homomorphism from $p$ to $p'$ that maps $\bar x$ to $\bar x$. In consequence, from Lemma 
\ref{le:main} we obtain that $h(X) = X$. 
Therefore, 
there must exist some $m \geq 0$ such that $g=h \circ (h_0 \circ h)^m$ is a homomorphism from $p$ to $D^*$ that satisfies that $h_0 \circ g$ is the identity on $X$. 
It follows from Lemma 
\ref{lemma:grohe-correct} that G has a $k$-clique.
\end{itemize} 
%

As for item (3), first notice that the CQS $S = (\Sigma,q')$ can be computed by simply enumerating the CQSs from $\class{O}$ until we find $S$ since, by Theorem~\ref{the:ucq-k-equiv-complexity-cqs}, we can check whether $S \not\in (\class{FG},\class{UCQ})_{F(k)}^{\equiv}$.
%
The same holds for $q$. From $q$ we can construct the CQs $p$ and $p'$, as well as the set of variables $X$, by applying Lemma~\ref{le:main}. We can then compute 
$\mu$ via an exhaustive search over $G^{p'}$. Notice that the construction of $q'$, $q$, $p$, $p'$, $X$, and $\mu$ depends only on $k$. 
Lemma~\ref{lemma:grohe-correct} states, on the other hand, that 
it is possible to construct $D^*$ in time polynomial, given $p'$, $p$, $X$, $\mu$, and $G$. Putting all these together, we obtain the existence of computable functions $f,g : \mathbb{N} \to \mathbb{N}$ as required in item (3), and the claim follows.
\end{proof}





\section{Proof of Lemma~\ref{lem:fc-vs-sfc}}

%

For brevity, let $\mathsf{finans}(q,D,\dep) = \bigcap_{M \in \fmods{D}{\dep}} q(M)$.

\medskip

$(\Rightarrow)$ Consider an $\ins{S}$-database $D$, a set $\dep \in \class{C}$ of TGDs over $\ins{S}$, and an integer $n \geq 0$. Assuming that $\adom{D} = \{d_1,\ldots,d_\ell\}$, for $\ell \geq 1$, let $D^+ = D \cup \{\text{\rm Dom}(d_1),\ldots,\text{\rm Dom}(d_\ell)\}$, and $\bar d = (d_1,\ldots,d_\ell)$.

\begin{lemma}\label{lem:sfc-auxiliary-lemma}
	There exists an instance $M^* \in \fmods{D^+}{\dep}$ such that, for every UCQ $q^*$ over $\ins{S} \cup \{\text{\rm Dom}\}$ of arity $\ell$ with at most $n+\ell$ variables, $\bar d \in q^*(M^*)$ implies $\bar d \in q^*(\chase{D^+}{\dep})$
\end{lemma}

\begin{proof}
	Let $q$ be the UCQ of arity $\ell$ consisting of all the CQs $q'$ with at most $n+\ell$ variables such that $\bar d \not\in q'(\chase{D^+}{\dep})$. Clearly, $\bar d \not\in q(\chase{D^+}{\dep})$. Since $\class{C}$ is finitely controllable, we conclude that $\bar d \not\in \mathsf{finans}(q,D^+,\dep)$. Thus, there exists an instance $M^* \in \fmods{D^+}{\dep}$ such that $\bar d \not\in q(M^*)$. Fix an arbitrary UCQ $q^*$ of arity $\ell$ with at most $n+\ell$ variables. Observe that $\bar d \not\in q^*(\chase{D^+}{\dep})$ implies $\bar d \not\in q^*(M^*)$ since, by construction, each CQ in $q^*$ occurs also in $q$, and we know that $\bar d \not\in q(M^*)$. The claim follows.
\end{proof}

We claim that the desired finite model $M(D,\dep,n)$ is the instance $M^{*}_{|\ins{S}}$, where $M^*$ is the $(\ins{S} \cup \{\text{\rm Dom}\})$-instance of $\fmods{D^+}{\dep}$ provided by Lemma~\ref{lem:sfc-auxiliary-lemma}. To this end, we need to show that, for every UCQ $q(\bar x)$ over $\ins{S}$ of arity $r \geq 0$ with at most $n$ variables, $q(\chase{D}{\dep}) = q(M^{*}_{|\ins{S}})$.
Fix an arbitrary tuple $\bar c \in \adom{D}^r$. Clearly, $\bar c \in q(\chase{D}{\dep})$ implies $\bar c \in q(M^{*}_{|\ins{S}} )$ since $M^{*}_{|\ins{S}}$ is a model of $D$ and $\dep$.
Assume now that $\bar c \in q(M^{*}_{|\ins{S}})$. We proceed to show that $\bar c \in q(\chase{D}{\dep})$.
It is not difficult to see that there exists a UCQ $\hat{q}$ of arity $\ell$ over $(\ins{S} \cup \{\text{\rm Dom}\})$ with at most $n+\ell$ variables such that:
\begin{eqnarray}
\bar c \in q(M^{*}_{|\ins{S}}) &\Longrightarrow& \bar d \in \hat{q}(M^*)\label{impl-1}\\
\bar d \in \hat{q}(\chase{D^+}{\dep}) &\Longrightarrow& \bar c \in q(\chase{D}{\dep})\label{impl-2}
\end{eqnarray}
Since $\hat{q}$ has at most $n+\ell$ variables, by Lemma~\ref{lem:sfc-auxiliary-lemma}, we get that
\begin{eqnarray}
\bar d \in \hat{q}(M^*)\ \Longrightarrow\ \bar d \in q(\chase{D^+}{\dep})\label{impl-3}
\end{eqnarray}
Since, by hypothesis, $\bar c \in q(M^{*}_{|\ins{S}})$, from (\ref{impl-1}) we get that $\bar d \in \hat{q}(M^*)$. Thus, by (\ref{impl-3}), we get that $\bar d \in q(\chase{D^+}{\dep})$, and hence, by (\ref{impl-2}), we conclude that $\bar c \in q(\chase{D}{\dep})$, as needed.

\medskip
$(\Leftarrow)$ Consider an $\ins{S}$-database $D$, a set $\dep \in \class{C}$ of TGDs over $\ins{S}$, and a UCQ $q(\bar x)$ over $\ins{S}$. Fix an arbitrary tuple $\bar c \in \adom{D}^{|\bar x|}$. We need to show that $\bar c \in q(\chase{D}{\dep})$ iff $\bar c \in \mathsf{finans}(q,D,\dep)$.
Clearly, the $(\Rightarrow)$ direction holds trivially. It remains to show that $\bar c \not\in q(\chase{D}{\dep})$ implies $\bar c \not\in \mathsf{finans}(q,D,\dep)$. Let $n \geq 0$ be the number of variables occurring in $q$. Since $\class{C}$ is strongly finitely controllable, there exists $M(D,\dep,n) \in \fmods{D}{\dep}$ such that $\bar c \not\in q(M(D,\dep,n))$. This immediately implies that $\bar c \not\in \mathsf{finans}(q,D,\dep)$.

\section{Proof of Theorem~\ref{the:fg-sfc}}
%

To show that $\class{FG}$ is strongly finitely controllable, by Lemma~\ref{lem:fc-vs-sfc}, it suffices to show that it is finitely controllable. This can be easily shown by exploiting the fact that the {\em guarded negation fragment of first-order logic} (GNFO)~\cite{BaCS15} enjoys the finite model property, i.e., if a GNFO sentence has a model, then it has a finite one. But let us first recall the guarded negation fragment.

GNFO restricts first-order logic by requiring that all occurrences of negation are of the form $\alpha \wedge \neg \varphi$, where $\alpha$ is an atom containing all the free variables of $\varphi$.
Formally, the formulas of GNFO are generated by the recursive
definition
\[
\varphi\ ::=\ R(t_1,\ldots,t_n)~|~t_1=t_2~|~\varphi_1 \wedge
\varphi_2~|~\varphi_1 \vee \varphi_2~|~\exists x \,
\varphi~|~\alpha \wedge \neg \varphi,
\]
where each $t_i$ is a term (constant or variable), and in the last clause, $\alpha$ is an atomic formula containing all free variables of $\varphi$.
We know that GNFO enjoys the finite model property, and we also have an upper bound on the size of finite models:

\begin{proposition}[\cite{BaCS15}]\label{pro:gnfo-fmp}
	Consider a GNFO sentence $\varphi$. If $\varphi$ has a model, then it has a finite one of size $2^{2^{\size{\varphi}^{O(1)}}}$.
\end{proposition}

We proceed to show that $\class{FG}$ is finitely controllable. Consider a database $D$, a set $\dep \in \class{FG}$, and a UCQ $q(\bar x)$. We need to show that
\[
q(\chase{D}{\dep})\ =\ \bigcap_{M \in \fmods{D}{\dep}} q(M).
\]
It should be clear that the $(\subseteq)$ direction holds trivially. It remains to show the $(\supseteq)$ direction.
Consider a tuple $\bar c \in \adom{D}^{|\bar x|}$. Our goal is to devise a GNFO sentence $\Phi$ such that
\begin{enumerate}
		\item $\bar c \in \bigcap_{M \in \fmods{D}{\dep}} q(M)$\ $\Longrightarrow$\ $\Phi$ does not have a {\em finite} model.
		
		\item $\Phi$ does not have a model\ $\Longrightarrow$\ $\bar c \in q(\chase{D}{\dep})$.
\end{enumerate}
Having such a sentence in place we get the $(\supseteq)$ direction since, by Proposition~\ref{pro:gnfo-fmp}, we can conclude that if $\Phi$ does not have a finite model, then it does not have a model at all, which in turn implies that $\bar c \in q(\chase{D}{\dep})$, as needed.
Let us now explain how $\Phi$ is constructed. We first observe that a frontier-guarded TGD $\sigma$ of the form $\varphi(\bar x,\bar y) \ra \exists \bar x \, \psi(\bar x,\bar z)$ can be equivalently rewritten as
\[
	\phi_\sigma\ =\ \neg\bigg(\exists \bar x \exists \bar y \big(\varphi(\bar x,\bar y) \wedge \neg \exists \bar z \, \psi(\bar x,\bar z)\big)\bigg),
\]
which is a GNFO sentence since all the free variables of $\exists \bar z \, \psi(\bar x,\bar z)$ occur in the guard of $\sigma$.
We also observe that the sentence $\neg q(\bar c)$ is trivially a GNFO sentence. Therefore,
\[
	\Phi\ =\ D\ \wedge\ \bigwedge_{\sigma \in \dep} \phi_\sigma\ \wedge\ \neg q(\bar c),
\]
which is clearly the desired GNFO sentence.

It remains to show that a finite witness is realizable, i.e., there is a computable function that takes as input an $\ins{S}$-database $D$, a set $\dep \in \class{FG}$ over $\ins{S}$, and an integer $n \geq 0$, and outputs the finite model $M(D,\dep,n)$. In fact, this follows from the proof of Lemma~\ref{lem:fc-vs-sfc} and Proposition~\ref{pro:gnfo-fmp}.
In the proof of Lemma~\ref{lem:fc-vs-sfc}, we actually show that $M(D,\dep,n)$ is the model of a GNFO sentence $\varphi$, and thus, due to Proposition~\ref{pro:gnfo-fmp}, we can assume that is of size $2^{2^{\size{\varphi}^{O(1)}}}$.
Therefore, we can compute the finite model $M(D,\dep,n)$ as follows: 
\begin{enumerate}
	\item we first compute the set $S_n$ of all the non-isomorphic UCQs over $\ins{S}$ with at most $n$ variables, which is clearly finite; and
	\item we enumerate all the non-isomorphic finite models $M$ of $\varphi$ of size $2^{2^{\size{\varphi}^{O(1)}}}$, which are finitely many, until we find one such that, for every $q \in S_n$, $q(\chase{D}{\dep}) = q(M)$. Note that the latter equality can be effectively checked since $\dep \in \class{FG}$. 
\end{enumerate}


\begin{thebibliography}{33}
	
	
	\ifx \showCODEN    \undefined \def \showCODEN     #1{\unskip}     \fi
	\ifx \showDOI      \undefined \def \showDOI       #1{#1}\fi
	\ifx \showISBNx    \undefined \def \showISBNx     #1{\unskip}     \fi
	\ifx \showISBNxiii \undefined \def \showISBNxiii  #1{\unskip}     \fi
	\ifx \showISSN     \undefined \def \showISSN      #1{\unskip}     \fi
	\ifx \showLCCN     \undefined \def \showLCCN      #1{\unskip}     \fi
	\ifx \shownote     \undefined \def \shownote      #1{#1}          \fi
	\ifx \showarticletitle \undefined \def \showarticletitle #1{#1}   \fi
	\ifx \showURL      \undefined \def \showURL       {\relax}        \fi
	\providecommand\bibfield[2]{#2}
	\providecommand\bibinfo[2]{#2}
	\providecommand\natexlab[1]{#1}
	\providecommand\showeprint[2][]{arXiv:#2}
	
	\bibitem[\protect\citeauthoryear{Abiteboul, Hull, and Vianu}{Abiteboul
		et~al\mbox{.}}{1995}]%
	{AbHV95}
	\bibfield{author}{\bibinfo{person}{Serge Abiteboul}, \bibinfo{person}{Richard
			Hull}, {and} \bibinfo{person}{Victor Vianu}.}
	\bibinfo{year}{1995}\natexlab{}.
	\newblock \bibinfo{booktitle}{\emph{Foundations of Databases}}.
	\newblock \bibinfo{publisher}{Addison-Wesley}.
	\newblock
	
	
	\bibitem[\protect\citeauthoryear{Baget, Mugnier, Rudolph, and Thomazo}{Baget
		et~al\mbox{.}}{2011b}]%
	{BMRT11}
	\bibfield{author}{\bibinfo{person}{Jean{-}Fran{\c{c}}ois Baget},
		\bibinfo{person}{Marie{-}Laure Mugnier}, \bibinfo{person}{Sebastian Rudolph},
		{and} \bibinfo{person}{Micha{\"{e}}l Thomazo}.}
	\bibinfo{year}{2011}\natexlab{b}.
	\newblock \showarticletitle{Walking the Complexity Lines for Generalized
		Guarded Existential Rules}. In \bibinfo{booktitle}{\emph{IJCAI}}.
	\bibinfo{pages}{712--717}.
	\newblock
	
	
	\bibitem[\protect\citeauthoryear{Baget, Lecl{\`e}re, Mugnier, and Salvat}{Baget
		et~al\mbox{.}}{2011a}]%
	{BLMS11}
	\bibfield{author}{\bibinfo{person}{Jean-Fran\c{c}ois Baget},
		\bibinfo{person}{Michel Lecl{\`e}re}, \bibinfo{person}{Marie-Laure Mugnier},
		{and} \bibinfo{person}{Eric Salvat}.} \bibinfo{year}{2011}\natexlab{a}.
	\newblock \showarticletitle{On rules with existential variables: {W}alking the
		decidability line}.
	\newblock \bibinfo{journal}{\emph{Artif. Intell.}} \bibinfo{volume}{175},
	\bibinfo{number}{9-10} (\bibinfo{year}{2011}), \bibinfo{pages}{1620--1654}.
	\newblock
	
	
	\bibitem[\protect\citeauthoryear{B{\'{a}}r{\'{a}}ny, Gottlob, and
		Otto}{B{\'{a}}r{\'{a}}ny et~al\mbox{.}}{2014}]%
	{BaGO14}
	\bibfield{author}{\bibinfo{person}{Vince B{\'{a}}r{\'{a}}ny},
		\bibinfo{person}{Georg Gottlob}, {and} \bibinfo{person}{Martin Otto}.}
	\bibinfo{year}{2014}\natexlab{}.
	\newblock \showarticletitle{Querying the Guarded Fragment}.
	\newblock \bibinfo{journal}{\emph{Logical Methods in Computer Science}}
	\bibinfo{volume}{10}, \bibinfo{number}{2} (\bibinfo{year}{2014}).
	\newblock
	
	
	\bibitem[\protect\citeauthoryear{B{\'{a}}r{\'{a}}ny, ten Cate, and
		Segoufin}{B{\'{a}}r{\'{a}}ny et~al\mbox{.}}{2015}]%
	{BaCS15}
	\bibfield{author}{\bibinfo{person}{Vince B{\'{a}}r{\'{a}}ny},
		\bibinfo{person}{Balder ten Cate}, {and} \bibinfo{person}{Luc Segoufin}.}
	\bibinfo{year}{2015}\natexlab{}.
	\newblock \showarticletitle{Guarded Negation}.
	\newblock \bibinfo{journal}{\emph{J. {ACM}}} \bibinfo{volume}{62},
	\bibinfo{number}{3} (\bibinfo{year}{2015}), \bibinfo{pages}{22:1--22:26}.
	\newblock
	
	
	\bibitem[\protect\citeauthoryear{Barcel{\'{o}}, Berger, and
		Pieris}{Barcel{\'{o}} et~al\mbox{.}}{2018}]%
	{BaBP18}
	\bibfield{author}{\bibinfo{person}{Pablo Barcel{\'{o}}},
		\bibinfo{person}{Gerald Berger}, {and} \bibinfo{person}{Andreas Pieris}.}
	\bibinfo{year}{2018}\natexlab{}.
	\newblock \showarticletitle{Containment for Rule-Based Ontology-Mediated
		Queries}. In \bibinfo{booktitle}{\emph{PODS}}. \bibinfo{pages}{267--279}.
	\newblock
	
	
	\bibitem[\protect\citeauthoryear{Barcel{\'{o}}, Feier, Lutz, and
		Pieris}{Barcel{\'{o}} et~al\mbox{.}}{2019a}]%
	{BFLP19}
	\bibfield{author}{\bibinfo{person}{Pablo Barcel{\'{o}}},
		\bibinfo{person}{Cristina Feier}, \bibinfo{person}{Carsten Lutz}, {and}
		\bibinfo{person}{Andreas Pieris}.} \bibinfo{year}{2019}\natexlab{a}.
	\newblock \showarticletitle{When is Ontology-Mediated Querying Efficient?}. In
	\bibinfo{booktitle}{\emph{LICS}}. \bibinfo{pages}{1--13}.
	\newblock
	
	
	\bibitem[\protect\citeauthoryear{Barcel{\'{o}}, Figueira, Gottlob, and
		Pieris}{Barcel{\'{o}} et~al\mbox{.}}{2019b}]%
	{BFGP19}
	\bibfield{author}{\bibinfo{person}{Pablo Barcel{\'{o}}}, \bibinfo{person}{Diego
			Figueira}, \bibinfo{person}{Georg Gottlob}, {and} \bibinfo{person}{Andreas
			Pieris}.} \bibinfo{year}{2019}\natexlab{b}.
	\newblock \showarticletitle{Semantic Optimization of Conjunctive Queries}.
	\newblock
	\newblock
	\shownote{Under submission, available at
		http://homepages.inf.ed.ac.uk/apieris/BFGP.pdf.}
	
	
	\bibitem[\protect\citeauthoryear{Barcel{\'{o}}, Gottlob, and
		Pieris}{Barcel{\'{o}} et~al\mbox{.}}{2016}]%
	{BaGP16}
	\bibfield{author}{\bibinfo{person}{Pablo Barcel{\'{o}}}, \bibinfo{person}{Georg
			Gottlob}, {and} \bibinfo{person}{Andreas Pieris}.}
	\bibinfo{year}{2016}\natexlab{}.
	\newblock \showarticletitle{Semantic Acyclicity Under Constraints}. In
	\bibinfo{booktitle}{\emph{PODS}}. \bibinfo{pages}{343--354}.
	\newblock
	
	
	\bibitem[\protect\citeauthoryear{Bienvenu, Hansen, Lutz, and Wolter}{Bienvenu
		et~al\mbox{.}}{2016}]%
	{BHLW16}
	\bibfield{author}{\bibinfo{person}{Meghyn Bienvenu}, \bibinfo{person}{Peter
			Hansen}, \bibinfo{person}{Carsten Lutz}, {and} \bibinfo{person}{Frank
			Wolter}.} \bibinfo{year}{2016}\natexlab{}.
	\newblock \showarticletitle{First Order-Rewritability and Containment of
		Conjunctive Queries in Horn Description Logics}. In
	\bibinfo{booktitle}{\emph{IJCAI}}. \bibinfo{pages}{965--971}.
	\newblock
	
	
	\bibitem[\protect\citeauthoryear{Bienvenu and Ortiz}{Bienvenu and
		Ortiz}{2015}]%
	{BiOr15}
	\bibfield{author}{\bibinfo{person}{Meghyn Bienvenu} {and}
		\bibinfo{person}{Magdalena Ortiz}.} \bibinfo{year}{2015}\natexlab{}.
	\newblock \showarticletitle{Ontology-Mediated Query Answering with
		Data-Tractable Description Logics}. In \bibinfo{booktitle}{\emph{Reasoning
			Web}}. \bibinfo{pages}{218--307}.
	\newblock
	
	
	\bibitem[\protect\citeauthoryear{Bienvenu, ten Cate, Lutz, and Wolter}{Bienvenu
		et~al\mbox{.}}{2014}]%
	{BCLW14}
	\bibfield{author}{\bibinfo{person}{Meghyn Bienvenu}, \bibinfo{person}{Balder
			ten Cate}, \bibinfo{person}{Carsten Lutz}, {and} \bibinfo{person}{Frank
			Wolter}.} \bibinfo{year}{2014}\natexlab{}.
	\newblock \showarticletitle{Ontology-Based Data Access: {A} Study through
		Disjunctive Datalog, CSP, and {MMSNP}}.
	\newblock \bibinfo{journal}{\emph{{ACM} Trans. Database Syst.}}
	\bibinfo{volume}{39}, \bibinfo{number}{4} (\bibinfo{year}{2014}),
	\bibinfo{pages}{33:1--33:44}.
	\newblock
	
	
	\bibitem[\protect\citeauthoryear{Blumensath}{Blumensath}{2010}]%
	{Blum10}
	\bibfield{author}{\bibinfo{person}{Achim Blumensath}.}
	\bibinfo{year}{2010}\natexlab{}.
	\newblock \showarticletitle{Guarded Second-Order Logic, Spanning Trees, and
		Network Flows}.
	\newblock \bibinfo{journal}{\emph{Logical Methods in Computer Science}}
	\bibinfo{volume}{6}, \bibinfo{number}{1} (\bibinfo{year}{2010}).
	\newblock
	
	
	\bibitem[\protect\citeauthoryear{Cal\`{\i}, Gottlob, and Kifer}{Cal\`{\i}
		et~al\mbox{.}}{2013}]%
	{CaGK13}
	\bibfield{author}{\bibinfo{person}{Andrea Cal\`{\i}}, \bibinfo{person}{Georg
			Gottlob}, {and} \bibinfo{person}{Michael Kifer}.}
	\bibinfo{year}{2013}\natexlab{}.
	\newblock \showarticletitle{Taming the Infinite Chase: Query Answering under
		Expressive Relational Constraints}.
	\newblock \bibinfo{journal}{\emph{J. Artif. Intell. Res.}}
	\bibinfo{volume}{48} (\bibinfo{year}{2013}), \bibinfo{pages}{115--174}.
	\newblock
	
	
	\bibitem[\protect\citeauthoryear{Cal\`{\i}, Gottlob, and Lukasiewicz}{Cal\`{\i}
		et~al\mbox{.}}{2012a}]%
	{CaGL12}
	\bibfield{author}{\bibinfo{person}{Andrea Cal\`{\i}}, \bibinfo{person}{Georg
			Gottlob}, {and} \bibinfo{person}{Thomas Lukasiewicz}.}
	\bibinfo{year}{2012}\natexlab{a}.
	\newblock \showarticletitle{A general {D}atalog-based framework for tractable
		query answering over ontologies}.
	\newblock \bibinfo{journal}{\emph{J. Web Sem.}}  \bibinfo{volume}{14}
	(\bibinfo{year}{2012}), \bibinfo{pages}{57--83}.
	\newblock
	
	
	\bibitem[\protect\citeauthoryear{Cal\`{\i}, Gottlob, and Pieris}{Cal\`{\i}
		et~al\mbox{.}}{2012b}]%
	{CaGP12}
	\bibfield{author}{\bibinfo{person}{Andrea Cal\`{\i}}, \bibinfo{person}{Georg
			Gottlob}, {and} \bibinfo{person}{Andreas Pieris}.}
	\bibinfo{year}{2012}\natexlab{b}.
	\newblock \showarticletitle{Towards more expressive ontology languages: {T}he
		query answering problem}.
	\newblock \bibinfo{journal}{\emph{Artif. Intell.}}  \bibinfo{volume}{193}
	(\bibinfo{year}{2012}), \bibinfo{pages}{87--128}.
	\newblock
	
	
	\bibitem[\protect\citeauthoryear{Chandra and Merlin}{Chandra and
		Merlin}{1977}]%
	{ChMe77}
	\bibfield{author}{\bibinfo{person}{Ashok~K. Chandra} {and}
		\bibinfo{person}{Philip~M. Merlin}.} \bibinfo{year}{1977}\natexlab{}.
	\newblock \showarticletitle{Optimal Implementation of Conjunctive Queries in
		Relational Data Bases}. In \bibinfo{booktitle}{\emph{STOC}}.
	\bibinfo{pages}{77--90}.
	\newblock
	
	
	\bibitem[\protect\citeauthoryear{Chekuri and Rajaraman}{Chekuri and
		Rajaraman}{2000}]%
	{ChRa00}
	\bibfield{author}{\bibinfo{person}{Chandra Chekuri} {and}
		\bibinfo{person}{Anand Rajaraman}.} \bibinfo{year}{2000}\natexlab{}.
	\newblock \showarticletitle{Conjunctive query containment revisited}.
	\newblock \bibinfo{journal}{\emph{Theor. Comput. Sci.}} \bibinfo{volume}{239},
	\bibinfo{number}{2} (\bibinfo{year}{2000}), \bibinfo{pages}{211--229}.
	\newblock
	
	
	\bibitem[\protect\citeauthoryear{Courcelle}{Courcelle}{2003}]%
	{Cour03}
	\bibfield{author}{\bibinfo{person}{Bruno Courcelle}.}
	\bibinfo{year}{2003}\natexlab{}.
	\newblock \showarticletitle{The monadic second-order logic of graphs {XIV:}
		uniformly sparse graphs and edge set quantifications}.
	\newblock \bibinfo{journal}{\emph{Theor. Comput. Sci.}} \bibinfo{volume}{299},
	\bibinfo{number}{1-3} (\bibinfo{year}{2003}), \bibinfo{pages}{1--36}.
	\newblock
	
	
	\bibitem[\protect\citeauthoryear{Dalmau, Kolaitis, and Vardi}{Dalmau
		et~al\mbox{.}}{2002}]%
	{DaKV02}
	\bibfield{author}{\bibinfo{person}{V{\'{\i}}ctor Dalmau},
		\bibinfo{person}{Phokion~G. Kolaitis}, {and} \bibinfo{person}{Moshe~Y.
			Vardi}.} \bibinfo{year}{2002}\natexlab{}.
	\newblock \showarticletitle{Constraint Satisfaction, Bounded Treewidth, and
		Finite-Variable Logics}. In \bibinfo{booktitle}{\emph{CP}}.
	\bibinfo{pages}{310--326}.
	\newblock
	
	
	\bibitem[\protect\citeauthoryear{Downey and Fellows}{Downey and
		Fellows}{1995}]%
	{DoFe95}
	\bibfield{author}{\bibinfo{person}{Rodney~G. Downey} {and}
		\bibinfo{person}{Michael~R. Fellows}.} \bibinfo{year}{1995}\natexlab{}.
	\newblock \showarticletitle{Fixed-Parameter Tractability and Completeness {II:}
		On Completeness for {W[1]}}.
	\newblock \bibinfo{journal}{\emph{Theor. Comput. Sci.}} \bibinfo{volume}{141},
	\bibinfo{number}{1{\&}2} (\bibinfo{year}{1995}), \bibinfo{pages}{109--131}.
	\newblock
	
	
	\bibitem[\protect\citeauthoryear{Fagin, Kolaitis, Miller, and Popa}{Fagin
		et~al\mbox{.}}{2005}]%
	{FKMP05}
	\bibfield{author}{\bibinfo{person}{Ronald Fagin}, \bibinfo{person}{Phokion~G.
			Kolaitis}, \bibinfo{person}{Ren{\'{e}}e~J. Miller}, {and}
		\bibinfo{person}{Lucian Popa}.} \bibinfo{year}{2005}\natexlab{}.
	\newblock \showarticletitle{Data exchange: semantics and query answering}.
	\newblock \bibinfo{journal}{\emph{Theor. Comput. Sci.}} \bibinfo{volume}{336},
	\bibinfo{number}{1} (\bibinfo{year}{2005}), \bibinfo{pages}{89--124}.
	\newblock
	
	
	\bibitem[\protect\citeauthoryear{Gottlob, Manna, and Pieris}{Gottlob
		et~al\mbox{.}}{2014a}]%
	{GoMP14}
	\bibfield{author}{\bibinfo{person}{Georg Gottlob}, \bibinfo{person}{Marco
			Manna}, {and} \bibinfo{person}{Andreas Pieris}.}
	\bibinfo{year}{2014}\natexlab{a}.
	\newblock \showarticletitle{Polynomial Combined Rewritings for Existential
		Rules}. In \bibinfo{booktitle}{\emph{KR}}.
	\newblock
	
	
	\bibitem[\protect\citeauthoryear{Gottlob, Rudolph, and Simkus}{Gottlob
		et~al\mbox{.}}{2014b}]%
	{GoRS14}
	\bibfield{author}{\bibinfo{person}{Georg Gottlob}, \bibinfo{person}{Sebastian
			Rudolph}, {and} \bibinfo{person}{Mantas Simkus}.}
	\bibinfo{year}{2014}\natexlab{b}.
	\newblock \showarticletitle{Expressiveness of guarded existential rule
		languages}. In \bibinfo{booktitle}{\emph{PODS}}. \bibinfo{pages}{27--38}.
	\newblock
	
	
	\bibitem[\protect\citeauthoryear{Gr{\"{a}}del, Hirsch, and Otto}{Gr{\"{a}}del
		et~al\mbox{.}}{2002}]%
	{GrHO02}
	\bibfield{author}{\bibinfo{person}{Erich Gr{\"{a}}del}, \bibinfo{person}{Colin
			Hirsch}, {and} \bibinfo{person}{Martin Otto}.}
	\bibinfo{year}{2002}\natexlab{}.
	\newblock \showarticletitle{Back and forth between guarded and modal logics}.
	\newblock \bibinfo{journal}{\emph{{ACM} Trans. Comput. Log.}}
	\bibinfo{volume}{3}, \bibinfo{number}{3} (\bibinfo{year}{2002}),
	\bibinfo{pages}{418--463}.
	\newblock
	
	
	\bibitem[\protect\citeauthoryear{Grohe}{Grohe}{2007}]%
	{Grohe07}
	\bibfield{author}{\bibinfo{person}{Martin Grohe}.}
	\bibinfo{year}{2007}\natexlab{}.
	\newblock \showarticletitle{The complexity of homomorphism and constraint
		satisfaction problems seen from the other side}.
	\newblock \bibinfo{journal}{\emph{J. {ACM}}} \bibinfo{volume}{54},
	\bibinfo{number}{1} (\bibinfo{year}{2007}), \bibinfo{pages}{1:1--1:24}.
	\newblock
	
	
	\bibitem[\protect\citeauthoryear{Johnson and Klug}{Johnson and Klug}{1984}]%
	{JoKl84}
	\bibfield{author}{\bibinfo{person}{David~S. Johnson} {and}
		\bibinfo{person}{Anthony~C. Klug}.} \bibinfo{year}{1984}\natexlab{}.
	\newblock \showarticletitle{Testing Containment of Conjunctive Queries under
		Functional and Inclusion Dependencies}.
	\newblock \bibinfo{journal}{\emph{J. Comput. Syst. Sci.}} \bibinfo{volume}{28},
	\bibinfo{number}{1} (\bibinfo{year}{1984}), \bibinfo{pages}{167--189}.
	\newblock
	
	
	\bibitem[\protect\citeauthoryear{Leone, Manna, Terracina, and Veltri}{Leone
		et~al\mbox{.}}{2019}]%
	{LMTV19}
	\bibfield{author}{\bibinfo{person}{Nicola Leone}, \bibinfo{person}{Marco
			Manna}, \bibinfo{person}{Giorgio Terracina}, {and}
		\bibinfo{person}{Pierfrancesco Veltri}.} \bibinfo{year}{2019}\natexlab{}.
	\newblock \showarticletitle{Fast Query Answering over Existential Rules}.
	\newblock \bibinfo{journal}{\emph{{ACM} Trans. Comput. Log.}}
	\bibinfo{volume}{20}, \bibinfo{number}{2} (\bibinfo{year}{2019}),
	\bibinfo{pages}{12:1--12:48}.
	\newblock
	
	
	\bibitem[\protect\citeauthoryear{Maier, Mendelzon, and Sagiv}{Maier
		et~al\mbox{.}}{1979}]%
	{MaMS79}
	\bibfield{author}{\bibinfo{person}{David Maier}, \bibinfo{person}{Alberto~O.
			Mendelzon}, {and} \bibinfo{person}{Yehoshua Sagiv}.}
	\bibinfo{year}{1979}\natexlab{}.
	\newblock \showarticletitle{Testing Implications of Data Dependencies.}
	\newblock \bibinfo{journal}{\emph{ACM Trans. Database Syst.}}
	\bibinfo{volume}{4}, \bibinfo{number}{4} (\bibinfo{year}{1979}),
	\bibinfo{pages}{455--469}.
	\newblock
	
	
	\bibitem[\protect\citeauthoryear{Marx}{Marx}{2010}]%
	{Marx10}
	\bibfield{author}{\bibinfo{person}{D{\'{a}}niel Marx}.}
	\bibinfo{year}{2010}\natexlab{}.
	\newblock \showarticletitle{Tractable hypergraph properties for constraint
		satisfaction and conjunctive queries}. In \bibinfo{booktitle}{\emph{STOC}}.
	\bibinfo{pages}{735--744}.
	\newblock
	
	
	\bibitem[\protect\citeauthoryear{{OWL Working Group}}{{OWL Working
			Group}}{2009}]%
	{owl2}
	\bibfield{author}{\bibinfo{person}{W3C {OWL Working Group}}.}
	\bibinfo{year}{2009}\natexlab{}.
	\newblock \bibinfo{booktitle}{\emph{{OWL~2 Web Ontology Language: Document
				Overview}}}.
	\newblock \bibinfo{publisher}{W3C Recommendation}.
	\newblock
	\newblock
	\shownote{Available at \url{http://www.w3.org/TR/owl2-overview/}.}
	
	
	\bibitem[\protect\citeauthoryear{Papadimitriou and Yannakakis}{Papadimitriou
		and Yannakakis}{1999}]%
	{PaYa99}
	\bibfield{author}{\bibinfo{person}{Christos~H. Papadimitriou} {and}
		\bibinfo{person}{Mihalis Yannakakis}.} \bibinfo{year}{1999}\natexlab{}.
	\newblock \showarticletitle{On the Complexity of Database Queries}.
	\newblock \bibinfo{journal}{\emph{J. Comput. Syst. Sci.}} \bibinfo{volume}{58},
	\bibinfo{number}{3} (\bibinfo{year}{1999}), \bibinfo{pages}{407--427}.
	\newblock
	
	
	\bibitem[\protect\citeauthoryear{Robertson and Seymour}{Robertson and
		Seymour}{1986}]%
	{EGT}
	\bibfield{author}{\bibinfo{person}{Neil Robertson} {and}
		\bibinfo{person}{Paul~D. Seymour}.} \bibinfo{year}{1986}\natexlab{}.
	\newblock \showarticletitle{Graph minors. V. Excluding a planar graph}.
	\newblock \bibinfo{journal}{\emph{J. Comb. Theory, Ser. {B}}}
	\bibinfo{volume}{41}, \bibinfo{number}{1} (\bibinfo{year}{1986}),
	\bibinfo{pages}{92--114}.
	\newblock
	
	
\end{thebibliography}
\end{document}